\documentclass[11pt]{article}

\usepackage[english]{babel}

\usepackage[a4paper,top=2.5cm,bottom=2.5cm,left=3cm,right=3cm,marginparwidth=1.75cm]{geometry}

\usepackage{amsmath}
\usepackage{graphicx}
\usepackage[colorlinks=true, allcolors=blue]{hyperref}
\usepackage[round]{natbib}
\usepackage{algorithm}
\usepackage{algpseudocode}
\usepackage{hyperref}
\usepackage{subcaption}
\usepackage{bm}

\usepackage[svgnames]{xcolor}
\definecolor{myblue}{RGB}{33, 150, 243}
\definecolor{myorange}{RGB}{255, 152, 0}

\newcommand{\la}[1]{\textcolor{magenta}{#1}}

\title{\textbf{Extending TCLUST to higher dimensions}}
\author{Lucía Trapote, Luis Ángel García-Escudero  and Agustín Mayo-Iscar\\
		University of Valladolid \\Department of Statistics and Operations Research, Valladolid, Spain}

\usepackage{amsfonts,amssymb,amsthm,mathrsfs, xcolor,listings}

\usepackage{algorithmicx}
\usepackage{algpseudocode}

\theoremstyle{plain}
\newtheorem{theorem}{Theorem}[section]

\newtheorem{lemma}{Lema}         
\numberwithin{lemma}{section} 
\newtheorem{proposition}[theorem]{Proposition}

\newtheorem{remark}{Remark}

\providecommand{\norm}[1]{\lVert#1\rVert}
\begin{document}
\maketitle

\begin{abstract}
Outliers are known to significantly distort the results of many commonly used clustering methods, often leading to unreliable partitions. To address this issue, several robust clustering approaches have been developed that not only reduce their influence but also facilitate the detection of meaningful outliers. This presentation focuses on robust clustering methods based on trimming, especially TCLUST, which extends the type of trimming used by MCD in one-population problems to the more general case of multiple and unknown clusters. While TCLUST performs well on low-dimensional data, it struggles with high-dimensional datasets due to the complexity of estimating a large number of parameters. The Robust Linear Grouping (RLG) method offers an alternative by assuming clusters lie near lower-dimensional subspaces, thereby combining clustering with dimensionality reduction. However, RLG has limitations when subspaces intersect and assumes overly simplistic isotropic orthogonal errors. A robust clustering method extending TCLUST will be presented, building on the High Dimensional Data Clustering (HDDC) approach by incorporating trimming and eigenvalue constraints. This new approach, called tHHDC, combines TCLUST and RLG, requiring careful modification and integration of both methodologies within that HDDC framework. A study of the theoretical properties of this approach, together with a feasible algorithm for its implementation, will be presented. The interest of the proposed methodology, along with the issue of selecting input parameters, will be illustrated through a simulation study and a real-data example.
\end{abstract}

\section{Introduction}\label{sec_1}

Outliers are a well-known challenge in cluster analysis, often leading to misleading or erroneous conclusions. Specifically, their presence can result in relevant and well-defined clusters being incorrectly merged or, conversely, in the detection of spurious clusters that do not correspond to any meaningful structure in the data. This motivates the interest in the use of robust clustering techniques 
\citep{garcia2010review,banerjee2012robust,ritter2014robust,garcia2016robustness,garcia2024robust}. 
These methods are specifically designed to withstand a certain proportion of outlying observations. 

In some cases, outliers can be interpreted as forming small, isolated clusters themselves, thereby suggesting an increase in the total number of clusters, $G$. However, inflating the number of clusters is not always a viable or effective strategy. It can be computationally infeasible and may undermine the interpretability of the clustering results. Additionally, clustered outliers represent one of the most difficult forms of contamination, 
as they can affect even robust statistical procedures. Consequently, adopting a unified approach 
that simultaneously addresses both clustering and outlier detection is a natural and effective strategy. This unified approach is also motivated by the close connection between cluster analysis and anomaly detection goals, since cluster analysis seeks to identify dense regions or ``crowds'' of data points, and anomaly detection focuses on identifying observations that lie far from these regions. 

Several robust clustering strategies have been proposed in the literature. One approach relies on heavy-tailed mixture models, such as mixtures of $t$-distributions or mixtures of contaminated Gaussian distributions \citep{peel2000robust,punzo2016robust}. Another strategy is noise modeling, where outliers are fitted through uniformly distributed noise components \citep{banfield1993model,coretto2016robust}. A third major direction consists of trimming-based approaches, where outliers are tried to be excluded from the clustering process \citep{cuesta1997trimmed,neykov2007robust,garcia2008general}.

In this paper, we focus on that trimming-based approach to robust clustering and, in particular in the TCLUST method \citep{garcia2008general}. Trimming methods are easy to interpret and statistically appealing: a predefined fraction $\alpha$ of the data is excluded, under the assumption that these observations are the most outlying ones, and standard statistical tools are subsequently applied to the remaining data. The trimming approach adopted is data-driven and is similar to approaches used in robust statistics with high-breakdown-point behavior already applied in regression and multivariate estimation, such as the Least Trimmed Squares (LTS) estimator \citep{rousseeuw1984least} and the Minimum Covariance Determinant (MCD) estimator \citep{rousseeuw1985multivariate}. The computational implementation of these techniques has been facilitated by the development of ``concentration steps'' or ``C-steps'', introduced by \cite{rousseeuw1999fast}.

The TCLUST method, applying trimming as well as appropriate constraints on the eigenvalues of the component scatter matrices, has shown good properties in terms of theoretical behavior, algorithmic performance, and robustness. We will review this TCLUST method in Section \ref{sec_21}. Unfortunately, these good properties become less satisfactory in practical applications when TCLUST is applied to high or even moderately high dimensions (for example, dimensions not smaller than 8 or 10).  This performance degradation of TCLUST as $p$ increases will be illustrated in Sections \ref{sec_4} and \ref{sec_5}. Two primary difficulties arise in high dimensions for TCLUST: first, the problem of initialization, as the algorithm heavily relies on a good random initialization for the concentration steps; and second, the high number of parameters to estimate, which grows rapidly with $p$ due to the consideration of full covariance structures. It is obvious that higher-dimensional problems are increasingly common in modern statistical practice, and therefore it seems reasonable to attempt to extend this good performance of TCLUST to higher dimensions.

Assuming that the observations are clustered around subspaces of lower dimension than the ambient space, it is reasonable to consider simultaneous robust clustering and dimensionality reduction. With this idea in mind, the robust linear clustering (RLG) method introduced in \cite{GarciaEscudero_et_al_2009_RobustLinearClustering} incorporates trimming into a subspace-based clustering approach. This procedure, which we will review in Section \ref{sec_21}, has some limitations due to its simplifying assumptions of isotropic orthogonal errors, not adequately modeling the projections onto the approximating spaces, and troubles with intersecting subspaces.

To extend the applicability of TCLUST to higher dimensions, we will present in Section \ref{sec_3} a compromise between TCLUST and RLG. This compromise builds on the HDDC methodology introduced in \cite{bouveyron2007high} by incorporating trimming and suitable constraints on the different dispersion-related parameters. After presenting the methodology in Section \ref{sec_31}, Section \ref{sec_32} introduces its population formulation, together with theoretical results on existence and consistency. A feasible algorithm for its practical implementation is presented in Section \ref{sec_33}. That algorithm requires certain modifications compared to the algorithms already proposed for RLG, TCLUST and HDDC. The methodology introduced in Section \ref{sec_3} involves the selection of a considerable number of tuning parameters. Section \ref{sec_4} provides guidance to help the reader understand the different roles these tuning parameters play. A procedure will also be introduced to determine, in a data-driven way, the intrinsic dimensions of the approximating subspaces. This procedure requires a modification of the objective function to include a model complexity penalty, similar to that considered in \cite{cerioli2018finding}. The practical interest of the proposal will be supported through simulations in Section \ref{sec_5} and by analyzing a well-known real dataset in Section \ref{sec_6}. Finally, Section \ref{sec_7} summarizes the main results obtained and outlines some open directions for future research.

\section{TCLUST and RLG}\label{sec_2}

\subsection{TCLUST}\label{sec_21}

Although the first robust clustering method considering a data-driven trimmed approach was the \emph{trimmed $k$-means} in \cite{cuesta1997trimmed}. Trimmed $k$-means, being an extension of the classical $k$-means method, inherits its preference for homoscedasticity and spherical clusters. To address this shortcoming, TCLUST was proposed in \cite{garcia2008general} as a more flexible, but still robust, alternative. TCLUST performs maximum likelihood estimation  by assuming that the data arise from $G$ Gaussian distributed components, while incorporating both trimming and regularization on the scatter matrix via eigenvalue ratio constraints.

If we consider a dataset $\lbrace x_1, \ldots, x_n \rbrace$ in  $\mathbb{R}^p$, TCLUST searches for $G$ centers $\mu_1, \dots, \mu_G$, $G$ symmetric positive definite scatter matrices $\Sigma_1, \dots, \Sigma_G$, and $G$ mixing weights $\pi_1, \dots, \pi_G$ satisfying $\sum_{g=1}^G \pi_g = 1$, along with a partition $\{R_0, R_1, \dots, R_G\}$ of the data, where $R_0$ contains $[ n \alpha ]$ observations considered as the potentially most outlying ones (for $\alpha \in [0,1)$ being the trimming level), and the rest of observations with indexes in $R_1$,..., $R_G$ are allocated to $G$ assumed Gaussian clusters. The goal is to maximize the \emph{classification trimmed log-likelihood}:
\begin{equation}\label{tclust}
    \sum_{g=1}^G \sum_{i \in R_g} \log\left[ \pi_g \, \phi(x_i; \mu_g, \Sigma_g) \right],
\end{equation}
where $\phi(x; \mu, \Sigma)$ denotes the density function of a $p$-variate normal distribution with mean $\mu\in \mathbb{R}^p$ and $\Sigma$ being a $p\times p$ scatter matrix. 
Another key feature of TCLUST is the consideration of the \emph{eigenvalues-ratio constraint}, which imposes that the scatter matrices of the clusters remain ``well-conditioned'' by bounding the ratio of the largest to the smallest eigenvalue across all clusters:
\begin{equation}\label{eigenavalus_ratio}
    \frac{\underset{g=1,\ldots,G}{\max}\:\:\underset{l=1,\dots,p}{\max} \lambda_l(\Sigma_g)}{\underset{g=1,\ldots,G}{\min}\:\:\underset{l=1,\dots,p}{\min} \lambda_l(\Sigma_g)} \leq c,
\end{equation}
where $\lambda_l(\Sigma)$ denotes the $l$-th largest eigenvalue of the matrix $\Sigma$, and $c \geq 1$ is a user-defined parameter. The inclusion of this constraint ensures that the maximization problem remains mathematically well-posed, because it avoids the unboundedness of the target function (\ref{tclust})  when a component concentrates around a single point, i.e. $\mu_g=x_i$, and its scatter matrix collapses, i.e. $|\Sigma_j|\downarrow 0$. Moreover, it provides robustness by preventing the algorithm from capturing so-called ``spurious'' components, formed by a few almost collinear observations, being almost degenerate and not very informative clusters  (with $|\Sigma_j| \simeq 0$).  If $c = 1$ and $\pi_1 = \dots = \pi_G$, the TCLUST method reduces to trimmed $K$-means.

In the particular case where the number of clusters is $G = 1$  (so that a single common covariance matrix is required) and the eigenvalue-ratio constraint parameter $c$ is large, the TCLUST approach becomes equivalent to the well-known Minimum Covariance Determinant (MCD) estimator introduced by \cite{rousseeuw1985multivariate}, which seeks the subset size $n-[n\alpha]$ with the smallest determinant of the sample covariance matrix. %The MCD estimator provides high breakdown point robustness and has become a cornerstone of robust multivariate statistics.

A closely related approach to TCLUST that generalizes the idea of trimming within a mixture model framework is based on the maximization of the \emph{trimmed mixture likelihood}. This methodology was initially proposed by \cite{neykov2007robust} and later extended by \cite{garcia2014robust} by incorporating the eigenvalues-ratio constraints. Instead of assigning each observation to a specific cluster during optimization, the trimmed mixture likelihood considers a full mixture density and trims a proportion $\alpha$ of the data points that contribute least to the overall mixture log-likelihood.
Formally, the goal in that trimmed mixture likelihood approach is to find $G$ weight parameters, $\pi_1, \dots, \pi_G$ satisfying $\sum_{g=1}^G\pi_g=1$, $G$ location parameters, $\mu_1, \dots, \mu_G$, and $G$ scatter matrices, $\Sigma_1, \dots, \Sigma_G$, as well as a subset $R \subset \{1, \dots, n\}$ with $\# R = n-[n\alpha]$, that maximize the trimmed log-likelihood:
\begin{equation}\label{mixtTCLUST}
   \sum_{i \in R} \log\left[ \sum_{g=1}^G \pi_g \, \phi(x_i; \mu_g, \Sigma_g) \right]. 
\end{equation}
The constraint (\ref{eigenavalus_ratio}) is also incorporated, again to get mathematically well-defined problems and prevent for detecting spurious solutions.

%In summary, both the MCD estimator (for $G = 1$) and trimmed mixture likelihoods (for $G > 1$) are closely related to TCLUST, and all share the common objective of robustly estimating parameters in the presence of contamination. Their differences lie primarily in how they allocate or weight observations and in the assumptions imposed on the data structure.

The TCLUST method exhibits good properties both in terms of asymptotic behavior and robustness. Under mild assumptions, existence of both empirical and population problem solutions can be guaranteed together with a consistency result \citep{garcia2008general}. 
Regarding influence behavior, it exhibits desirable theoretical properties. Specifically, it has been shown to possess a bounded influence function \citep{ruwet2012}, and good breakdown point properties for ``well-clusterized'' datasets \citep{ruwet2013}. Additionally, an efficient algorithm \citep{fritz2013fast} is also available for its application based on adapted C-steps (kind of trimmed classification/mixture EM steps). Notably, this type of algorithm can be applied with the \texttt{tclust()} function of the \texttt{tclust} package in \texttt{R} \citep{fritz2012} or through its implentation within the \texttt{FSDA} toolbox for \texttt{Matlab} \citep{cerioli2012}.

Higher-dimensional problems are increasingly common in modern statistical practice, particularly in areas such as bioinformatics, image analysis, and finance. Although TCLUST performs well in low-dimensional settings, it faces notable challenges as the number of variables $p$ increases. This will be illustrated in Sections \ref{sec_5} and \ref{sec_6}. Two primary difficulties arise in high dimensions when using TCLUST (hard and mixture versions): first, the problem of initialization, as the algorithm becomes more sensitive to starting values, increasing the risk of convergence to suboptimal solutions; and second, the high number of parameters to estimate, which grows rapidly with $p$ and can lead to computational inefficiency and overfitting, especially when dealing with full covariance structures.

Although the eigenvalue ratio constraint serves to ``regularize'' the target functions (\ref{tclust}) and (\ref{mixtTCLUST}) when $c$ is small and prevents eigenvalues in $\Sigma_g$ from being arbitrarily close to 0, unfortunately, a small $c$, i.e. $c$ close to 1, would also impose highly constrained cluster partitions, with almost spherical, equal-scatter, trimmed $k$-means-type structures, which can be unsatisfactory for many applications.

\subsection{RLG}\label{sec_22}

In high-dimensional settings, it is unrealistic to assume that the data are uniformly dispersed throughout the ambient space. Instead, it is common to assume that the observations lie close to lower-dimensional structures, such as affine subspaces or manifolds, which capture the intrinsic variability of the data. In clustering problems, this translates into the assumption that each cluster is concentrated around its own subspace, potentially with a different orientation and dimensionality. Such an assumption aligns with many practical applications, where different groups or classes of observations are governed by distinct generative mechanisms that vary along only a few relevant directions.

A common strategy for clustering and dimension reduction involves first applying Principal Component Analysis (PCA) to reduce dimensionality, followed by clustering on the transformed data. This sequential or ``tandem'' approach may distort the original cluster structure and lead to suboptimal results, as noted by \citet{Chang1983}. A more effective strategy is to perform clustering and dimension reduction simultaneously, for instance, by assuming that the observations are distributed around a union of $G$ affine subspaces, each associated with a different cluster. %To model this, several probabilistic frameworks have been proposed. One such framework is the mixture of probabilistic principal component analyzers (PPCA), introduced by \citet{TippingBishop1999}, where each cluster is modeled using a local PPCA structure. Another widely used model is the mixture of factor analyzers (MFA), which generalizes the PPCA mixture model by allowing for more flexible covariance structures; see \citet{GhahramaniHinton1997} and \citet{McLachlanPeel2000} for foundational developments in this direction.

With the aim of robustifying the process of detecting clustering around affine subspaces,  the \textit{Robust Linear Grouping}  (RLG) searches for $G$ affine subspaces $B_1,...,B_G$ with intrinsic dimensions $q_1,..., q_G$  and a partition $\{ R_0,R_1,...,R_G\}$ with $\#R_0=[n \alpha]$   minimizing
  $$
  \sum_{g= 1}^G\sum_{i\in R_g}\|x_i-\text{Pr}_{\mathcal{B}_g}(x_i)\|^2,
  $$
where  $\text{Pr}_{\mathcal{B}}(x)$ denotes the orthogonal projection of $x\in\mathbb{R}^p$ onto subspace $\mathcal{B}$. This method (a particular case of it where $q_g=p-1$ for $g=1,...,G$) was introduced in \cite{GarciaEscudero_et_al_2009_RobustLinearClustering}. A particular case of it, as a robustification by trimming of the classical PCA, was already introduced through the LTS-PCA method with $G=1$ \citep{maronna2005principal,croux2017robust}. Moreover, trimmed $k$-means can be also seen as a particular case when $q_1=...=g_G=0$, where affine subspaces are just centroids. The algorithm proposed for its implementation is analogously based on concentration steps. This method can be also applied throughout the \texttt{rlg()} function in the \texttt{tclust} package or by applying the \texttt{fsda} Matlab toolbox. 

As RLG is just based on orthogonal distances to the cluster-specific approximating affine subspaces, then it assumes isotropic orthogonal errors, and modeling certain heteroskedasticity for those errors. Note also that it does not take into account the particular position of the projected observation, and does not take into account that information for labeling outliers. In fact, RLG exhibits troubles with intersecting subspaces as can be seen in Figure \ref{figure1}. It can be observed that one of the clusters ``absorbs'' part of another, due to the fact that observations distant from the main cluster that clearly defines the approximating subspace eventually become incorporated as that subspace extends.

\begin{figure}[htbp]
        \centering
        \includegraphics[clip,trim=0.25cm 1cm 0.5cm 1cm,scale=0.7]{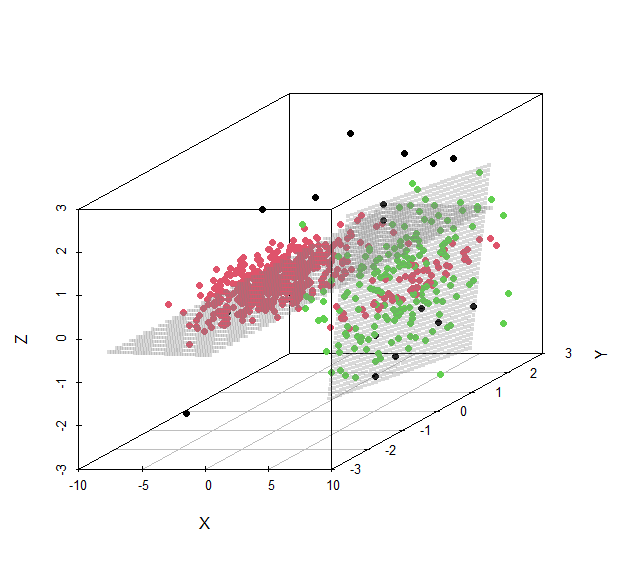}
        \caption{Troubles with intersecting subspaces  when using RLG with $q_1=q_2=2$. The trimmed observations are the black points.}
        \label{figure1}
    \end{figure} 

\section{Trimmed HDDC}\label{sec_3}

\subsection{Methodology}\label{sec_31}

 The robust clustering approaches associated to TCLUST presented in Section \ref{sec_21} require estimating full $\Sigma_g$ covariance matrices when $c>1$, which implies that the number of parameters grows quadratically with the data dimension. However, we can assume that high-dimensional data often concentrate near low-dimensional subspaces, as considered in Section \ref{sec_22}. That assumption is going to be adopted, but following the High-Dimensional Data Clustering (HDDC) approach introduced in \cite{bouveyron2007high}. Therefore, it is assumed that each class-conditional distribution follows a multivariate Gaussian with location vectors $\mu_1,...,\mu_G$ and scatter matrices $\Sigma_1,...,\Sigma_G$. However, a particular parsimonious structure is assumed for the $\Sigma_g$ covariance matrices. To present this parsimonious assumption, let $Q_g$ be the orthogonal matrix whose columns are the eigenvectors of $\Sigma_g$. The covariance matrix $\Sigma_g$ is assumed to satisfy
$\Delta_g = Q_g^T \Sigma_g Q_g$,
where $\Delta_g$ is a diagonal matrix which contains the sorted eigenvalues of $\Sigma_g$. To reduce the number of parameters, it is assumed that $\boldsymbol{\Delta}_g$ has a parsimonious structure
\begin{equation}\label{delta}
 \boldsymbol{\Delta}_g = 
\begin{bmatrix}
\lambda_{g1}   &         & 0       &         &         &         \\
         & \ddots  &         &         &         &         \\
0        &         & \lambda_{gq_g} &         &         &         \\
         &         &         & \lambda_g     &         & 0       \\
         &         &         &         & \ddots  &         \\
         &         &         & 0       &         & \lambda_g     
\end{bmatrix}
\in \mathbb{R}^{p \times p},   
\end{equation}
where $\lambda_{gl} > \lambda_g$ for $l \leq q_g$, and $q_g \in \lbrace 0, \ldots, p - 1 \rbrace$ is an unknown intrinsic dimension specific to the $g$-th group. Note that if $q_g=0$ then no $\lambda_{gl}$ values exist %there no exists $\lambda_{gl}$ values 
and, in that particular case, $\boldsymbol{\Sigma}_g=\boldsymbol{\Delta}_g=\lambda_g \mathbf{I}_p$, where $\mathbf{I}_p$ is the $p\times p$ identity matrix. A class-specific subspace $\mathcal{B}_g$ is defined as the affine subspace spanned by the $q_g$ eigenvectors associated with the $q_g$ largest eigenvalues $\lambda_{gl}$ of $\boldsymbol{\Sigma}_g$, and such that $\mu_g \in \mathcal{B}_g$. The orthogonal complement $\mathcal{B}_g^\bot$ satisfies $\mathbb{R}^p = \mathcal{B}_g \oplus \mathcal{B}_g^\bot$, with $\mu_g \in \mathcal{B}_g^\bot$ as well. In $\mathcal{B}_g^\bot$, the variance is modeled by a single parameter $\lambda_g$. $\mathcal{B}_g$ is referred to as the specific subspace of the $g$-th group and it is expected that  most of the data in that group lie near this subspace. Let $\text{Pr}_{\mathcal{B}_g}(x) = \tilde{Q}_g \tilde{Q}_g^T(x - \mu_g) + \mu_g$ and $\text{Pr}_{\mathcal{B}_g^\bot}(x) = \bar{Q}_g \bar{Q}_g^T(x - \mu_g) + \mu_g$ denote the projections of $x$ onto $\mathcal{B}_g$ and $\mathcal{B}_g^\bot$, respectively. Here, $\tilde{Q}_g \in \mathbb{R}^{p \times p}$ is the matrix formed by the first $q_g$ columns of $Q_g$, corresponding to the directions of greatest variance, and completed with $p - q_g$ zero columns to preserve dimensionality. Similarly, $\bar{Q}_g = Q_g - \tilde{Q}_g$ contains the remaining directions and satisfies $\bar{Q}_g^T \tilde{Q}_g = 0$. 

To enhance robustness to outliers, we also consider the use of trimming within that HDDC framework. Therefore, it is considered that a proportion $\alpha$ of observations are discarded by using trimmed likelihoods when maximizing  trimmed classification likelihoods (\ref{tclust}) or trimmed mixture likelihood (\ref{mixtTCLUST}), with the parsimonious parameterization of the $\Sigma_g$ matrices assumed in the HDDC approach, yielding the $\boldsymbol{\Delta}_g$  matrices shown in (\ref{delta}).

Moreover, two different sets of eigenvalue-ratio constraints are also considered in the maximizations of the trimmed likelihoods (\ref{tclust}) or  (\ref{mixtTCLUST}) by imposing:
    \begin{equation}\label{ER}
  \dfrac{\underset{g=1,\ldots,G}{\max}\:\:\underset{l=1,\ldots,q_g}{\max}\lambda_{gl}}{\underset{g=1,\ldots,G}{\min}\:\:\underset{l=1,\ldots,q_g}{\min}\lambda_{gl}}\leq c_1
\qquad\textup{and}\qquad  
\dfrac{\underset{g=1,\ldots,G}{\max}\lambda_{g}}{\underset{g=1,\ldots,G}{\min}\lambda_{g}}\leq c_2,
    \end{equation}
for two prefixed constants $c_1$ and $c_2$, both being greater than or equal to 1%greater or equal than 1
. We assume that $c_2 < \infty$, while $c_1$ may also be taken as $c_1=\infty$ if desired. Note that setting $c_1 = \infty$ allows the largest eigenvalues to remain unconstrained if desired, whereas $c_2 < \infty$ is always required to ensure the existence of solutions for both the associated empirical and population problems, as well as to guarantee consistency results. These restrictions on the eigenvalues may be seen as an extension of those in (\ref{eigenavalus_ratio}) and lead to mathematically well-defined constrained trimmed likelihood maximization problems (see Section \ref{sec_32}), promoting a higher degree of regularization so as to help avoid the detection of non-interesting spurious solutions.  The choice $c_1=1$ would force the largest eigenvalues to be the same in all the different groups and, analogously, with the smallest eigenvalues when $c_2=1$. However, the main advantage of this approach is the fact that the parameters $c_1$ and $c_2$ allow us to achieve a certain (controlled) freedom in how we want to handle the different scattering of the groups in $\mathcal{B}_g$ and $\mathcal{B}_g^\bot$.

We can summarize all the parameters involved for the constrained maximization of the trimmed likelihoods (\ref{tclust}) and (\ref{mixtTCLUST}) through a set of parameters that we simply denote as $\theta$ where \( \theta = (\pi_1, \ldots, \pi_G, \mu_1, \ldots, \mu_G, \Sigma_1, \ldots, \allowbreak \Sigma_G) \) includes weights \( \pi_g \in [0,1] \) with \( \sum_{g=1}^G \pi_g = 1 \), mean vectors \( \mu_g \in \mathbb{R}^p \), and symmetric positive definite covariance matrices \( \Sigma_g \in \mathbb{R}^{p \times p} \), that do satisfy the parsimonious assumptions for their eigenvalues yielding $\boldsymbol{\Delta}_g$  matrices of eigenvalues fulfilling the pair of eigenvalues-ratio constraints in (\ref{ER}). The set of parameters $\theta$ that obeys these conditions will be denoted as $\Theta_{c_1, c_2}$ throughout the manuscript.

For the proposed algorithms, that will be presented in Section \ref{sec_33}, we will see that their basic C-steps require to handle expressions of the form $\log[\pi_g  \phi(x_i; \mu_g, \Sigma_g)]$. As done in \cite{bouveyron2007high}, this term can be rewritten as 
\begin{align}\label{posterior}
\log[\pi_g \, \phi(x_i; \mu_g, \Sigma_g)] 
&= \log(\pi_g) - \dfrac{1}{2} \Big( 
\| \text{Pr}_{\mathcal{B}_g}(x_i) - \mu_g \|^2_{\mathcal{B}_g} 
+ \dfrac{1}{\lambda_g} \| x_i - \text{Pr}_{\mathcal{B}_g}(x_i) \|^2 \nonumber\\
&\quad + \sum_{l=1}^{q_g} \log(\lambda_{gl}) + (p - q_g)\log(\lambda_g) 
+ p\log(2\pi) \Big),
\end{align}
where $\| \text{Pr}_{\mathcal{B}_g}(x_i) - \mu_g \|^2_{\mathcal{B}_g}$ is the Mahalanobis distance (in a TCLUST-type of distance in $\mathcal{B}_g$)  of the projection of $x_i$ into the approximating subspace with respect to the ``centroid'' used to define that affine subspace. On the other hand, $\| x_i - \text{Pr}_{\mathcal{B}_g}(x_i) \|^2/\lambda_g$ is the orthogonal distance (in a RLG-type of distance), but not assuming isotropic orthogonal errors, because of the $\lambda_g$ terms. This view of $\log[\pi_g \, \phi(x_i; \mu_g, \Sigma_g)]$ is a convenient way to see that this approach serves indeed as a kind of compromise between TCLUST (within the approximating affine space) and RLG (in its orthogonal).

If $\mathcal{B}_g$ corresponds to the affine subspace passing through $\mu_g$ and spanned by the first $q_g$ columns of $Q_g$, which are denoted as $u_{g1},...,u_{gq_g}$, then
\begin{align}\label{bayes_hd_first}
\log[\pi_g \phi(x_i; \mu_g, \Sigma_g)] 
&= \log(\pi_g) - \dfrac{1}{2} \Bigg( 
\sum_{l=1}^{q_g} \dfrac{\langle x_i - \mu_g, u_{gl} \rangle^2}{\lambda_{gl}} \nonumber \\
&\quad + \dfrac{1}{\lambda_g} \left\| x_i - \mu_g - \sum_{l=1}^{q_g} 
\langle x_i - \mu_g, u_{gl} \rangle u_{gl} \right\|^2 \nonumber \\
&\quad + \sum_{l=1}^{q_g} \log(\lambda_{gl}) 
+ (p - q_g) \log(\lambda_g) + p \log(2\pi) \Bigg).
\end{align}
Note that only the first $q_g$ eigenvalues/eigenvectors corresponding to the $q_g$ largest eigenvalues in $\boldsymbol{\Delta}_g$, together with the smallest common eigenvalue $\lambda_g$, are required when computing expression (\ref{bayes_hd_first}). Therefore, computing all eigenvectors is not required. Even though a more reduced notation for the parameters in $\theta$ required to compute $\log[\pi_g \phi(x_i; \mu_g, \Sigma_g)]$ could be adopted, we retain the full parametrization of $\theta$ in $\Theta_{c_1,c_2}$ for clarity of exposition.

A graphical summary of the notation considered can be observed in Figure \ref{esquema}.
\begin{figure}[htbp]
    \includegraphics[width=\linewidth,trim=0 2cm 0 5cm,clip]{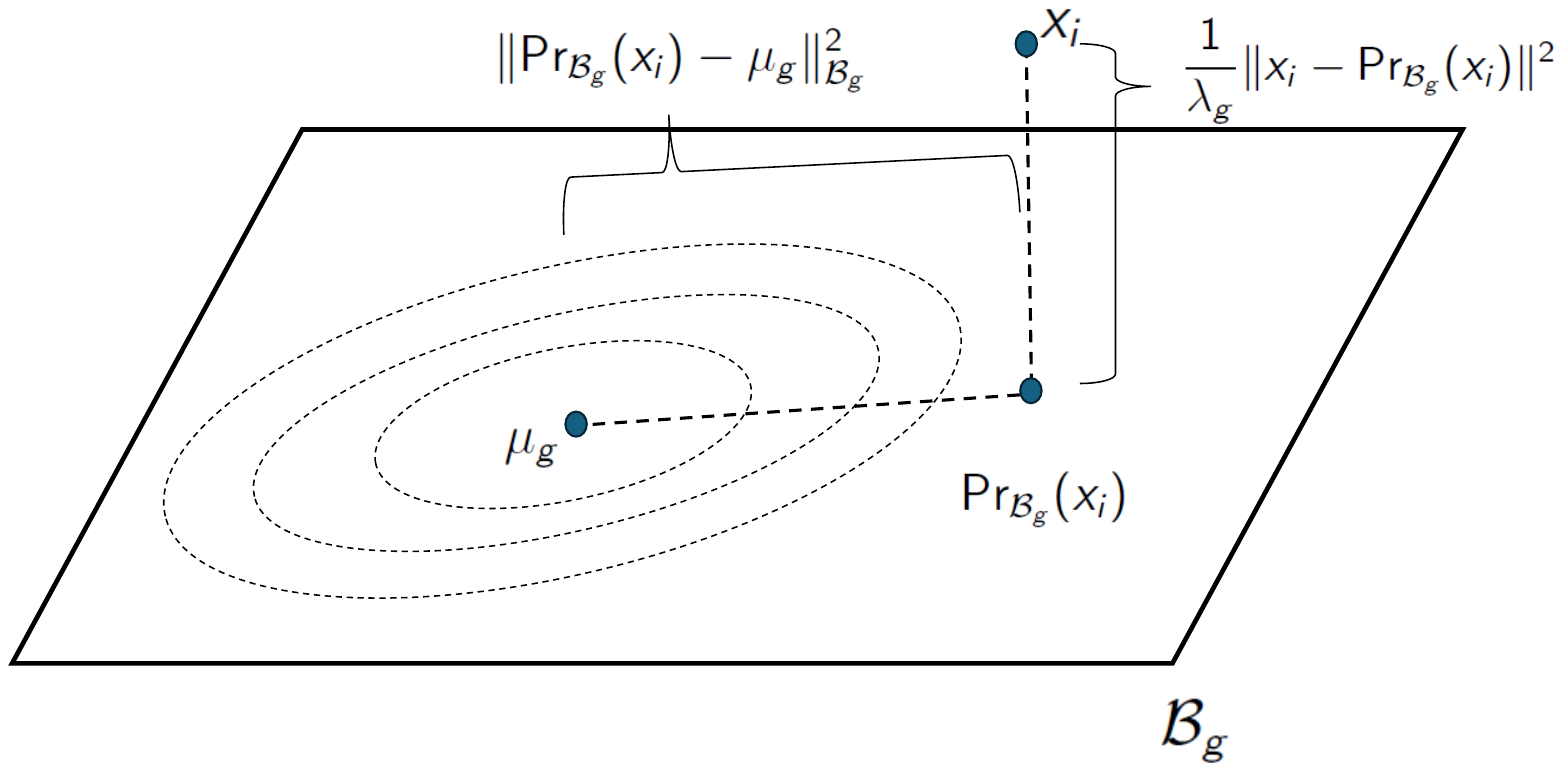}
    \caption{Summary of the type of distances considered (see text for explanation).}
    \label{esquema}
\end{figure} 

A methodology combining HDDC and trimming was briefly outlined in \cite{Bellas2012RobustClustering}, although in a more restricted setting and with a less comprehensive treatment.

\subsection{Existence and consistency results}\label{sec_32}

In this section, we will introduce population versions for the empirical problems given in Section \ref{sec_31} and analyze their theoretical properties. With that purpose, we introduce when dealing with the \textit{classification trimmed-likelihood} a set of assignment functions \( z_g: \mathbb{R}^p \rightarrow \{0,1\} \), for \( g = 0, 1, \ldots, G \). Given a probability measure $P$, we consider maximizing
\begin{equation}\label{problema_hard}
    \text{CL}(\theta,P):=\mathbb{E}_P \left[ \sum_{g=1}^G z_g(\cdot)\log\left(\pi_g\phi(\cdot; \mu_g, \Sigma_g)  \right) \right],
\end{equation}
with respect to assignment functions
$z_g : \mathbb{R}^p \rightarrow \{0,1\}$ satisfying $ \sum_{g=0}^G z_g(x) = 1$, for every $x\in\mathbb{R}^p$, and that $\mathbb{E}_P[z_0(\cdot)] = \alpha$. Additionally, the parameter $\theta\in\Theta_{c_1, c_2}$, with $\Theta_{c_1, c_2}$ as defined in Section \ref{sec_31}.

As $P$ is often unknown, if $P_n$ stands for the empirical measure from the result $X_1=x_1,...,X_n=x_n$ of an independent identically distributed (i.i.d) sample $X_1,...,X_n$, we can replace $P$ by $P_n$ (i.e., by $P_n=1/n\sum_{i=1}^n\delta_{\{ x_i\}}$) to recover the (empirical) classification trimmed HDDC problem introduced in Section \ref{sec_31}. Notice that, perhaps, $\mathbb{E}_{P_n}[z_0( \cdot)] = \alpha$ cannot be exactly achieved but this fact will not be important in our reasoning.

On the other hand, for the \textit{mixture} trimmed  HDDC problem, we may consider the constrained maximization of 
\begin{equation}\label{problema_mixt}
    \text{ML}(\theta,P):=\mathbb{E}_P \left[  z(\cdot)\log\left(\sum_{g=1}^G\pi_g\phi(\cdot; \mu_g, \Sigma_g)  \right) \right],
\end{equation}
with respect to assignment functions $z : \mathbb{R}^p \rightarrow \{0,1\}$
 such that $\mathbb{E}_P[z(\cdot)] = 1-\alpha$, and $\theta$ is constrained to be in $\Theta_{c_1, c_2}$. If $P_n$ stands for the empirical measure, we also recover the trimmed mixture likelihood in Section \ref{sec_31}.  

The assignments function can be obtained in both, trimmed classification and mixture likelihoods, from $\theta$ as follows. Given $\theta$, let us define the functions 
\begin{equation}\label{D_a}
  \left\{\begin{array}{ll}
D^H(x;\theta)&=\underset{{g=1,\ldots,G}}{\max}\lbrace D_g(x;\theta)\rbrace\\
     & \\
     D^M(x;\theta)&=\sum_{g=1}^G \exp(D_g(x;\theta)),
\end{array} \right.  
\end{equation}
where $D_g(x;\theta)=\log[\pi_g \, \phi(x; \mu_g, \Sigma_g)],$
is defined as in (\ref{posterior}). When considering hard assignment, the assignment functions can be defined as
\[
z_g(x):=z_g(x;\theta)=I\lbrace\lbrace x\::\: D^H(x;\theta)=D_g(x;\theta)\rbrace\cap\lbrace x\::\:D^H(x;\theta)\geq R(\theta;P)\rbrace \rbrace,
\]
with
$R(\theta;P)=\inf_{u}\lbrace H(u;\theta, P)\geq\alpha\rbrace$ for $H(u;\theta, P)= P(D^H(\cdot,\theta)\leq u)$. 
On the other hand, when considering the mixture approach, we have 
$$
z(x):=z(x;\theta)=I\lbrace x:D^M(x;\theta)\geq R(\theta;P) \rbrace,
$$
with now $H(u;\theta, P)= P(D^M(\cdot,\theta)\leq u)$.

To exclude in the subsequent analysis the probability distributions obviously unappropriated for the introduced approach, we will assume on the underlying distribution $P$ the very mild condition in (\ref{PR}). Failure to satisfy this condition would make the proposed methodology completely inappropriate for that $P$.
\begin{align}
& \text{The distribution } P \text{ is not concentrated on } G \text{ subspaces of dimensions } \nonumber\\
& q_1, \ldots, q_G \text{ after removing a probability mass equal to } \alpha \label{PR}
\end{align}

Under that mild assumption, we can state the following general existence result, that can be applied to both the population $P$ and the empirical $P=P_n$ problems when the $q_g$, $g=1,...,G$ intrinsic dimensions and the number of clusters $G$ are fixed beforehand:

\begin{proposition}\label{prop_existencia}
If (\ref{PR}) holds for the distribution $P$, then there exists some $\theta\in\Theta_{c_1, c_2}$ such that the maximum of (\ref{problema_hard}) and (\ref{problema_mixt}) when $\theta\in\Theta_{c_1, c_2}$ is attained whenever $\alpha \in (0,1)$ ($1\leq c_1 \leq \infty$ is allowed but $1\leq c_2 < \infty$ is required).    
\end{proposition}

Note that, in the previous result, no moment conditions on $P$ are required, provided that a strictly positive trimming level $\alpha > 0$ is considered. This implies that the existence result remains valid even for heavy-tailed distributions, since a strictly positive trimming level can effectively remove the influence of heavy tails.  Note that the condition (\ref{PR}) trivially holds for absolutely continuous distributions $P$, as well as for empirical measures $P_n$ based on a sufficiently large samples from such a continuous distributions P.

The following consistency result also establishes that the empirical constrained maximizer converges to the population constrained maximizer for the unknown distribution $P$ as the sample size $n$ increases:

\begin{proposition}\label{prop_consistencia}
    Let $\alpha>0$ and $P$ be a distribution with a strictly positive density function and such that $\theta_0\in\Theta_{c_1, c_2}$ is the unique constrained maximizer of (\ref{problema_hard}) (or, respectively, (\ref{problema_mixt})). If $\theta_n\in \Theta_{c_1, c_2}$ denotes a sequence of constrained sample maximizers based on the empirical measure $P_n$ obtained from an i.i.d. sample $X_1,...,X_n$ from $P$, then $\theta_n\rightarrow\theta_0$ almost surely (in both cases).
\end{proposition}

It is not straightforward to provide conditions under which the population maximizer $\theta_0$ is unique for a given $P$, specially in the case of the use of classification trimmed likelihoods. This can also be observed in simpler and highly constrained classification likelihood problem with $\alpha = 0$, which lead to the $k$-means method (see \cite{carcamo2024uniqueness}). Nevertheless, this assumption is at least plausible when dealing with $G$ well-separated normal components, although we acknowledge that a rigorous proof is definitely challenging. Note also that the stated consistency result, even in cases where the data sets are generated from Gaussian components satisfying the required conditions on their covariance matrices, does not guarantee that the empirical estimators converge exactly to the parameters of those Gaussian components. The consistency result should be understood as convergence toward the parameters that solve the population-level maximization of the trimmed likelihood functions, which only coincide with the parameters of the underlying Gaussian distributions in the ideal case where the components do not overlap at all, and where the scatter matrices estimated by trimming are corrected with appropriate correction factors. A detailed discussion of all these issues in the context of TCLUST is provided in \cite{garcia2025choice}.

On the one hand, the identifiability of the parameters of finite mixtures of multivariate normal distributions is well established. These identifiability results ensure that the population maximizer is unique whenever the underlying $P$ belongs to the corresponding finite-mixture model. However, this property may be lost when a positive trimming proportion $\alpha>0$ is introduced into the population framework. On the other hand, when the classification likelihood is considered, even in the absence of trimming, it is far from straightforward to establish conditions under which this uniqueness property holds. This issue is well-known in the case of $k$-means clustering, which corresponds to the simplest and most constrained classification-likelihood problem (see \cite{carcamo2024uniqueness}). Nevertheless, the assumption of uniqueness is at least plausible when dealing with $G$ well-separated normal components, although we acknowledge that providing a rigorous proof is highly challenging. Moreover, the consistency result should be interpreted as convergence toward the parameters that maximize the population-level trimmed (or classification) likelihood target function. These parameters coincide with those of the underlying Gaussian distributions only in the idealized case where the components do not overlap at all and appropriate correction factors are applied to the scatter matrices estimates to account for the effect of trimming. A detailed discussion of these issues in the context of TCLUST can be found in \cite{garcia2025choice}.

\subsection{Algorithm}\label{sec_33}

An algorithm to approximately maximize the classification trimmed likelihood \eqref{tclust} or the mixture trimmed likelihood \eqref{mixtTCLUST}, under the parsimonious scatter parameterization and the eigenvalue-ratio constraints (\ref{ER}), will be presented in this section.

The proposed algorithm relies on multiple random initializations (\texttt{nstart} random initializations) of the parameters, followed by iterative C-steps until convergence or until a maximum number of iterations (\texttt{csteps}) is attained. The iterated parameters providing the highest value of the corresponding (classification or mixture) trimmed likelihood target functions are chosen as final output of the algorithm. The C-steps can be viewed as a type of trimmed (classification) EM step, where $z_{ig}$ ``membership values'' are updated using the parameters $\theta^{(t-1)}$ from the previous iteration. These $z_{ig}$ values are binary (0--1) in the case of the trimmed classification likelihood, or posterior probabilities in the case of mixture likelihoods, but in both settings it is necessary to compute the $D_g(\cdot,\theta)$ values as in (\ref{posterior}). Trimming is applied by setting the membership values $z_{ig}$ equal to 0 for the fraction $\alpha$ of observations with the smallest $D^a(x;\theta)$ values, with $a=H$ and $a=M$, that is, for those observations that are least likely to have been generated from the current parameter values. %After the calculus of these membership values and applying trimming, the parameters are updated in $\theta^{(t)}$
Finally, in the C-steps, the parameters are updated through the constrained maximization of the so-called ``trimmed complete-data likelihood,'' defined as
$$
         \sum_{i=1}^n \sum_{g=1}^G z_{ig}\log\left[\pi_g\phi(x_i;\mu_g^{(t)},\Sigma_g^{(t)})\right].
$$
In fact, parameters are updated starting from weighted sample means and weighted sample covariance matrices with weights given by these $z_{ig}$ values, where also a proper truncation of the eigenvalues has to be performed. That proper eigenvalue truncation is a an extension of the eigenvalue truncation introduced in \cite{fritz2013fast}. It can be shown that a complete C-step increases the trimmed likelihoods monotonically, but several random initializations are performed in order to obtain a global maximum rather than a local maximum.  It is also important to note that only the first eigenvectors/eigenvalues (corresponding to the $q_g$ largest eigenvalues), $\{u_{gl}^{(t)}, \lambda_{gl}^{(t)}\}_{l=1}^{q_g}$, and the common smallest eigenvalue $\lambda_{g}^{(t)}$ for each component, $g=1,\ldots,G$, need to be updated.

\begin{enumerate}
    \item \textbf{Random initializations:}  Consider \texttt{nstart} random initializations by defining \texttt{nstart} initial parameter sets, each denoted by $\theta^{(0)}$. Details of this initialization step are given in Remark \ref{re1}.
    
    \item  \textbf{C-steps:} For each $\theta^{(0)}$, repeat the C-steps for a maximum of \texttt{csteps} iterations ($t = 1, \ldots, \texttt{csteps}$), or until a convergence criterion is met (for instance, when the change in the updated parameters is sufficiently small).
    \begin{enumerate}
        \item[2.1.] \textbf{Update $z_{ig}$  membership values}: From the parameters $ \theta^{(t-1)} $  in the previous iteration, compute $D^a(x_i; \theta^{(t-1)})$ as in (\ref{D_a}), where $a = H$ for the classification likelihood maximization and $ a = M $ for the mixture case, for each observation and sort them in $$D^a(x_{(1)}; \theta^{(t-1)})\leq D^a(x_{(2)}; \theta^{(t-1)}) \leq ... \leq D^a(x_{(n)}; \theta^{(t-1)}).$$ Select the set $ R_0$ consisting of the $ [n \alpha] $ indexes with the smallest $D^a(x_i; \theta^{(t-1)})$ values, i.e., $R_0=\{i: D^a(x_i; \theta^{(t-1)}) \leq D^a(x_{([n \alpha])}; \theta^{(t-1)})\}$.%, which are the indexes of the trimmed observations. %The remaining indexes, collected in the set $ \{ 1, \ldots, n \} \setminus \mathcal{H}_0 $, are those corresponding to the non-trimmed ones.
        %In the case of hard assignment, $ \mathcal{H} $ is split as $ \mathcal{H} = \{ \mathcal{H}_1, \ldots, \mathcal{H}_G \} $, where
        %$\mathcal{H}_g = \{ i \in \mathcal{H} : D_g(x_i; \theta^{(t-1)}) = D_i \}$ that denotes the indexes of the observations assigned to component $g$.
        Then, compute the membership values $z_{ig}$ as
    $$
    z_{ig} =
    \begin{cases}
    1 & \text{if } D_g(x_i,\theta^{(t-1)})=D^H(x_i; \theta^{(t-1)}) \\
    0 & \text{otherwise}
    \end{cases},
    $$
    for classification trimmed likelihoods, and 
    $$
    z_{ig} = \frac{\exp (D_g(x_i,\theta^{(t-1)}))}{D^M(x_i; \theta^{(t-1)})},
    $$
    for mixtures trimmed likelihoods. Finally, perform the trimming step by imposing $ z_{ig} = 0 $ for all $ g = 1, \ldots, G$, if $i \in R_0 $.
    
        \item[2.2.] \textbf{Update parameters:} %while $t\leq \texttt{csteps}$ (or not meeting the specified convergence criteria):
        \begin{enumerate}
            \item[2.2.1.] \textit{Update weights:} 
            $$ \pi_g^{(t)} = \frac{\sum_{i=1}^n z_{ig}}{[n(1 - \alpha)]} $$
            \item[2.2.2.] \textit{Update means:} 
            $$ \mu_g^{(t)} = \frac{\sum_{i=1}^n z_{ig} x_i}{\sum_{i=1}^n z_{ig}} $$
            
            \item[2.2.3.] \textit{Update subspaces and scatter parameters:} Start by computing the weighted sample covariance matrices
            $$ S_g = \frac{1}{\sum_{i=1}^n z_{ig}} \sum_{i=1}^n z_{ig} (x_i - \mu_g^{(t)})(x_i - \mu_g^{(t)})^\top $$
            Extract the eigenvectors $ u_{g1}^{(t)}, \ldots, u_{gq_g}^{(t)}$ associated with the $ q_g $ largest eigenvalues $d_{g1},\ldots, d_{gq_g}$, and compute:
            $$d_g = \frac{1}{p - q_g} \left( \text{trace}(S_g) - \sum_{l=1}^{q_g} d_{gl} \right).$$

            Since only the eigenvectors corresponding to the $q_g$ largest eigenvalues in \eqref{bayes_hd_first} are required, the Arnoldi method can be applied to $S_g$ to reduce computation time, as previously noted in \cite{bouveyron2007high}.

            It is clear that the eigenvalues  in 
           $\left\{ \{ d_{gl} \}_{l=1}^{q_g} \right\}_{g=1}^G$ and $\{ d_g \}_{g=1}^G$ do not necessarily should satisfy the required constraints, depending on $ c_1 $ and $ c_2 $. A eigenvalues-truncation operation has to be then performed, which can be seen as an adaptation of the approach followed for TCLUST and introduced in \cite{fritz2013fast}. Details on this truncation of eigenvalues will be given Remark \ref{re2}. The result of the complete process returns updated
           $ \{ \lambda_{gl}^{(t)} \}_{l=1}^{q_g} $ and $  \lambda_g^{(t)}$ eigenvalues for every $g=1,...,G$. 
        \end{enumerate}
        
    \end{enumerate}
    \item \textbf{Output:} Return the iterated parameters that achieve the maximum value of the target function in (\ref{tclust}) or (\ref{mixtTCLUST}) among the \texttt{nstart} initializations.
\end{enumerate}

%To improve computational efficiency, we use a two-stage strategy in which many random initializations are generated and only a few initial concentration steps are performed, after which the most promising solutions are selected and fully iterated until convergence using a larger number of steps. This possibility will be detailed in Remark \ref{re3}.

\begin{remark}\label{re1} \normalfont One of the drawbacks of the application of TCLUST when increasing dimension $p$ is that $G(p+1)$ observations needed to be randomly selected to initialize the $\Sigma_g^{(0)}$ matrices \citep{fritz2013fast}. Note that $p+1$ random observations in general position are needed to derive non-singular scatter matrices, for $p$-variate variates, by computing their sample covariance matrix. Additionally, when applying TCLUST, it is required $G (p+1)$ observations ideally free of outliers but also partitioned in such a way that the first $p+1$ observation belongs to one of the $G$ underlying cluster, the following $p+1$ observations to the second cluster and so on. Of course, this is a idealized situation and TCLUST can work if this ideal initialization assumptions do not hold, but it serves to give us an idea on how difficult became the task of initializing TCLUST when $p$ is large. On the other hand, we will see that $ \sum_{g=1}^G (q_g + 2) $ random observations are actually needed, which is notably smaller than $G (p+1)$, when applying the proposed trimmed HDDC algorithm, which makes more feasible to think of relying in appropriate random initializations, mainly when the intrinsic dimension $q_g$ are small in comparison with the ambient space dimension $p$. 

For the initialization of the parameters in $\theta^{(0)}$, we propose randomly selecting a subsample  
$ \{ x_{g_1}, \ldots, x_{g_{q_g+2}} \} \subset \{ x_1, \ldots, x_n \}$. Note that $q_g + 2$ random observations in general position is the smallest number of observations required to initialize the parameters of each cluster. The initial location parameters can be set  as  $\mu_g^0 = (x_{g_1} + \cdots + x_{g_{q_g+2}})/(q_g + 2)$, and a sample covariance matrix $S_g$ can be computed using these $ q_g + 2 $ observations which has rank equal to $q_g+1$. We then extract the eigenvectors $u_{g1}^{(0)}, \ldots, u_{gq_g}^{(0)}$ corresponding to the $q_g$ largest eigenvalues of that $S_g$ matrix for initializing the vectors spanning the subspace $\mathcal{B}_g$, and the $q_g$ largest eigenvalues are considered as initial $\lambda_{g1}^{(0)}, \ldots, \lambda_{gq_g}^{(0)}$ eigenvalues. Finally, we take $\lambda_g^{(0)}$ being equal to the smallest eigenvalue of the $S_g$ matrix divided by $p - q_g$.Of course, these initial eigenvalues may not satisfy the required eigenvalue constraint, but since this approach is used only for initialization, the constraints can be enforced later via eigenvalue truncation in subsequent C-steps. Moreover, to make the computation of eigenvectors and eigenvalues more efficient, the following strategy can be employed. Let $ \mathfrak{X} $ be a $(q_g + 2) \times p $ matrix whose columns correspond to the observations centered by $ \mu_g^0 $. Instead of computing the eigenvalues and eigenvectors of the covariance matrix  $S_g = \frac{1}{q_g + 2} \mathfrak{X}^\top \mathfrak{X}$, which is of dimension $ p \times p $, we work with its transpose $\frac{1}{q_g + 2} \mathfrak{X} \mathfrak{X}^\top,$ which is of size $ (q_g + 2) \times (q_g + 2) $, since $ q_g + 2 \ll p $ and both matrices share the same non-zero eigenvalues due to the properties of Gram matrices. In this way, if $u$ is an eigenvector of $ (q_g + 2) \times (q_g + 2) $ matrix associated with eigenvalue $\lambda$, then $v = \frac{1}{\sqrt{\lambda}} \mathfrak{X}^\top u$
is also an eigenvector of $S_g$ associated with the same eigenvalue $\lambda$. Regarding the initialization of the weights, the $\pi_g^{(0)}$ can be set either to $1/G$ for all components or randomly in $[0,1]$, ensuring that they sum to one.
\end{remark}

\begin{remark}\label{re2} \normalfont It is not difficult to see (essentially combining the arguments in \cite{bouveyron2007high} and \cite{garcia2008general}) that the vectors spanning the optimal approximating subspaces $\mathcal{B}_g$ should be the eigenvectors of the weighted sample covariance $S_g$ matrices in Step 2.2.3. of the algorithm associated to the $q_g$ largest eigenvalues. With respect to how the eigenvalues should be updated, let us start from the $\left\{ \{ d_{gl} \}_{l=1}^{q_g} \right\}_{g=1}^G$ and $\{ d_g \}_{g=1}^G$ derived from the weighted sample covariance $S_g$ matrices, as described in Step 2.2.3. of the presented algorithm. As justified in \cite{garcia2008general} and \cite{fritz2013fast}, conditional on the available eigenvectors, it would be ideally needed to replace the values $\{ d_{gl} \}_{l=1}^{q_g}$ and $d_g$ by new values $\{ \lambda_{gl} \}_{l=1}^{q_g}$ and $\lambda_g$ such that they minimize 
$$
\sum_{g=1}^G n_g \sum_{l=1}^{q_g} \left( \log \lambda_{gl} + \frac{d_{gl}}{\lambda_{gl} }\right) 
+ \sum_{g=1}^{G} n_g (p-q_g) \left( \log \lambda_{g} + \frac{d_{g}}{\lambda_{g}}\right),
$$
while still ensuring that the values $\{ \lambda_{gl} \}_{l=1}^{q_g}$ and $\lambda_g$ satisfy the relative order among the largest and smallest eigenvalues, together with the imposed constraints (\ref{eigenavalus_ratio}), and where $n_g = \sum_{i=1}^n z_{ig}$. With that purpose, let us use the following notation for the truncated eigenvalues $$[d]^{m}_{c} = \min\left\{\max\{d, m\}, cm \right\}.$$
We may consider the following two separated minimizations  
\begin{equation}\label{ind1}
 \min_{m_1>0} \sum_{g=1}^G n_g \sum_{l=1}^{q_g} \left( \log [d_{gl}]^{m_1}_{c_1} + \frac{d_{gl}}{[d_{gl}]^{m_1}_{c_1} }\right),   
\end{equation}
and
\begin{equation}\label{ind2}
  \min_{m_2>0} \sum_{g=1}^{G} n_g (p-q_g) \left( \log [d_{g}]^{m_2}_{c_2} + \frac{d_{g}}{[d_{g}]^{m_2}_{c_2} }\right),  
\end{equation}
for the two constants $c_1$ and $c_2$ considered for the eigenvalue ratio constraints in (\ref{ER}). These two minimization can be efficiently carried out by performing an slightly modified version of the optimal eigenvalues-truncation procedure in \cite{fritz2013fast}. Let $m_1^{\mathsf{opt}}$ and $m_2^{\mathsf{opt}}$ be the values of $m_1$ and $m_2$ where the minimum values of, respectively, (\ref{ind1}) and (\ref{ind2}) are attained. We can use these two values, $m_1^{\mathsf{opt}}$ and $m_2^{\mathsf{opt}}$, to obtain the initial truncated eigenvalues
$$
d_{gl}^* = [d_{gl}]^{m_1^{\mathsf{opt}}}_{c_1},\text{ for }l=1,...,q_g\text{, and } d_{g}^*=[d_{g}]^{m_2^{\mathsf{opt}}}_{c_2}.
$$
If these initial $d_{gl}^*$ and $d_{g}^*$ do directly satisfy the desired constraint imposed on the largest and smallest eigenvalues (i.e. it holds that $d_{gl}^*\geq d_{g}^*$ for every $g$ and $l=1,...,q_g$) then we can trivially update $\lambda_{gl}^{(t)} =d_{gl}^*$, $l=1,...,q_g$, and $\lambda_{g}^{(t)} =d_{g}^*$ in Step 2.2.3 of the algorithm. On the other hand, if the constraints do not hold for any $g$ then a iterative procedure has to applied. In that case, if there exists $g \in \{1,\ldots,G\}$ and $j \leq q_g$ such that $d_{gl}^* \geq d_{g}^*$ for $l<j$ but $d_{gj}^* < d_{g}^*$, then we redefine all the eigenvalues violating the constraints and setting them equal to a common value given by
\begin{equation}\label{new_equal}
 d_{gj}^* = d_{gj+1}^* = \cdots = d_{gq_g}^* = d_{g}^*
 = \frac{\sum_{l=j}^{q_g}[d_{gl}]^{m_1^{\mathsf{opt}}}_{c_1} 
   + (p-q_g)[d_{g}]^{m_2^{\mathsf{opt}}}_{c_2}}
   {p-j} .
\end{equation}
Note that the common value in (\ref{new_equal}) coincides with the value of $\lambda$ that minimizes the following expression
$$\sum_{l=j}^{q_g} \left( \log \lambda + \frac{[d_{gl}]^{m_1^{\mathsf{opt}}}_{c_1} }{\lambda }\right) 
+  (p-q_g) \left( \log \lambda + \frac{[d_{g}]^{m_2^{\mathsf{opt}}}_{c_2}}{\lambda}\right).$$
With these updated values, $d_{gl}^*$ and $d_{g}^*$, the minimizations (\ref{ind1}) and (\ref{ind2}) are carried out again, where the original $d_{gl}$ and $d_{g}$ are replaced by these new $d_{gl}^*$ and $d_{g}^*$ values. In a completely analogous and iterative manner, any subsequently truncated eigenvalues that still violate the constraints are replaced by a common eigenvalue as in (\ref{new_equal}). This iterative procedure is repeated until no further changes are required, yielding the final updates $\lambda_{gl}^{(t)}$ and $\lambda_{g}^{(t)}$ in Step 2.2.3 of the proposed algorithm. In most cases, we have observed that only a few iterations are required for the proposed eigenvalue update. Obviously, the eigenvalue truncation procedure becomes considerably simpler when $c_1 = \infty$, since the iterative scheme is no longer needed in that case.

\end{remark}

\begin{remark}\label{re3}
    \normalfont {To improve the computational efficiency of the proposed algorithm, we adopt a two stage strategy for the C-steps. The proposal is to consider a number of random starts \texttt{nstart} as large as possible, but to apply only a very small number \texttt{cstep1} of initial C-steps (for example, \texttt{cstep1} equal to 2 or 4) to each random initialization. Based on these preliminary iterations, only the $\texttt{nkeep}$ most ``promising'' initializations, namely those yielding the largest values of (\ref{tclust}) or (\ref{mixtTCLUST}), are retained and subsequently iterated until convergence or until a larger number of final C-steps $\texttt{cstep2}$ is reached ($\texttt{cstep1} < \texttt{cstep2}$). This approach avoids fully iterating poor initializations and makes more effective use of the available computational effort, allowing for a better exploration of the parameter space than would be achieved with a smaller number of unnecessarily fully iterated initializations. There is also the possibility of running versions of the proposed algorithms assuming equal weights for the components, i.e., we keep the weights fixed as $\pi_g^{(t)} = 1/G$ from initialization. This is achieved by selecting \texttt{equal.weights=TRUE} in the code implementation, which contributes to obtaining groups of more similar sizes in terms of assignments for non-trimmed observations based on posterior probabilities.}
\end{remark}

The complete code to replicate the experiments and apply the proposed methodology can be found in the GitHub repository: \url{https://github.com/luciatrapote/thddc-experiments}.

\section{Selection of tuning parameters}\label{sec_4}

The application of the tHDDC approach requires the specification of several tuning parameters. Of course, the selection of all these parameters, such as $G$, $q_1,\ldots,q_G$, $\alpha$, $c_1$, and $c_2$, is a challenging task. This difficulty is analogous to that encountered in many other robust (and even non-robust) clustering methods. In fact, determining the number of clusters $G$ is rarely straightforward in most widely used clustering techniques. In the robust clustering framework, this problem becomes even more complex, because $G$ and $\alpha$ may be interrelated. For instance, a higher trimming level $\alpha$ may lead to the removal of small clusters, which are then treated as outliers and discarded, thereby reducing the number of clusters $G$. Similarly, the parameters $c_1$ and $c_2$ also determine the type of clusters the user is interested in identifying and may affect the number $G$ of clusters required.

It is also important to stress that this process of parameter selection cannot be fully automated, since (subjective) user choices are often required at different stages of the analysis to determine the type of clusters to be detected and which the user is willing to assume as ``cluster'', which in turn depends on the ultimate purpose of the clustering analysis. Anyway, it makes sense to assist the user in reducing the set of ``interesting'' partitions 
to a smaller number, associated with particularly meaningful choices of the parameters  \citep{garcia2011exploring,cerioli2018finding,garcia2025choice}. This is precisely the type of strategy that has been followed in some tools developed within the TCLUST framework, and similar ideas could be transferred to the present setting. In this work, however, we do not aim to explore or discuss these possibilities. Instead, we focus on presenting a nearly automated way of determining the intrinsic dimensions $q_1,\ldots,q_G$, together with some graphical diagnostic tools that allow us to reassess the extent to which our parameter choices may lead to questionable decisions or, conversely, to confirm the cluster and outlier assignments and to clarify the quality of these decisions.

\subsection{Automated determination of the intrinsic dimensions}\label{catell}

As already mentioned, even when the number of clusters $G$ is fixed by the user, selecting the intrinsic dimensions $q_1,\ldots,q_G$ remains a challenging task. In this section, we show that these $q_1,\ldots,q_G$ can be treated as internal parameters to be estimated within the iterative steps of the algorithm, thereby partially alleviating the burden of manual specification and turning them into data-driven, cluster-dependent parameters. To this end, we adopt the procedure described in \cite{bouveyron2007high}, which relies on a Cattell-type criterion based on consecutive differences $d_{gl+1}-d_{gl}$ when $l=0,...,q_{\max}-1$, where $d_{gl}$ are the sorted eigenvalues of the weighted sample covariance matrix $S_g$ computed in step 2.2.3 of the algorithm in Section \ref{sec_32}. As the intrinsic dimensions are not fixed in advance, the initialization step in Step 1 of the proposed algorithm, detailed in Remark \ref{re1}, can be carried out with different ``initial'' intrinsic dimensions. Our proposal is to begin the initialization described in Remark \ref{re1} with equal intrinsic dimensions, that is, $q_1=\cdots=q_G=q_{\text{ini}}$. We have observed that, in many cases, the final intrinsic dimensions obtained by the procedure are not excessively sensitive to the choice of $q_{\text{ini}}$, although in other situations they clearly are. Our initial experiments (see Section \ref{sec_5}) seem to indicate that it is often advisable to start with a not excessively large value of $q_{\text{ini}}$, as this tends to facilitate the correct initial partitioning into clusters of the $G \times (q_{\text{ini}} + 2)$ random observations used for the initialization. Note that we have set a $q_{\max}$ value which acts as an upper bound on the clusters' intrinsic dimensions, so with the assumption that $q_g<q_{\max}$ for a sufficiently large $q_{\max}\leq p-1$.  To be more precise, instead of differences, we propose using normalized sorted differences to update the $q_g$ in the C-steps by setting
$$
q_g = \max \left\{  j \in \{0,\ldots,q_{\max}-1\} : 
\frac{d_{gj+1}-d_{gj}}{\max\limits_{l=0,\ldots,q_{\max}-1} \left\{ d_{gl+1}-d_{gl} \right\}} 
> \texttt{tresh} \right\},
$$
and where \texttt{tresh} is a (small) threshold value.

With $q_{\max}$ and \texttt{tresh} fixed, after applying C-steps for the \texttt{nstart} random initializations, 
we obtain \texttt{nstart} (or \texttt{nkeep} values when applying Remark \ref{re3}) values of the target trimmed likelihood 
(either mixture or classification), possibly associated with different intrinsic dimensions $q_g$, $g=1,\ldots,G$. 
The problem is that if we simply choose the solution with the largest value of  (\ref{problema_hard}) or (\ref{problema_mixt}), we would tend to favor solutions with larger intrinsic dimensions $q_g$, 
thus encouraging ``overfitting'' and leading to models that are more complex than actually needed. To prevent this failure, \cite{bouveyron2007high} propose choosing the iterated solution with the largest complex-penalized likelihood (in fact, they proposed a BIC-type solution). With this difficulty in mind, we propose the iterated solution with the smallest value of 
$$
	-2\times \textsf{trimmed log-likelihood}+\textsf{complexity penalty},
$$
where ``\textsf{trimmed log-likelihood}'' stands for the value attained in (\ref{tclust}) or (\ref{mixtTCLUST}) after iteration, and for measuring that ``\textsf{complexity penalty}'' it is proposed using
\begin{align*}
&\log\big([n(1 - \alpha)]\big) 
\left[ 
     \underbrace{G - 1}_{\text{weight pars.}} + \underbrace{Gp}_{\text{means pars.}} +
      \underbrace{ 1 + \left( \sum_{g=1}^G q_g  - 1 \right) \left( 1 - \frac{1}{c_1} \right)}_{\text{largest eigenvalues pars.}}
\right.\\
    & \left.
     \underbrace{+ 1 + (G - 1) \left(1 - \frac{1}{c_2} \right)}_{\text{smallest eigenvalues pars.}} +
    \underbrace{\sum_{g=1}^G \left( q_g p - \frac{q_g(q_g - 1)}{2}\right)}_{\text{first $q_g$ orthonormal eigenvectors}} 
\right]. 
\end{align*}
In the previous expression, we take into account a kind of ``number of free parameters'', noting that higher values of $c_1$ and $c_2$ in (\ref{ER}) yield less restricted solutions. Note how the ``\textsf{complexity penalty}'' changes as $c_1$ and $c_2$ approach their extreme values, namely $1$ and $\infty$. This is obviously an extension of the modified BIC introduced in \cite{cerioli2018finding}.

\subsection{Graphical summaries and diagnostic tools}\label{sec4.2}

The vectors $u_{g1}, u_{g2}, \ldots, u_{gq_g}$ that generate the approximating subspaces $\mathcal{B}_g$ can be regarded as loading vectors and may be interpreted as ``modes of variation'' that serve to characterize how observations fluctuate around the clusters' means $\mu_g$. Associated with these loadings, the cluster scores $t_{ig,l}=\langle x_i - \mu_g, u_{gl}\rangle$ for $l=1,\ldots,q_g$ provide a low-dimensional representation of each observation within its cluster. In different plots, for each component $g$, we can represent the scores for the observations assigned to that cluster. The scores corresponding to the trimmed observations can also be displayed, by assigning each of them to the cluster that yields the highest posterior probability under the estimated parameters.
The scores (of both trimmed and non-trimmed observations) not only capture the proximities of observations within the corresponding cluster, but also provide interpretability through their direct relationship with the loading vectors defined above.

To assess cluster adequacy and detect potential atypical observations, two types of diagnostic measures can be considered. The score distances
$$
\text{SD}_{ig}
= \left\| \operatorname{Pr}_{\mathcal{B}_g}(x_i) - \mu_g \right\|_{\mathcal{B}_g}
= \sqrt{\sum_{l=1}^{q_g} \frac{t_{ig,l}^2}{\lambda_{gl}}}
$$
quantifies how far the projection of $x_i$ onto the subspace $\mathcal{B}_g$ lies from the clusters' mean centers $\mu_g$, measured in the metric induced by the eigenvectors and eigenvalues corresponding to the $q_g$ largest eigenvalues in each cluster. Complementarily, the orthogonal distances
$$
\text{OD}_{ig}
= \frac{1}{\sqrt{\lambda_g}} \big\| x_i - \operatorname{Pr}_{\mathcal{B}_g}(x_i) \big\|
= \frac{1}{\sqrt{\lambda_g}}
\left\lVert x_i - \mu_g - \sum_{l=1}^{q_g} t_{ig,l} u_{gl} \right\rVert
$$
provide the properly scaled orthogonal residual deviations of $x_i$ from the approximating $\mathcal{B}_g$  subspaces, and therefore captures the extent to which the observation departs from the linear structure defining cluster $g$.

As already commented, both distances are computed for the indices in $A_g=\{i: D_{ig}=D_i\}$, where trimmed observations are also included. To determine whether a point should be considered extreme with respect to its cluster, cutoff values are employed. For the score distances, the threshold $\sqrt{\chi^{2}_{q_g;0.975}}$
can be used, while for the orthogonal distances the cutoff is given by $(\widehat{\mu}_{\text{OD}}+\widehat{\sigma}_{\text{OD}} z_{0.975})^{3/2}$, where $\widehat{\mu}_{\text{OD}}$ and $\widehat{\sigma}_{\text{OD}}^{2}$ denote the robust mean and variance obtained via the MCD estimator applied to $\{\text{OD}_{ig}^{2/3}\}_{i=1}^n$, following the approach in \cite{hubert2005robpca}. These graphical diagnostic tools provide a coherent and robust framework for evaluating cluster membership and detecting potential outliers within each group.

In addition to the graphical summaries described above, further exploratory tools can be used to assess the stability of both cluster assignments and trimming decisions along the so-called ``discriminant factors''. Let $\hat{\theta}$ denote the estimated parameters obtained after the maximization of \eqref{problema_hard} or \eqref{problema_mixt}. 
Recall from Section \ref{sec_32} that, for any $\theta$ and any component $g$, the value $D_g(x;\theta)=\log\!\big(\pi_g \phi(x;\mu_g,\Sigma_g)\big)$ is associated to how likely is that $x$ would have been generated by the normal component $g$. Therefore, if we consider ordered values $ D_{(1)}(x_i;\hat{\theta})\leq\cdots\leq D_{(G)}(x_i;\hat{\theta}),$ we can gain insight into how reliable or trustworthy the classification decision for $x_i$ is when using tHDDC. To be more precise, if $x_i$ is a \textit{not trimmed observation} and is assigned to cluster $g$, i.e. $D_{(G)}(x_i;\hat{\theta})=D_g(x_i;\hat{\theta})$, then the strength of this allocation can be evaluated by comparing the largest and second-largest component contributions. Accordingly, we define
$$\text{DF}(i)=D_{(G)}(x_i;\hat{\theta}) - D_{(G-1)}(x_i;\hat{\theta}),\text{ when } x_i \text{ is not trimmed}.$$
Note that using a difference of logarithms is equivalent to consider the logarithm of a ratio, and we can denote it as ``discriminant factor'' and write as $\text{DF}(i)$.
On the other hand, for the 
\textit{trimmed observations} $x_i$,
let us consider  $x_{(1)},\ldots,x_{([n\alpha])}$ the trimmed observations by using the same notation as in Step 2, while $x_{([n\alpha]+1)}$ denotes the first non-trimmed observation. Although trimmed observations are not assigned to any cluster, their maximal component contribution can still be computed. The trimming decision for $x_i$ can therefore be assessed by comparing $D_{(G)}(x_i;\hat{\theta})$ with the boundary value $D_{(G)}(x_{(\lceil n\alpha\rceil+1)};\hat{\theta})$ and we consequently use
$$  \text{DF}(i)=D_{(G)}(x_{(\lceil n\alpha\rceil+1)};\hat{\theta})-D_{(G)}(x_i;\hat{\theta})\text{ for }x_i\text{ when trimmed.}$$

In both cases, for both trimmed and untrimmed observations, we have that $\mathrm{DF}(i) \geq 0$, but $\mathrm{DF}(i)$ values close to zero indicate borderline or ambiguous assignments or trimming decisions, whereas large positive values correspond to clear cluster assignments or clearly identified outlying trimmed observations.
Discriminant-type measures of this nature were previously considered in robust clustering based on trimming in \citet{fritz2012tclust} and, in the absence of trimming, in \citet{VanAelst2006}. Note that in these two previous work, $-\text{DF}(i)$ was used instead. However, here we adopt the opposite sign convention so that the resulting plots are more consistent with standard silhouette-type displays \citep{Rousseeuw1987}, yielding a similar interpretation. %This plot is an ordered representations of the $\text{DF}(i)$ values.
Some illustrations of these proposed graphical summaries will be given in Section \ref{sec_6}.

\section{Simulation study}\label{sec_5}
In this section, we present simulation results for artificial datasets to illustrate the performance of tHDDC compared with the two previously commented TCLUST and RLG approaches. 

Three scenarios were considered and, in all cases, two groups were generated from Gaussian distributions in $\mathbb{R}^{200}$. The total sample size is $n = 1000$, with 570 observations generated from the first Gaussian distribution, 380 generated from the second Gaussian distribution, and an additional 50 observations corresponding to uniformly distributed noise component. A parameter $\delta$ controls the separation between the two Gaussian distributions in such a way that when $\delta$ is close to 0, the groups exhibit substantial overlap, whereas larger values of $\delta$ yield more clearly separated clusters. The covariance matrices are constructed through spectral decompositions of the form $\Sigma_g = U_g \Delta_g U_g^\top$, where $\Delta_g$ are diagonal matrices containing controlled eigenvalues and $U_g$ are random orthogonal matrix. When entering the eigenvalues, we use the notation $\texttt{seq}(a,b;q)$ to denote a sequence of $q$ equally spaced points with endpoints $a$ and $b$, that is, with step size $(b - a)/(q - 1)$. The random orthogonal matrices $U_g$ will be obtained by considering SVD decompositions of $200\times 200$ random matrices with independent standard normal entries. In addition to the two normally distributed populations, each scenario included $50$ contaminating observations generated independently from a uniform distribution on the hypercube $[-2,2]^{200}$. This region was chosen because it approximately overlaps with the range where the observations from the Gaussian components takes values. 

Two different orthogonal matrices $U_1$ and $U_2$ are randomly taken in Scenarios~1 and~2. This setting will reflect situations in which groups are going to differ not only in their intrinsic dimensionality and shape of the components, but also in the orientation of their latent subspaces. To be more precise, we consider:

\begin{itemize}
    \item \textbf{Scenario~1: } \textit{Cluster 1} with $q_1 = 3$, $\mu_1 = 0 \cdot 1_{200}$, and $\Sigma_1 = \mathrm{diag}(\texttt{seq}(9,7;3),\allowbreak 0.15, 0.15, \ldots)$; and \textit{Cluster 2} with $q_2 = 1$, $\mu_2 = \delta \cdot 1_{200}$, and $\Sigma_2 = \mathrm{diag}(5, 0.45, 0.45, \ldots)$.
    \item \textbf{Scenario~2: } \textit{Cluster 1} with $q_1 = 10$, $\mu_1 = 0 \cdot 1_{200}$, and $\Sigma_1 = \mathrm{diag}(\texttt{seq}(20,11;10),\allowbreak 0.15, 0.15, \ldots)$; and \textit{Cluster 2} with $q_2 = 1$, $\mu_2 = \delta \cdot 1_{200}$, and $\Sigma_2 = \mathrm{diag}(\texttt{seq}(15,11;5),\allowbreak 0.25, 0.25, \ldots)$.
\end{itemize}

In contrast, for Scenario~3 both groups will share the same orthogonal matrix, i.e., $U_1 = U_2$. This scenario was deliberately included to induce a higher degree of overlap between the clusters. Indeed, when the intrinsic dimensions of the subspaces are small (as in Scenario~1, with $q_1 = 3$ and $q_2 = 1$), substantial overlap may arise even if the covariance matrices are oriented differently. However, for higher intrinsic dimensions (such as $q_1 = 10$ and $q_2 = 5$), distinct orientations typically lead to well-separated subspaces. Therefore, in Scenario~3 we enforce a common orientation in order to create a more challenging clustering setting, where separation relies primarily on differences in eigenvalues and mean vectors rather than on subspace orientations. Specifically, in Scenario~3, we consider:
\begin{itemize}
    \item \textbf{Scenario~3: } \textit{Cluster 1} with $q_1 = 10$, $\mu_1 = 0 \cdot 1_{200}$, and $\Sigma_1 = \mathrm{diag}(\texttt{seq}(20,15;10),\allowbreak 0.15, 0.15, \ldots)$; and \textit{Cluster 2} with $q_2 = 1$, $\mu_2 = \delta \cdot 1_{200}$, and $\Sigma_2 = \mathrm{diag}(\texttt{seq}(18,16;5),\allowbreak 0.25, 0.25, \ldots)$.
\end{itemize}

 %Figure~\ref{fig:simulation_study} illustrates this effect. The figure is organized into three columns and four rows, where each row corresponds to one of the chosen values of $\delta$ ($0$, $0.5$, $1$, and $1.5$). For each scenario, we display 50 randomly selected observations, with the x-axis representing the 200 ambient dimensions and the y-axis showing the corresponding component values of each observation. The left column displays the results for Scenario~1, the central column corresponds to Scenario~2, whereas the right column corresponds to Scenario~3.

%\begin{figure}[htb]
 %   \centering
 %   \includegraphics[width=0.95\textwidth]{graficas_simulaciones_delta.png}
 %   \caption{Visualization of the simulated datasets for different values of $\delta$ ($0$, $0.5$, $1$, and $1.5$, from top to bottom). Each row corresponds to one value of $\delta$; the left column shows Scenario~1, the central column shows Scenario~2 and the right column shows Scenario~3.}
 %   \label{fig:simulation_study}
%\end{figure}

The performance of the TCLUST, RLG, and tHDDC methods for scenarios 1 to 3 will be analyzed for different values of the parameter $\delta$. Specifically, the values $\delta$ used are \texttt{seq}($-0.3$,$0.3$;7). For each method and each value of $\delta$, 20 realizations were generated, and these three approaches were applied to the resulting data sets. %Table~\ref{tab:results} summarizes the results, reporting the average execution time and mean accuracy obtained for each method.
The trimming proportion and the number of groups were set at the true values used in the synthetic data generation for all three analyzed methods, namely $\alpha=0.05$ and $G=2$. When applying tHHDC, we considered four settings with respect to the  determination of the $q_g$ parameters: the true intrinsic dimensions are provided, and the Cattell's method is applied for determining them with $\texttt{thresh} = 0.3$, and initial intrinsic dimensions $q_{\text{ini}}$ being smaller,  of the same order, and larger than the true ones. In this way, the true intrinsic dimensions are first directly specified for Scenario~1, with $q_1 = 3$ and $q_2 = 1$, and for Scenarios~2 and 3, with $q_1 = 10$ and $q_2 = 5$. On the other hand, initial intrinsic dimensions for applying Cattell's method being $q_{\text{ini}} = 1$, 3 and 5 are used for Scenario 1, while $q_{\text{ini}} = 1$, 5 and 15 are used for Scenarios 2 and 3.
%On the other hand, $q_{\text{ini}} = 1$, representing smaller dimensions than the original ones; $q_{\text{ini}} = 3$, corresponding to dimensions of the order of the original ones; and $q_{\text{ini}} = 5$, representing larger dimensions are used for Scenario 1. For Scenarios 2 and 3, $q_{\text{ini}} = 1$, representing smaller intrinsic dimensions than the original ones; $q_{\text{ini}} = 5$, corresponding to dimensions of the order of the original ones; and $q_{\text{ini}} = 15$, representing larger dimensions are used.
For RLG, the true dimensions were always used to initialize the intrinsic dimensions of each subspace. For the tHDDC and TCLUST methods, the constants involved in the constraints on the eigenvalues were set as $c_1 = 5$ and $c_2 = 3$ for tHDDC. These values ensure that the constraints on the eigenvalues are satisfied in the three considered scenarios. On the other hand, $c = 12$ for TCLUST has been used, which corresponds to the default value in the \texttt{tclust} package. It is worth noting that larger values of $c$ for TCLUST (such as $c=150$) were also tested to cover the actual eigenvalue ratio, but no performance improvements were observed. Regarding the remaining algorithmic parameters required by the tHHDC, TCLUST, and RLG, we set  for all three methods the same values $\texttt{nstart} = 250$  and $\texttt{nkeep} =5$ and, as discussed in Remark \ref{re3}, the concentration steps are divided into two stages: first, $\texttt{cstep1} = 2$ initial concentration steps are applied to all initializations, and subsequently $\texttt{cstep2} = 25$ final concentration steps are performed on the $\texttt{nkeep} =5$ most promising ones. 

Figure~\ref{grafico1} reports the results of the simulation study for the scenarios considered. The boxplots display the clustering accuracy, where the set of trimmed observations is considered as an extra cluster for comparison purposes (take also into account that the same fixed trimming level is considered for the compared methods). %\textcolor{red}{This ensures that the metric reflects the proportion of correctly assigned observations, including the successful detection of noise as a ``trimmed'' category. } %The figure displays boxplots of the accuracy\textcolor{red}{, defined as the proportion of correctly classified observations. In this calculation, trimmed observations are treated as an additional distinct group, thus accounting for the correct identification of both clustered and outlying data points across the three competing methods.}(define as the proportion of correctly classified-assigned observations, including both assigned and trimmed observations \la{[No me queda claro???]}) obtained by the three competing methods. 
Regarding the model parameters, the true intrinsic dimensions are initially provided for both tHDDC and RLG. However, their subsequent handling differs between the two methods: for tHDDC, the algorithm is executed with \texttt{Cattell=TRUE}, meaning that these intrinsic dimensions are dynamically re-estimated at each iteration. In contrast, the dimensions remain fixed for RLG at their true initial values throughout the entire iterative process. %For tHDDC and RLG, the intrinsic dimensions are initially provided at their true values. %Results are shown for different values of the separation parameter $\delta$.

\begin{figure}[htbp]
\centering

\begin{subfigure}{0.48\linewidth}
  \centering
  \includegraphics[width=\linewidth]{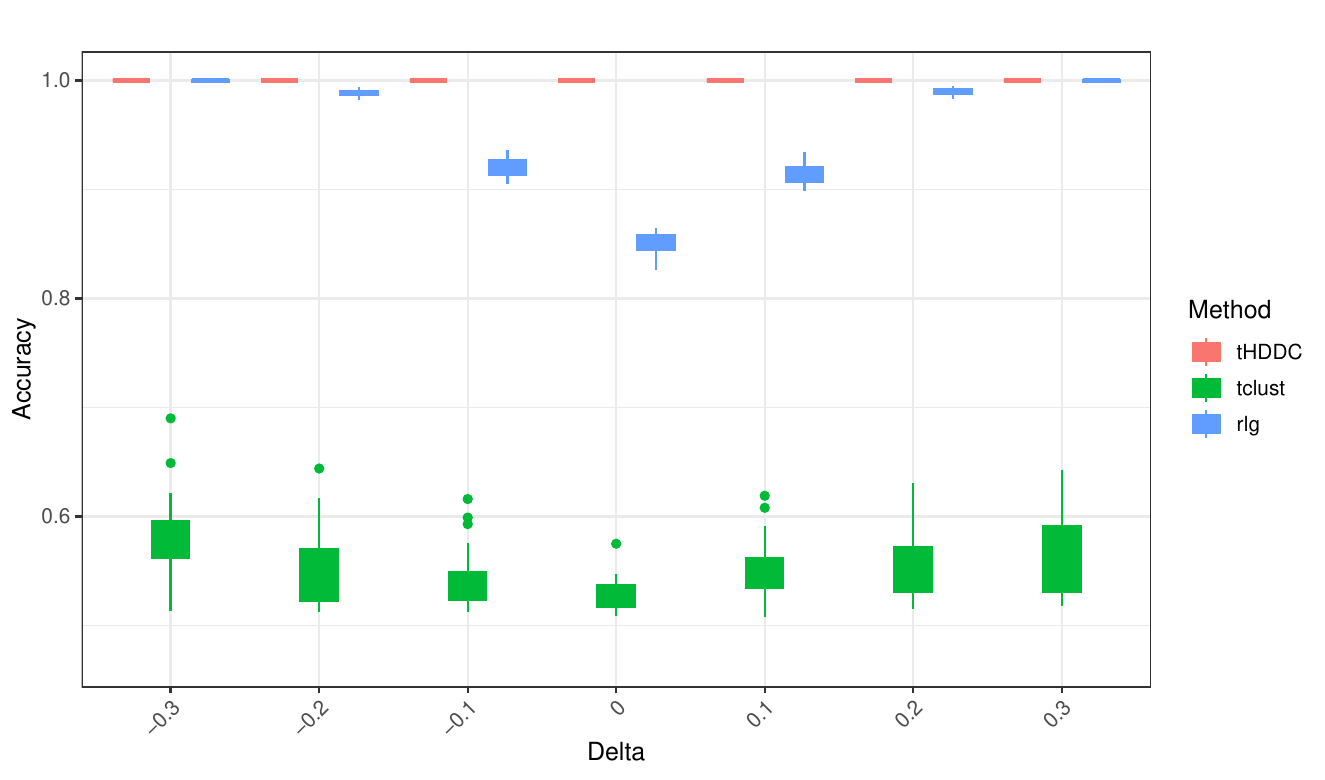}
  \caption{}
  \label{grafico1_scen1}
\end{subfigure}
\hfill
\begin{subfigure}{0.48\linewidth}
  \centering
  \includegraphics[width=\linewidth]{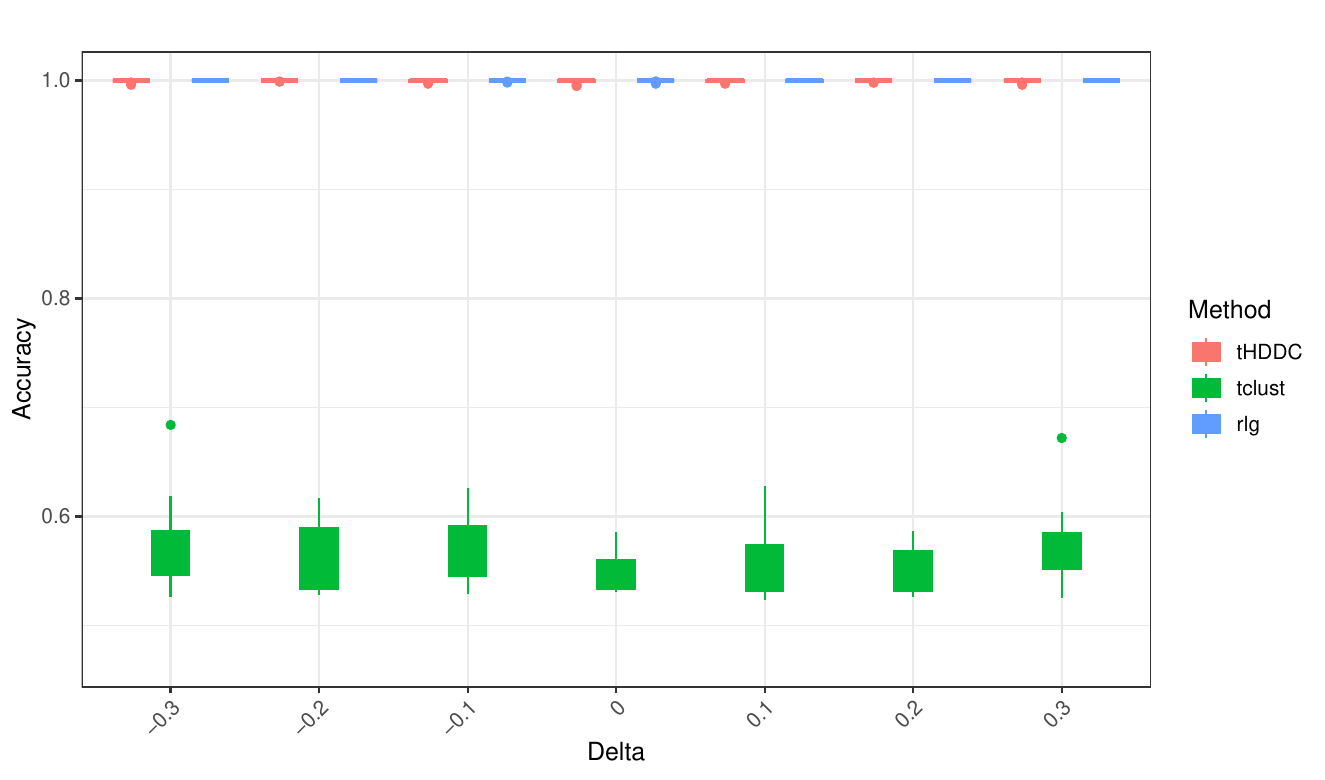}
  \caption{}
  \label{grafico1_scen2}
\end{subfigure}

\begin{subfigure}{0.48\linewidth}
  \centering
  \includegraphics[width=\linewidth]{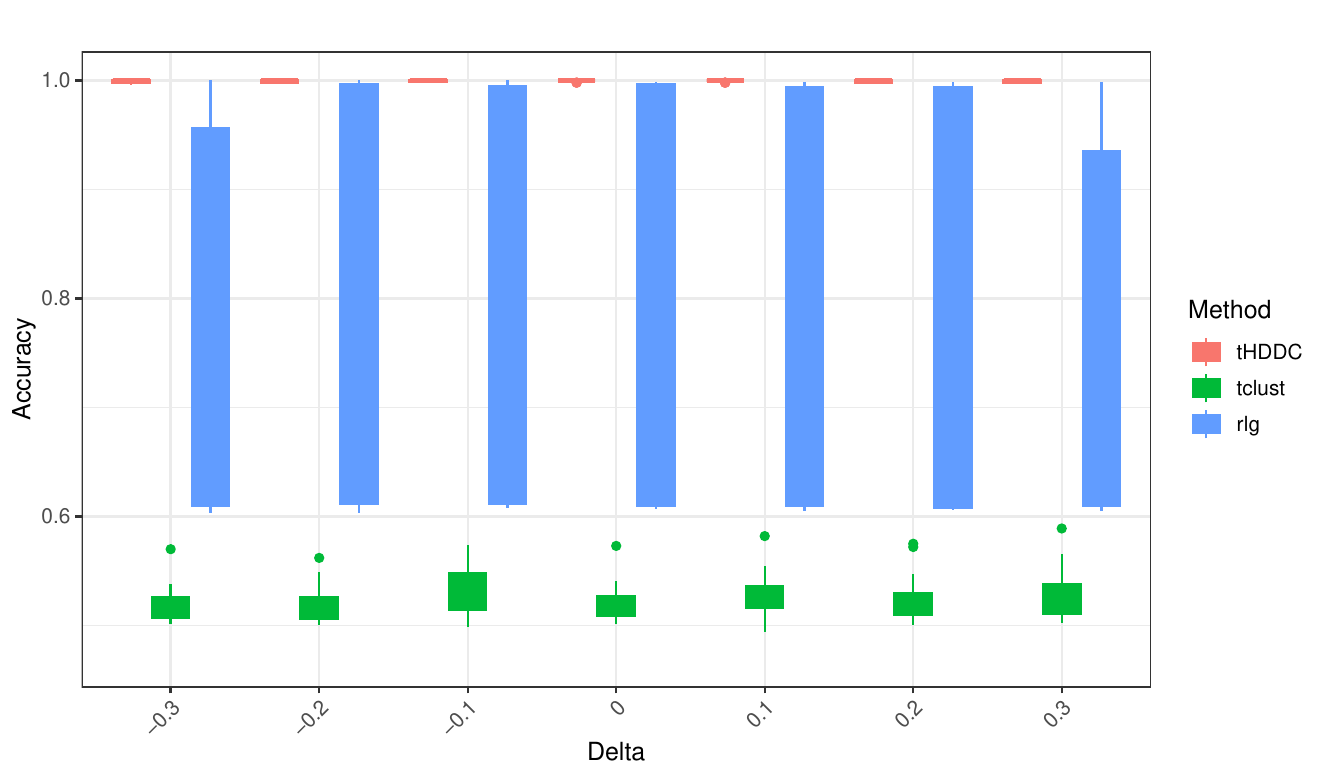}
  \caption{}
  \label{grafico1_scen3}
\end{subfigure}

\caption{Boxplots of the accuracy obtained by TCLUST, RLG, and tHDDC 
for different values of $\delta$, and where the true intrinsic dimensions are initially provided for RLG and tHDDC (fixed in the case of RLG and further updated in the case of tHDDC). The results for Scenario~1 are shown in (a), for Scenario~2 in (b), and Scenario~3  in (c).}
\label{grafico1}
\end{figure}

For Scenario~1 (Figure~\ref{grafico1_scen1}), the proposed tHDDC method achieves nearly perfect accuracy across all values of the separation parameter $\delta$. In contrast, the performance of the RLG method strongly depends on the separation between the group means. As $\delta$ approaches zero, that is, as the overlap between the two Gaussian distributions increases, the accuracy deteriorates. %Specifically, RLG attains perfect accuracy for $\delta=\pm0.3$, values close to one for $\delta=\pm0.2$, around $0.9$ for $\delta=\pm0.1$, and approximately $0.85$ when the two groups share identical mean vectors.
The TCLUST method performs considerably worse in this scenario. %, with accuracy values typically ranging between $0.5$ and $0.6$. Nevertheless, a decreasing trend in performance as $\delta$ approaches zero is also observed for TCLUST, similarly to RLG.
The tHDDC method again achieves near-perfect classification accuracy for Scenario~2 (Figure~\ref{grafico1_scen2}) and all values of the separation parameter $\delta$. The RLG method also exhibits almost perfect accuracy across all values of $\delta$, which is likely due to the reduced overlap between the underlying subspaces, as the groups do not share the same orthogonal matrix defining their orientations. As in the previous scenario, TCLUST continues to yield unsatisfactory results for all values of $\delta$. Finally, in Scenario~3 (Figure~\ref{grafico1_scen3}), tHDDC maintains again excellent performance across all values of  $\delta$. In fact, this scenario highlights the limitations of the RLG approach when the underlying subspaces associated with the clusters intersect. In this case, the boxplots display substantial variability, indicating a wide dispersion of accuracy values across repetitions, with average performance consistently below $0.9$. Once again, TCLUST exhibits poor results.

We have also analyzed the behavior of the proposed tHDDC method under different choices for the initial values of  $q_{\text{ini}}$ in the automatized Cattell's based procedure for estimating the intrinsic dimensions presented in Subsection \ref{catell} with $\texttt{thresh}=0.3$. Figure~\ref{grafico2} displays the corresponding boxplots for accuracy in the three considered scenarios.

\begin{figure}[htbp]
\centering

\begin{subfigure}{0.48\linewidth}
  \centering
  \includegraphics[width=\linewidth]{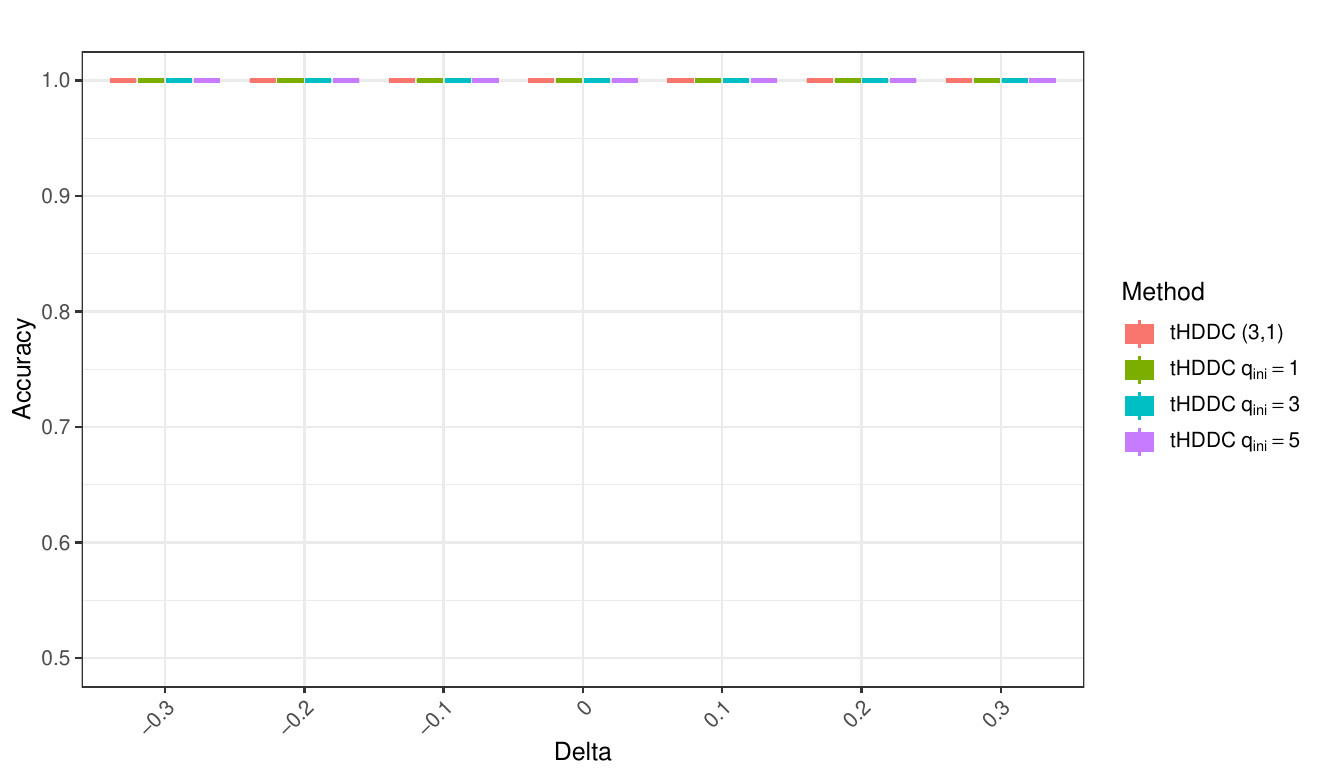}
  \caption{}
  \label{grafico2_scen1}
\end{subfigure}
\hfill
\begin{subfigure}{0.48\linewidth}
  \centering
  \includegraphics[width=\linewidth]{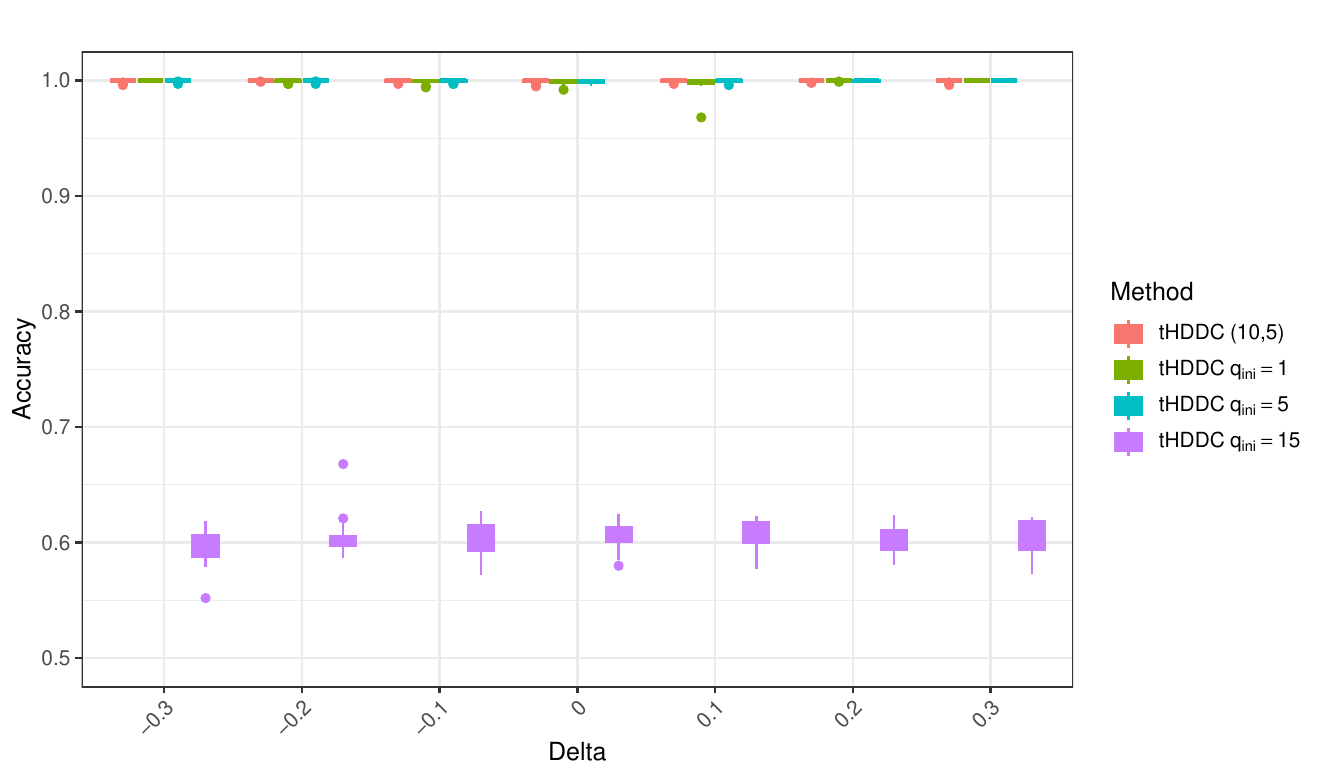}
  \caption{}
  \label{grafico2_scen2}
\end{subfigure}

\begin{subfigure}{0.48\linewidth}
  \centering
  \includegraphics[width=\linewidth]{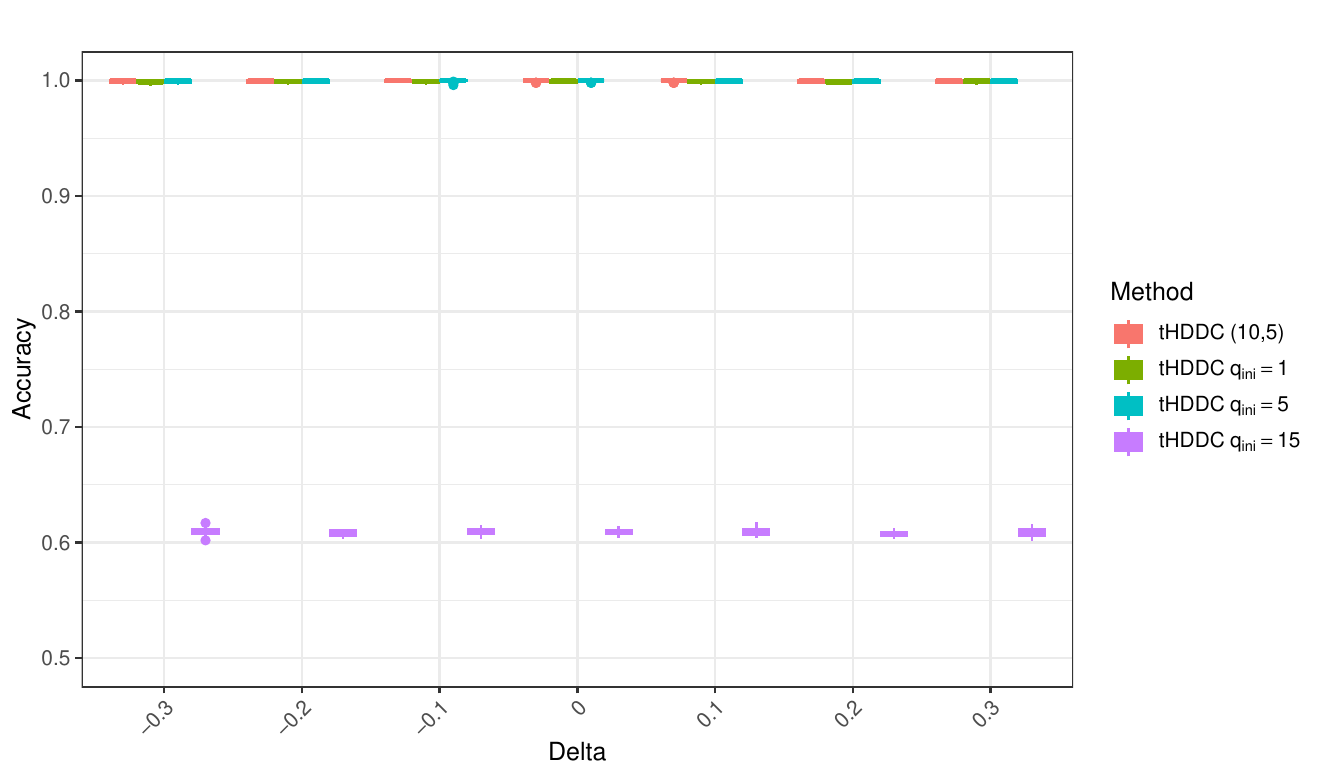}
  \caption{}
  \label{grafico2_scen3}
\end{subfigure}

   \caption{Boxplots of the classification accuracy of the tHDDC method when $q_1$ and $q_2$ are equal to true values of the intrinsic dimensions, and also when using different values  $q_{\text{ini}}$ for their  automatized determination by using the Cattell's approach. The results for Scenario~1 are shown in (a), for Scenario~2 in (b), and Scenario~3  in (c).}
\label{grafico2}
\end{figure}

For Scenario 1 (Figure~\ref{grafico2_scen1}), we observe almost perfect accuracy values for all $q_{\text{ini}}$ values analyzed. Although for Scenario~2 (Figure~\ref{grafico2_scen2}) and Scenario~3 (Figure~\ref{grafico2_scen3}), the tHDDC method achieves near-perfect classification accuracy for all values of $\delta$, the performance deteriorates when the intrinsic dimensions are initialized with large values ($q_{\text{ini}} = 15$). In this case, accuracy drops to approximately $0.6$–$0.65$, which can be attributed to the increased difficulty of identifying suitable initial subsets when a large number of observations must be selected to initialize parameters when $q_{\text{ini}} = 15$. Note that two randomly selected subsets with $15+2$ observations each, in this case, should ideally be selected from each cluster and correctly partitioned. For the remaining initialization choices, the results are remarkably stable and highly accurate. In all cases, the boxplots for the accuracies are concentrated in the interval $[0.95, 1]$.%, with only a few outlying points corresponding to very small misclassifications rates. %In particular, the worst observed accuracies are approximately $0.98$ and $0.993$, corresponding to only 20 and 7 misclassified observations, respectively, out of a total of 1000.
As mentioned, we have also analyzed other values of the constant $c$ for the eigenvalue-ratio constraint in (\ref{eigenavalus_ratio}) when applying TCLUST, including a larger value $c = 150$, which satisfies the true eigenvalue ratio in the data generation scheme. The results do not improve; in general, the accuracy rate never exceeds 0.6.

%\begin{table}[ht]
%\centering
%\begin{tabular}{c ccccccc}

%$\delta$ &
%$-0.3$ & $-0.2$ & $-0.1$ & $0.0$ & $0.1$ & $0.2$ & $0.3$ \\

%Ratio 1 &
%2.903551 & 2.841197 & 2.892752 & 2.913823 & 2.865148 & 2.850567 & 2.857786 \\
%Ratio 2 &
%2.719323 & 2.739667 & 2.761210 & 2.789511 & 2.725773 & 2.756563 & 2.708481 \\
%Ratio 3 &
%2.529054 & 2.508126 & 2.581753 & 2.581341 & 2.551942 & 2.505011 & 2.509095 \\

%\end{tabular}
%\caption{Ratio de tiempos para los tres escenarios en función de $\delta$.}
%\label{tab:ratio_escenarios}
%\end{table}

%We now turn to the analysis of the computational cost of the different methods. All simulations were carried out on a high-performance computing server equipped with \textcolor{blue}{four Intel(R) Xeon(R) Gold 6148 processors running at 2.40\,GHz, providing a total of 160 hardware threads. The system is endowed with 1~TB of RAM and 6~TB of available disk space in the \texttt{/home} directory shared among users. In addition, the server includes an NVIDIA RTX A6000 GPU with 48~GB of GDDR6 memory, based on the NVIDIA Ampere architecture, featuring 10.752 CUDA cores, 336 third-generation Tensor Cores, and 84 second-generation RT Cores}. All algorithms were executed using 4 threads and 20 workers.

%\subsection{Computational time analysis}

Regarding computing times, all simulations were carried out on a %high-performance computing
server equipped with four Intel Xeon Gold 6148 processors (2.4\,GHz, 160 hardware threads) and 1~TB of RAM. In addition, the server includes an NVIDIA RTX A6000 GPU with 48~GB of GDDR6 memory. %All algorithms were executed using 4 threads and 20 workers.
We focus our comparison on TCLUST and tHHDC, comparing the execution times of tHDDC (for different values of $q_{\text{ini}}$) with those of TCLUST. We do not include RLG in this computational time comparison, since when the intrinsic dimensions are fixed, its computational cost is comparable to that of tHDDC. Indeed, both methods require the computation of similar quantities such as eigenvalues, eigenvectors, and distance measures. In fact, if Cattell’s method were also applied to estimate the intrinsic dimension in RLG, the associated computational burden would be similar to that of tHDDC, as it would involve analogous steps. Moreover, regarding that comparison of computational times, it is important to note that TCLUST is available through the \texttt{tclust} \textsf{R} package, in which part of the implementation is written in \texttt{C++}. In contrast, our proposed tHDDC method is currently implemented entirely in \textsf{R}, although a \texttt{C++} implementation is planned for future work. As a consequence, a direct comparison of absolute execution times between these methods is not entirely fair. Therefore, we are going to base our comparison in a similar purely \textsf{R} implementation of both TCLUST and tHHDC.

Table~\ref{ratio_escenarios} reports the ratio between the computation time of TCLUST and that of tHDDC (initialized with the true intrinsic dimensions) for different values of the separation parameter $\delta$ and for the three previously introduced simulation scenarios.

\begin{table}[ht]
\centering
\begin{tabular}{|c|c|c|c|c|c|c|c|}
\hline
$\delta$ &
$-0.3$ & $-0.2$ & $-0.1$ & $0.0$ & $0.1$ & $0.2$ & $0.3$ \\
\hline
Scenario 1 &
2.923 & 2.944 & 2.898 & 2.922 & 2.958 & 2.876 & 2.941 \\

Scenario 2 &
2.663 &  2.628 & 2.640 & 2.640 & 2.662 & 2.679 & 2.670 \\

Scenario 3 &
2.511 & 2.480 & 2.522 & 2.460 & 2.536 & 2.483 & 2.476 \\
\hline
\end{tabular}
\caption{Ratio between the computation times of TCLUST and tHDDC for the three scenarios and for each value of the separation parameter $\delta$.}
\label{ratio_escenarios}
\end{table}

%These results show that the computational cost of TCLUST is consistently higher than that of tHDDC. Specifically, the ratio ranges approximately between $2.87$ and $2.95$ for Scenario~1, between $2.62$ and $2.66$ for Scenario~2, and between $2.45$ and $2.53$ for Scenario~3. The highest ratios occur for $\delta = 0$, corresponding to the most challenging clustering setting due to maximum overlap between groups. This behavior is consistent with the results shown in Figure~\ref{grafico1}, where the classification performance of TCLUST also deteriorates as the overlap increases. Overall, these results indicate that tHDDC is between 2.5 and 3 times faster than TCLUST under comparable conditions.

These results show that the computational cost of TCLUST is consistently higher than that of tHDDC across all scenarios and values of the separation parameter $\delta$. %The ratios remain relatively stable, ranging approximately between $2.87$ and $2.95$ for Scenario~1, between $2.62$ and $2.68$ for Scenario~2, and between $2.46$ and $2.54$ for Scenario~3.
Although minor fluctuations are observed with different values of $\delta$, no clear pattern emerges, indicating that the relative computational advantage of tHDDC is largely independent of the degree of separation between groups. Overall, tHDDC is roughly 2.5 to 3 times faster than TCLUST under comparable conditions.%, with the largest gains observed in Scenario~1, which involves lower-dimensional subspaces.

We finally analyze the computational cost of tHDDC under different values of the initial intrinsic dimension $q_{\text{ini}}$ for the Cattell's procedure described in Subsection \ref{catell}. Figure \ref{grafico3} displays the boxplots of the execution times.% for Scenario~1 in Figure \ref{grafico3_scen1}, Scenario~2 in Figure \ref{grafico3_scen2}, and Scenario~3 in Figure \ref{grafico3_scen3}.

\begin{figure}[htbp]
\centering

\begin{subfigure}{0.48\linewidth}
  \centering
  \includegraphics[width=\linewidth]{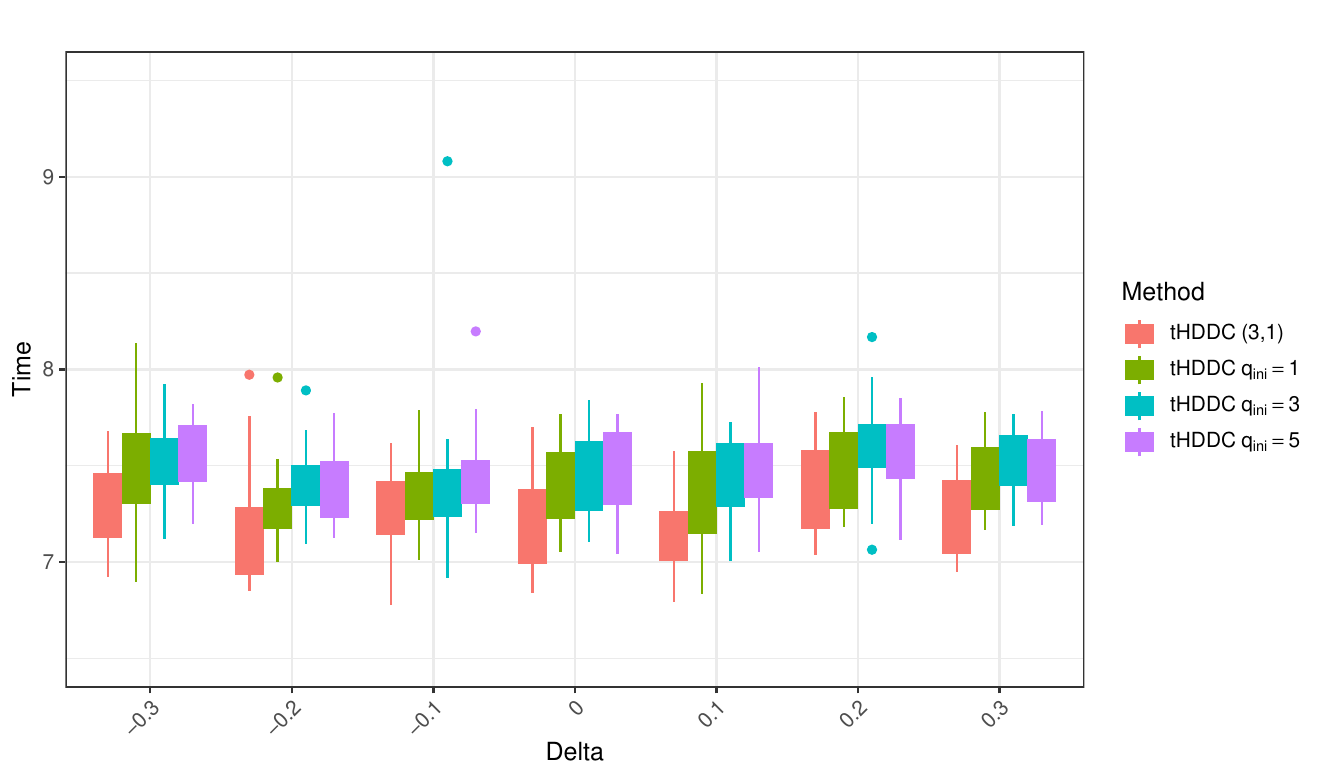}
  \caption{}
  \label{grafico3_scen1}
\end{subfigure}
\hfill
\begin{subfigure}{0.48\linewidth}
  \centering
  \includegraphics[width=\linewidth]{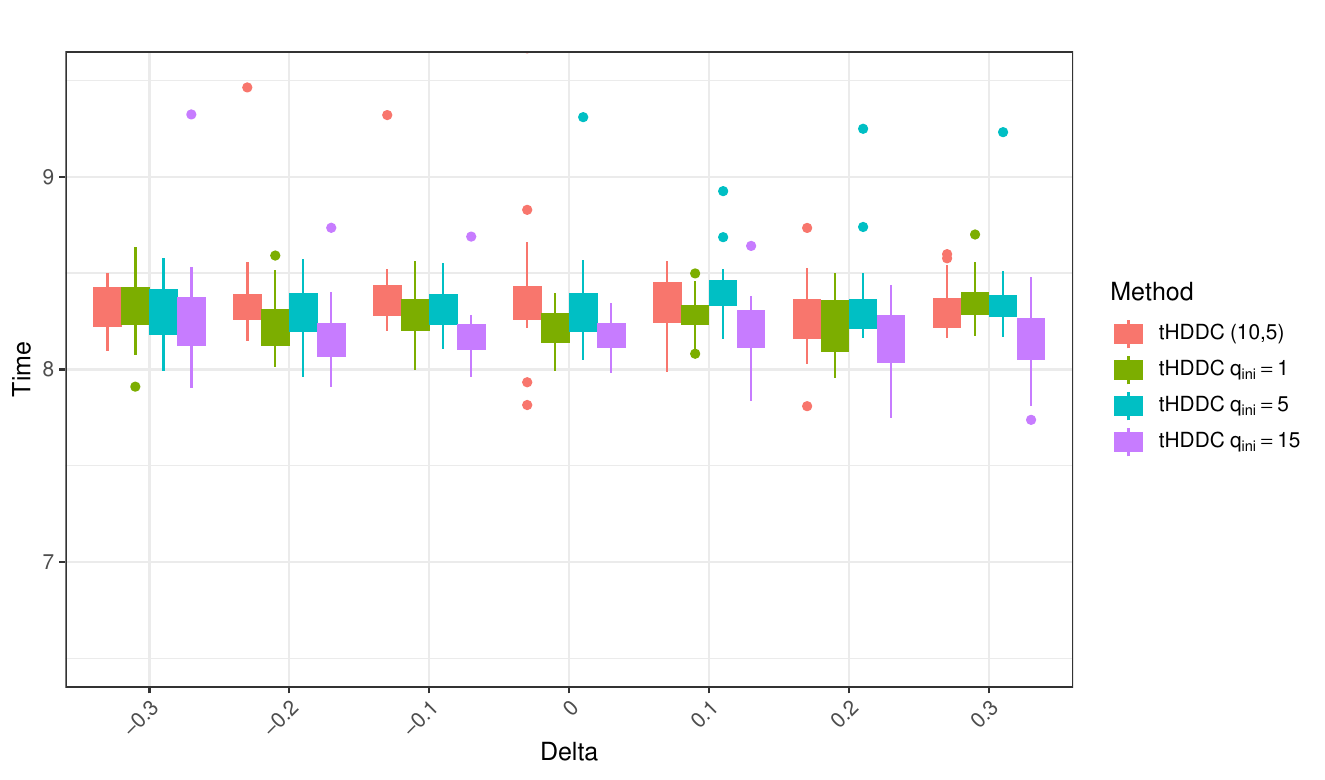}
  \caption{}
  \label{grafico3_scen2}
\end{subfigure}

\begin{subfigure}{0.48\linewidth}
  \centering
  \includegraphics[width=\linewidth]{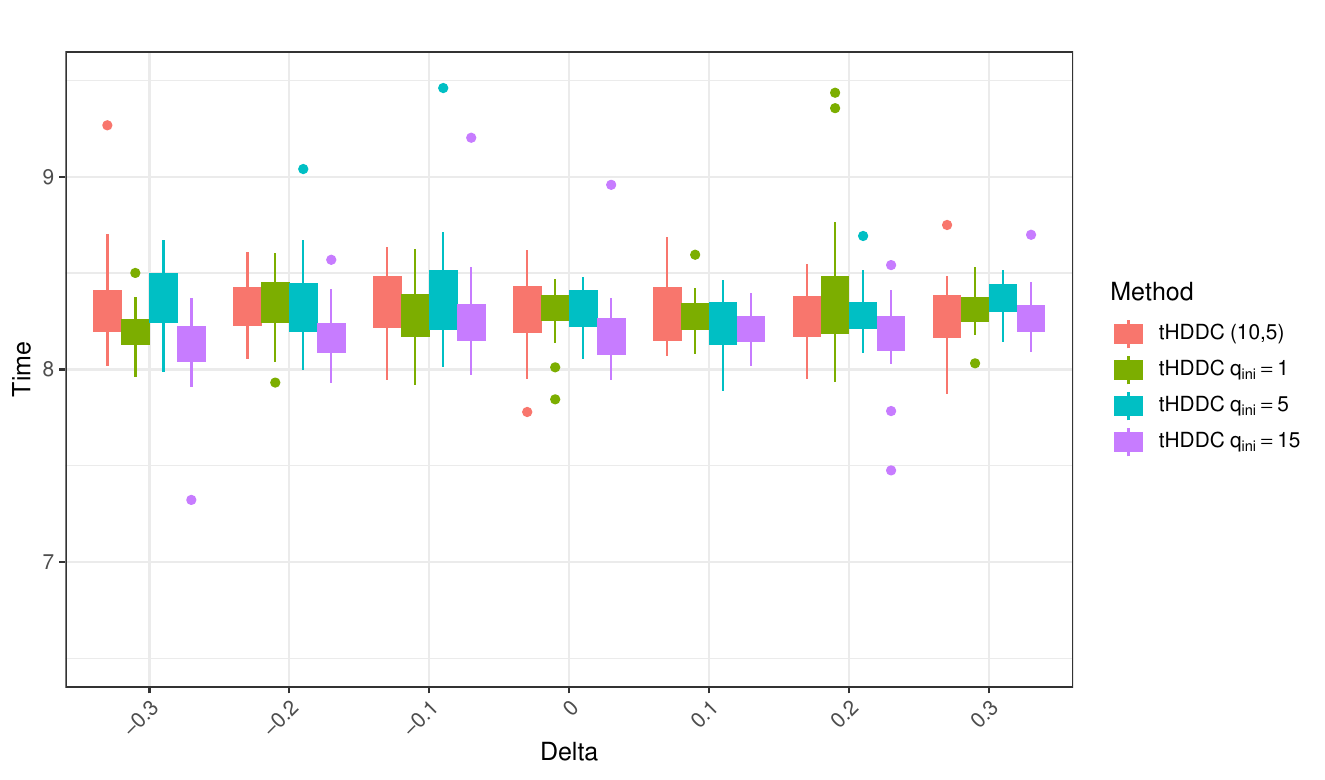}
  \caption{}
  \label{grafico3_scen3}
\end{subfigure}

\caption{Boxplots of the computation time for tHDDC different initial $q_{\text{ini}}$ values for automated determination of the intrinsic dimensions, for each value of $\delta$. The results for Scenario~1 are shown in (a), for Scenario~2 in (b), and Scenario~3  in (c).}
\label{grafico3}
\end{figure}

In Scenario~1 (Figure~\ref{grafico3_scen1}), which represents the simplest setting, the tHDDC method initialized with the true intrinsic dimensions achieves the lowest median computation time. The initialization with $q_{\text{ini}} = 1$ leads to slightly higher computation times, while larger initial values result in increased computational cost. This behavior is expected, since larger initial dimensions require the computation of a greater number of eigenvalues and eigenvectors, particularly during the initialization stage of the clustering algorithm. For Scenario~2 (Figure~\ref{grafico3_scen2}) and Scenario~3 (Figure~\ref{grafico3_scen3}), the differences in computational time among the various initializations are less pronounced. In these cases, the execution times are more homogeneous across different values of $q_{\text{ini}}$. Interestingly, the initialization with $q_{\text{ini}} = 15$ most of the times yields lower computation times than the other initializations. This effect can be explained by the fact that, for $q_{\text{ini}} = 15$, the resulting clustering quality is poorer, leading to faster convergence due to less demanding optimization steps.

%\begin{figure}[htb]
%    \centering
%    \begin{minipage}[b]{0.47\textwidth}
%        \centering
%        \includegraphics[scale=0.4]{tclust_scen1.png}
%    \end{minipage}
%    \hfill
%    \begin{minipage}[b]{0.47\textwidth}
%        \centering
%        \includegraphics[scale=0.4]{tclust_scen2.png}
%    \end{minipage}
   
%    \begin{minipage}[b]{0.65\textwidth}
%        \centering
%        \includegraphics[scale=0.4]{tclust_scen3.png}
%    \end{minipage}
%    \caption{tclust}
%    \label{grafico_tclust}
%\end{figure}

Finally, we also take advantage of this simulation study to further assess the performance of the Cattell-based procedure used by tHDDC to estimate the intrinsic dimensions of the cluster-specific subspaces. For each scenario and each value of the separation parameter $\delta$, we analyze the estimated dimensions obtained at convergence of the algorithm. Figures~\ref{grafico_q1} and~\ref{grafico_q2} summarize these results by reporting boxplots of the estimated intrinsic dimensions for each scenario and each value of $\delta$. Figure~\ref{grafico_q1} displays boxplots of the estimated intrinsic dimensions corresponding to the larger intrinsic dimension in each scenario, whereas Figure~\ref{grafico_q2} shows those for the smaller ones. %For ease of interpretation, the vertical axes are labeled according to the true intrinsic dimensions in each scenario: in Scenario~1, the larger and smaller dimensions are $3$ and $1$, respectively, while in Scenarios~2 and~3 they are $10$ and $5$. Boxplots centered around these reference values indicate accurate estimation, whereas values predominantly above or below correspond to overestimation or underestimation, respectively.
\begin{figure}[htbp]
\centering

\begin{subfigure}{0.48\linewidth}
  \centering
  \includegraphics[width=0.95\linewidth]{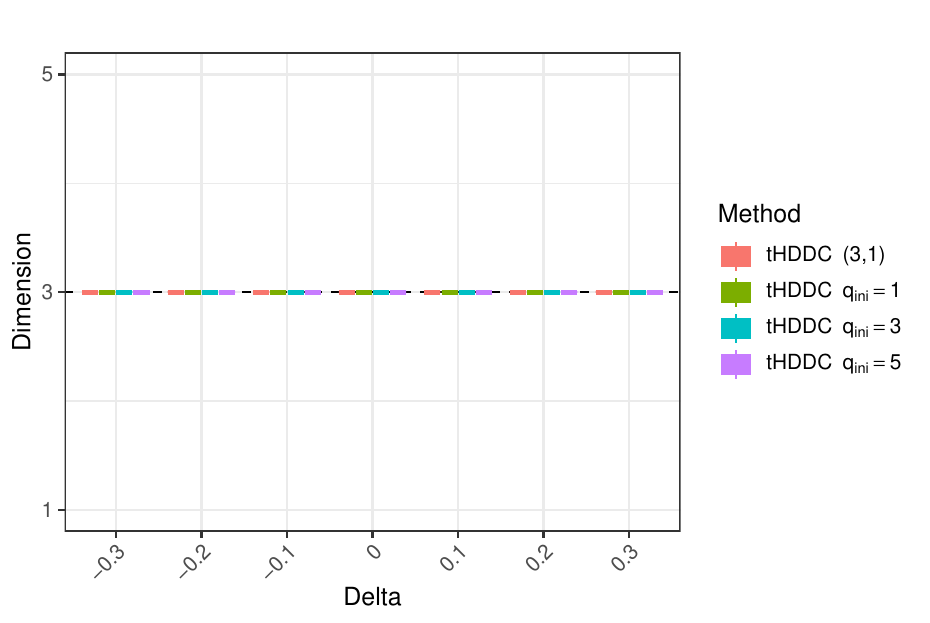}
  \caption{}
  \label{q1_scen1}
\end{subfigure}
\hfill
\begin{subfigure}{0.48\linewidth}
  \centering
  \includegraphics[width=0.95\linewidth]{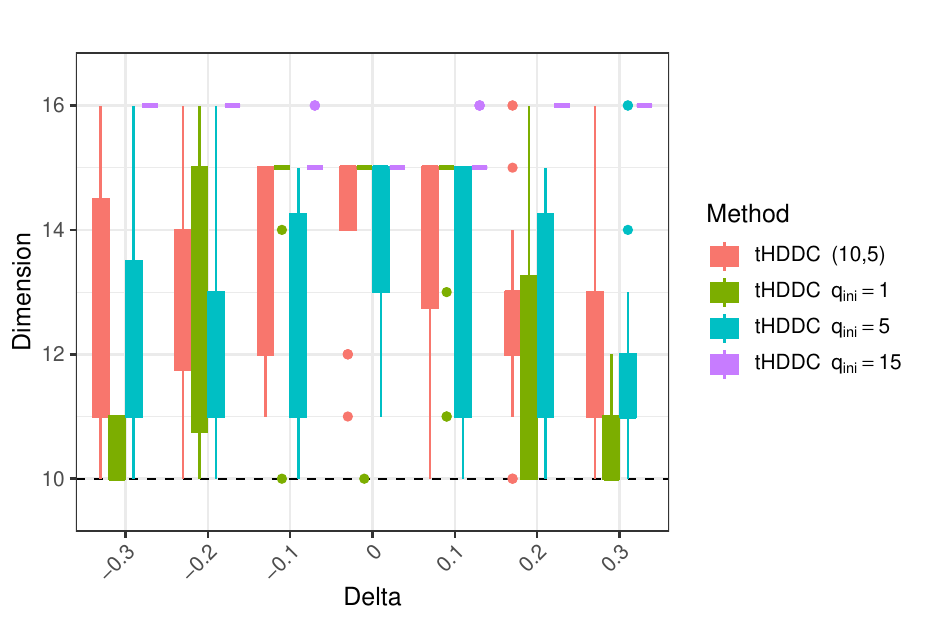}
  \caption{}
  \label{q1_scen2}
\end{subfigure}

\begin{subfigure}{0.48\linewidth}
  \centering
  \includegraphics[width=0.95\linewidth]{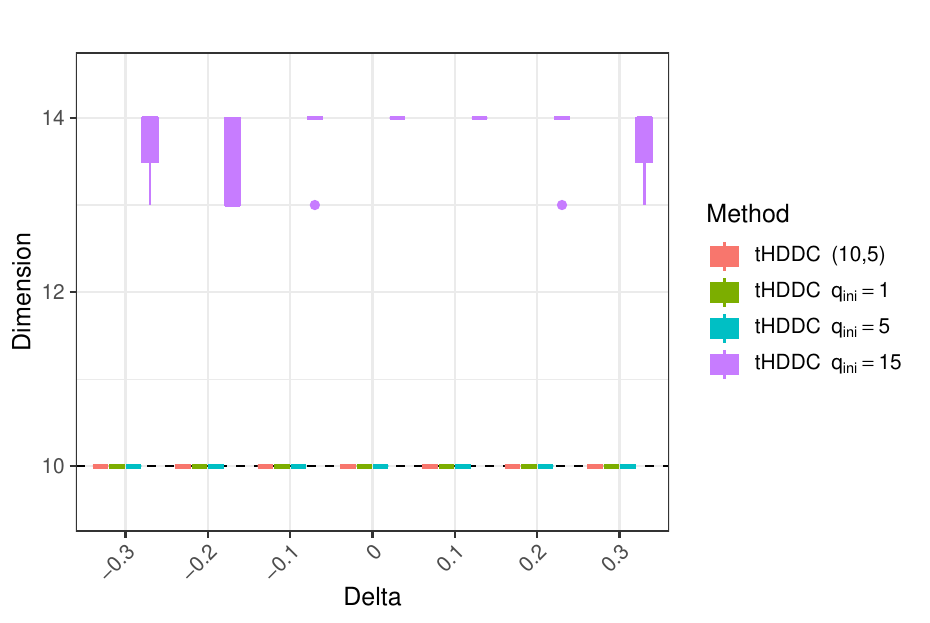}
  \caption{}
  \label{q1_scen3}
\end{subfigure}

\caption{Boxplots of the estimated larger intrinsic dimension when using the Cattell-based procedure. The results for Scenario~1 are shown in (a), for Scenario~2 in (b), and Scenario~3  in (c).  The horizontal dashed line indicates the true value to be estimated.}
\label{grafico_q1}
\end{figure}

Across all three scenarios, the intrinsic dimensions are accurately recovered in the majority of replications. In Scenario~1, where the true dimensions are small ($q_1 = 3$ and $q_2 = 1$), the Cattell-based criterion consistently identifies the correct dimensionality for both clusters, independently of $\delta$. In Scenario~2, the larger intrinsic dimension tends to be overestimated across the different initializations of $q_{\text{ini}}$, %typically by 1 to 6 extra dimensions
with the worst results observed for the initialization $q_{\text{ini}} = 15$, but the smaller intrinsic dimension  is generally well estimated, except again for $q_{\text{ini}} = 15$. %, where overestimation of 8 to 10 dimensions is observed.
For Scenario~3, the only initialization that produces notably poor estimates for both dimensions  is again $q_{\text{ini}} = 15$. All other initializations provide accurate estimates for both the larger and smaller intrinsic dimensions in this scenario. These comments can be related with previous findings regarding clustering accuracy and computational time: the initialization $q_{\text{ini}} = 15$ yielded very poor accuracy in Scenarios~2 and~3.%, despite an apparently lower computation time. %This can be explained by the fact that overestimating the intrinsic dimensions produces clusters that do not match the original ones, leading to faster convergence of the algorithm but substantially degraded clustering performance.

\begin{figure}[htbp]
\centering

\begin{subfigure}{0.48\linewidth}
  \centering
  \includegraphics[width=0.95\linewidth]{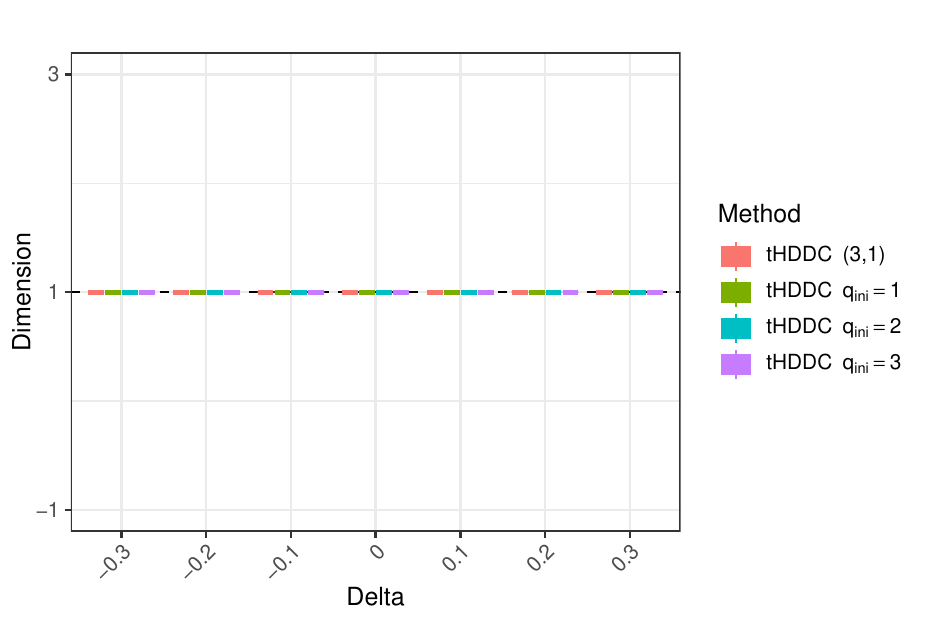}
  \caption{}
  \label{q2_scen1}
\end{subfigure}
\hfill
\begin{subfigure}{0.48\linewidth}
  \centering
  \includegraphics[width=0.95\linewidth]{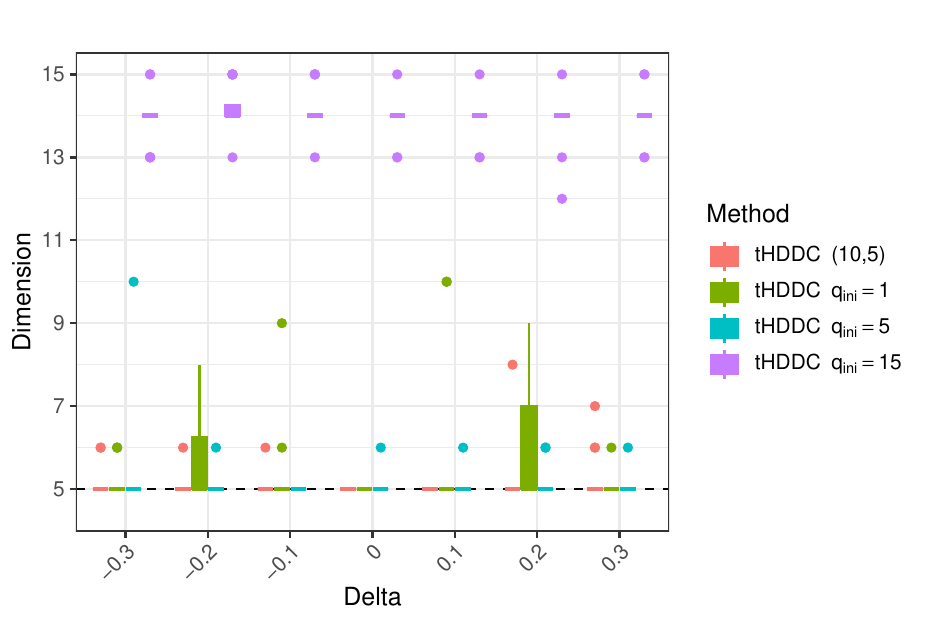}
  \caption{}
  \label{q2_scen2}
\end{subfigure}

\begin{subfigure}{0.48\linewidth}
  \centering
  \includegraphics[width=0.95\linewidth]{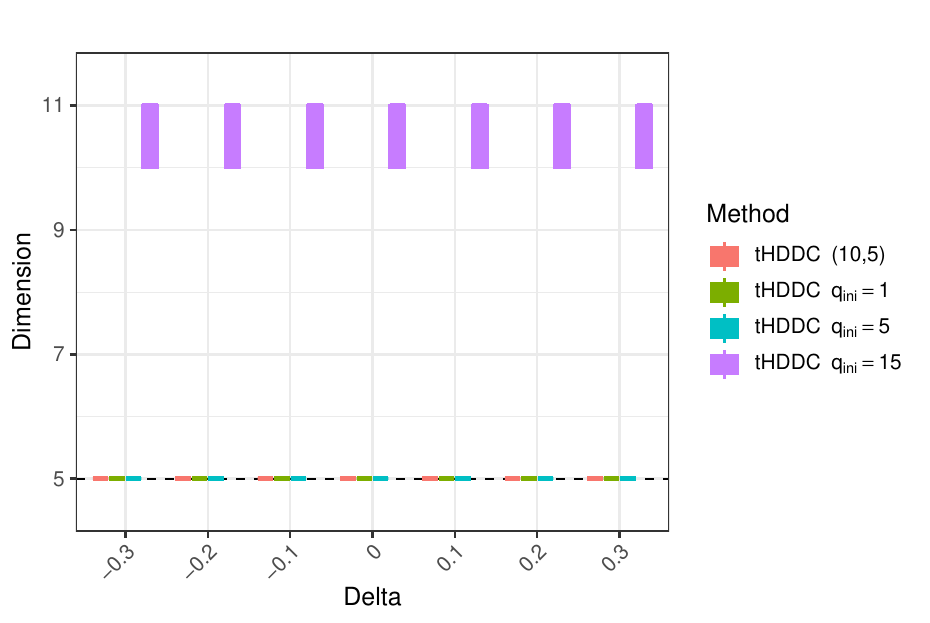}
  \caption{}
  \label{q2_scen3}
\end{subfigure}

\caption{Boxplots of the estimated smaller intrinsic dimension when using the Cattell-based procedure. The results for Scenario~1 are shown in (a), for Scenario~2 in (b), and Scenario~3  in (c).  The horizontal dashed line indicates the true value to be estimated.}
\label{grafico_q2}
\end{figure}

Overall, all the results shown in this section highlight the robustness of the proposed tHDDC method with respect to the choice of initialization, both in terms of classification accuracy and computational efficiency. However, based on the %empirical evidence presented in
results shown in this section, it is perhaps recommended using relatively small values of $q_{\text{ini}}$. Although this choice may result in a modest increase in computational cost, it seems to lead to more stable and accurate results than unnecessarily large $q_{\text{ini}}$ values. %initializations and avoids the accuracy degradation observed when large initial values are used.

\section{Real data example: Digits data}\label{sec_6}

For a real data example, we use a data set containing images of handwritten digits collected from envelopes processed by the United States Postal Service (USPS). Each digit, originally varying in size and orientation, was segmented, standardized, and normalized, and finally represented as a $16 \times 16$ grayscale image, i.e., a 256-dimensional vector. We focus on a subset of digits ``$3$'', ``$5$'', and ``$8$'', yielding $1756$ images, commonly referred to as \texttt{USPS358}, and available in the R package \texttt{MBCbook} \citep{MBCbook}. These digits were chosen because they are typically the most difficult to distinguish in handwriting recognition, due to their more similar shapes.
Moreover, we introduce outliers into the dataset of interest in a controlled manner. Specifically, we select $195$ observations corresponding to some remaining digits (i.e., digits other than ``$3$'', ``$5$'', and ``$8$'') from the \texttt{USPS} dataset available in the R package \texttt{Rdimtools} \citep{Rdimtools}, so that the resulting dataset surely contains at least a fraction $10\%$ of outliers to be trimmed. %In the version of the dataset available in the \texttt{Rdimtools} package.%, digits $6$ and $7$ are not present, and digit $1$ appears in a nonstandard diagonal representation; these digits were therefore excluded from the contamination mechanism. 
In addition, with the aim of incorporating clearly atypical observations, we introduce another $45$ artificially generated outliers. %Since each observation is represented as a vector in $\mathbb{R}^{256}$ (equivalently, a $16\times16$ grayscale image),
These outliers are constructed by defining simple geometric or structural patterns on the corresponding image grid, such as checkerboard patterns, vertical or horizontal stripe patterns, %filled images with nearly constant intensity,
spiral-like shapes, and concentric square frames. %To avoid generating identical observations, small random perturbations are added to the pixel intensities of each pattern, producing five slightly different realizations of each configuration.
Finally, we also include a set of outliers obtained by inverting the grayscale intensities of selected digit images from the original dataset. %(i.e., replacing each pixel value $x$ by $2-x$). This transformation preserves the underlying digit structure while reversing foreground and background intensities, resulting in images that are visually very different from the typical observations in the dataset.
%These additional observations are intended to represent extreme and structurally different patterns, allowing us to assess the robustness of the proposed methodology under severe contamination scenarios.

For this experiment, we employed the TCLUST and tHDDC procedures. In both cases, we set $G=3$, and since the dataset contains $12\%$ artificially added outliers, we have selected a higher trimming level $\alpha=0.2$, as it is also reasonable to expect additional atypical observations in a dataset derived from non-standard annotations of handwritten digits. Our aim at such a high trimming level is to retain some of the ``clearest'' digit representations within each cluster, that can serve as a starting point for further analyses.

For tHDDC, following the insights gained from the simulation study, we initialize the intrinsic dimensions when applying the Cattell's procedure with a small $q_{\text{ini}}=1$ value. We additionally set $q_{\text{max}}=20$ and \texttt{tresh=0.2}. The constants appearing in the eigenvalue constraints were set to $c_1=5$ and $c_2=1.1$, enforcing similarity among the smallest eigenvalues while allowing some flexibility for largest ones. For TCLUST, we set $c=12$, allowing a certain departure from spherical and equally scattered clusters while maintaining protection against spurious solutions. Finally, for both TCLUST and tHDDC (in its trimmed classification likelihood version), we used \texttt{nstart} = 200, \texttt{nkeep} = 5, \texttt{csteps1} = 10, and \texttt{csteps2} = 150.

With these choices, the dimensions estimated using Cattell's method when running tHDDC in a typical run were $q_1=13$, $q_2=14$, and $q_3=20$. The associated confusion matrix is reported in Table \ref{tab:tHDDC}. On the other hand, the confusion matrix obtained when using TCLUST is shown in Table \ref{tab:TCLUST}. 
%The confusion matrices obtained from a representative run \la{[Que es un ``representative run''?]} of the robust procedures are reported in Figure~\ref{fig:confusion_matrices}(a)--(b). The classification errors reported below are computed considering only the observations that survive the trimming step.
Comparing the performance of the two robust procedures, tHDDC achieves a classification error (computed considering only the observations that survive the trimming step) slightly above $7\%$ %(Figure~\ref{fig:confusion_matrices}a),
whereas TCLUST exhibits a substantially larger error exceeding $38\%$. %(Figure~\ref{fig:confusion_matrices}b). 
On the server described above, the computational time required by tHDDC is approximately three times smaller than that of TCLUST.

For comparison, we also consider the HDDC method \citep{bouveyron2007high}, where no trimming is performed, when using the \texttt{HDclassif} package \citep{JSSv046i06} with the input parameters \texttt{K} = 3, \texttt{itermax} = 150, \texttt{algo} = \texttt{EM}, \texttt{d\_select} = ``\texttt{Cattell}'', \texttt{d\_max} = 20, \texttt{threshold} = 0.2, \texttt{init} = \texttt{mini-em}, and \texttt{mini.nb} = (200, 10) in order to match as closely as possible the settings of the previous application of tHDDC (for instance, \texttt{d\_max} = 20 is analogous to setting $q_{\text{max}}=20$, or \texttt{mini.nb} = (200, 10) for \texttt{nstart} = 200, and \texttt{csteps1} = 10), we obtain the confusion matrix reported in Table \ref{tab:HDDC}. We clearly observe that HDDC suffers from the presence of outliers, exhibiting a classification error above $39\%$, since, for instance, it confuses digits ``$3$'' and ``$5$'' %(see the last row of Figure~\ref{fig:confusion_matrices}c)
and forms a cluster containing most of the outliers together. %with a mixture of the three target digits (see the second row of  Figure~\ref{fig:confusion_matrices}c).
This behavior under contamination is not surprising for HDDC, as the method is not designed to be robust to contamination.

%%%% FORMATO TABLE

\begin{table}[ht]
\centering

\begin{subtable}{0.32\textwidth}
\centering
\resizebox{\textwidth}{!}{ % Ajusta la tabla al ancho de la columna si es muy grande
\begin{tabular}{|c|ccccc|}
\hline
 & out1 & out2 & 3 & 5 & 8 \\
\hline
0 & 195 & 45 & 49 & 31 & 79 \\
1 & 0 & 0 & 10 & 0 & \textit{429} \\
2 & 0 & 0 & \textit{529} & 1 & 32 \\
3 & 0 & 0 & 70 & \textit{524} & 2 \\
\hline
\end{tabular}}
\caption{tHDDC}
\label{tab:tHDDC}
\end{subtable}
\hfill % Espacio elástico entre tablas
\begin{subtable}{0.32\textwidth}
\centering
\resizebox{\textwidth}{!}{
\begin{tabular}{|c|ccccc|}
\hline
 & out1 & out2 & 3 & 5 & 8 \\
\hline
0 & 180 & 45 & 40 & 88 & 46 \\
1 & 4 & 0 &307 &  \textit{251} & 29 \\
2 & 5 & 0 & 31 & 43 & \textit{459} \\
3 & 6 & 0 & \textit{280} & 174 & 8 \\
\hline
\end{tabular}}
\caption{TCLUST}
\label{tab:TCLUST}
\end{subtable}
\hfill
\begin{subtable}{0.32\textwidth}
\centering
\resizebox{\textwidth}{!}{
\begin{tabular}{|c|ccccc|}
\hline
 & out1 & out2 & 3 & 5 & 8 \\
\hline
1 & 0 & 0 & 7 & 0 & \textit{369} \\
2 & \textit{193} & 26 & 6 & 71 & 97 \\
3 & 2 & 19 & \textit{645} & \textit{485} & 76 \\
\hline
\end{tabular}}
\caption{HDDC}
\label{tab:HDDC}
\end{subtable}
\caption{Classification tables for the Digits data when using tHDDC in (a), TCLUST in (b), and  HDDC in (c). Columns correspond to the supposedly ``true'' digits (3, 5 and 8), where label ``out1'' denotes the group of outliers consisting of digits different from $3$, $5$, and $8$, while label ``out2'' corresponds to the group of artificially synthetic outliers. Rows represent the clusters detected, with row $0$ corresponding to trimmed observations.}
\label{tab:comparativa}
\end{table}

We also focus on illustrating some graphical tools for exploring the results of applying tHDDC. First, the cluster mean vectors $\mu_g$, for $g=1,2,3$, can be visualized in Figure \ref{cluster_centers}, where some kind of prototypical digits ``$3$'', ``$5$'', and ``$8$'' can be observed in their perhaps most ``standard'' form.
\begin{figure}[htbp]
    \centering
    \includegraphics[width=0.8\textwidth]{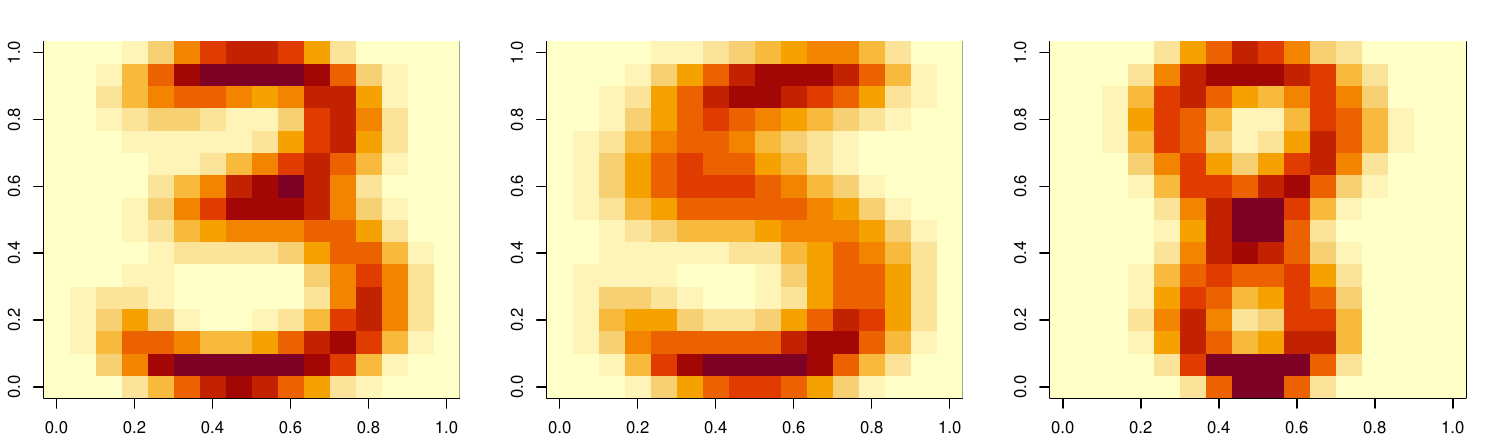}
    \caption{Cluster centers means.}
    \label{cluster_centers}
\end{figure}

%On the other hand, the dimensions estimated using Cattell's method with \texttt{thresh}$=0.2$ and $q_{\text{max}}=20$ were $q_1=13$, $q_2=14$, and $q_3=20$, and consequently
As commented, the estimated intrinsic dimensions when applying tHDDC were $13$, $14$ and $20$, but, in Figure~\ref{loading_vectors}, we only represent the first $8$ eigenvectors obtained for each cluster, commonly referred to as loading vectors. Each loading vector assigns positive weights (orange cells), negative weights (blue cells), or near-zero weights (white cells) to the different pixels of the $16 \times 16$ grid. These eigenvectors capture the main sources of within-class variability and, together with their associated eigenvalues, provide the flexibility needed to correctly classify digits whose handwriting deviates from these most ``standard’’ forms shown in Figure \ref{cluster_centers}. For instance, in the particular case of cluster associated to digit $3$, we observe in Figure~\ref{fig:a} that the model serves to accommodate a range of varied shapes: more rounded versions that touch the borders of the grid both at the top and bottom; more centered and compact shapes; or versions where the upper part of the digit is noticeably smaller. %This illustrates how the intrinsic subspaces determined by tHDDC can be used to describe the heterogeneity observed in the handwritten digits.

\begin{figure}[htbp]
\centering

\begin{subfigure}{0.48\linewidth}
  \centering
  \includegraphics[width=\linewidth]{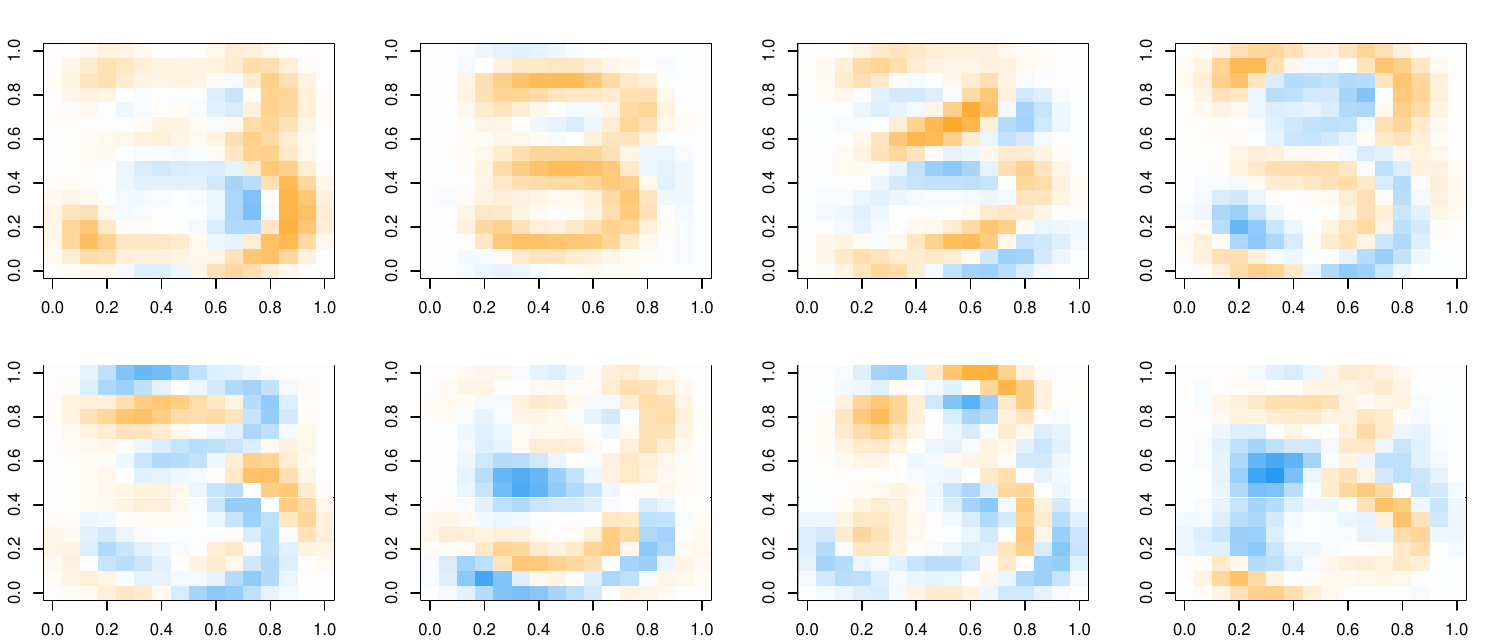}
  \caption{}
  \label{fig:a}
\end{subfigure}
\hfill
\begin{subfigure}{0.48\linewidth}
  \centering
  \includegraphics[width=\linewidth]{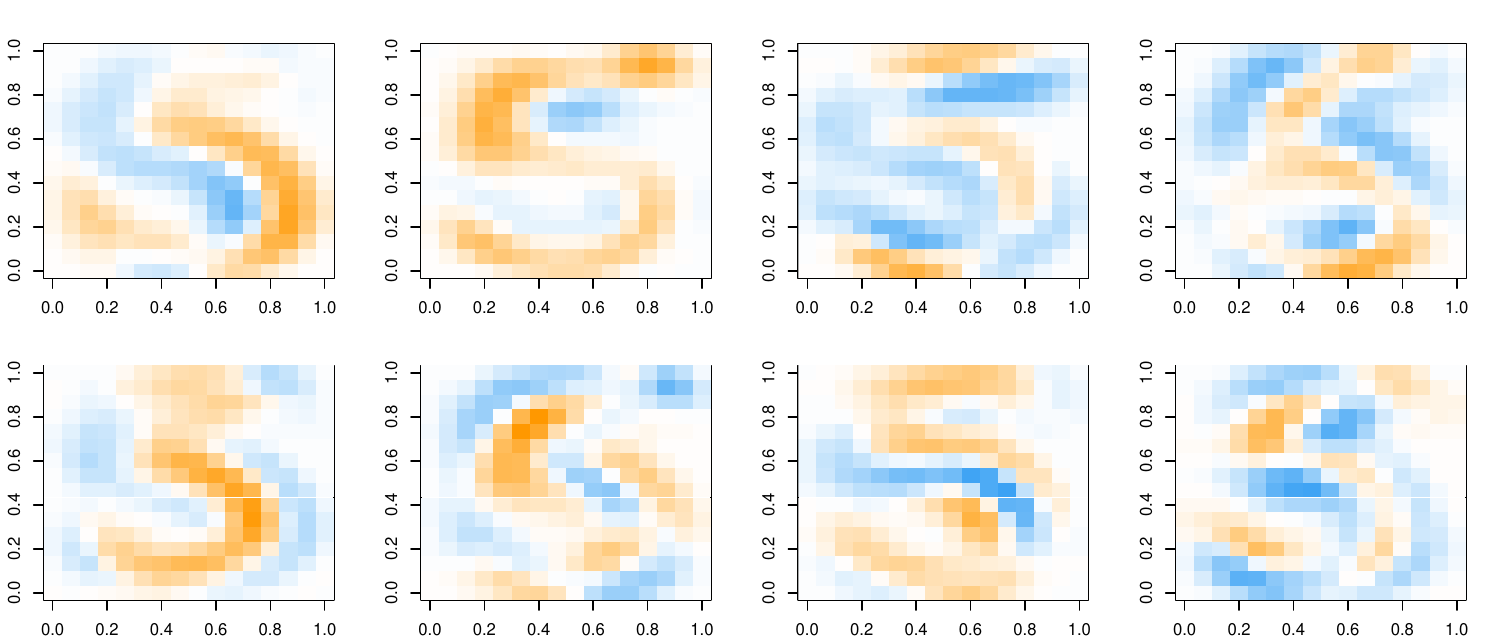}
  \caption{}
  \label{fig:b}
\end{subfigure}

\begin{subfigure}{0.48\linewidth}
  \centering
  \includegraphics[width=\linewidth]{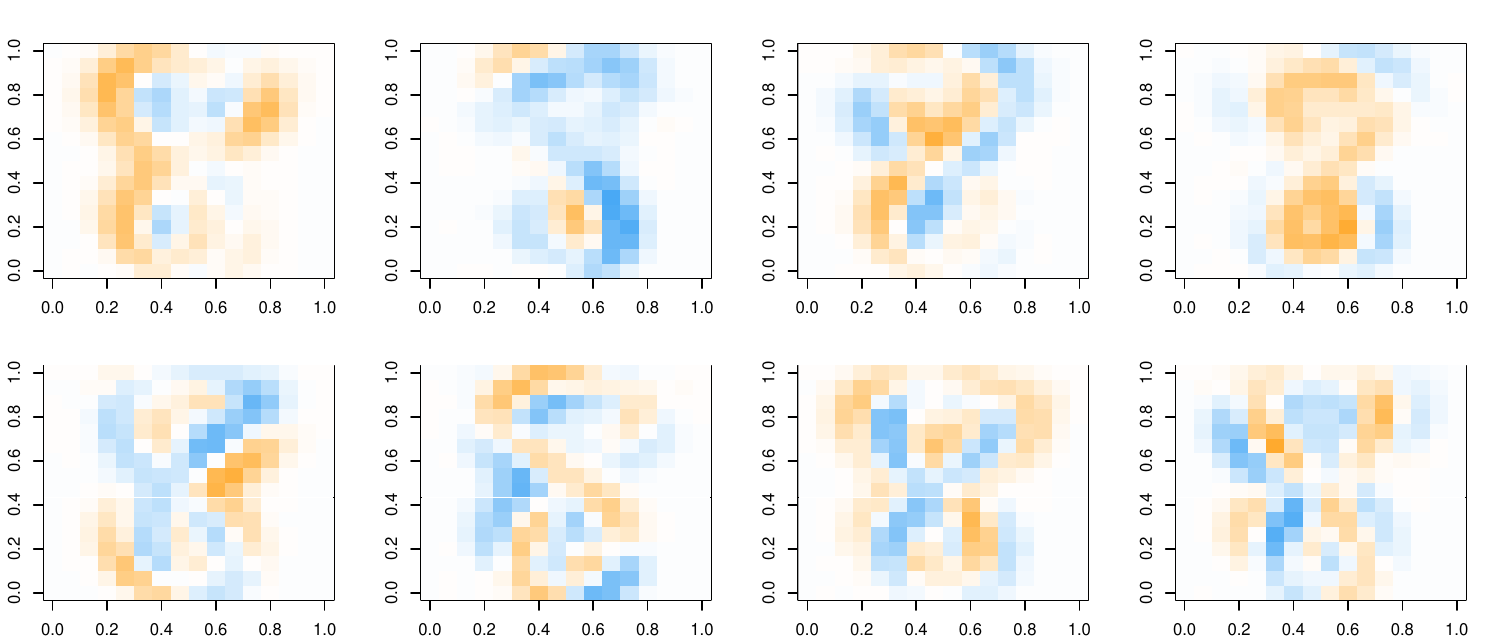}
  \caption{}
  \label{fig:c}
\end{subfigure}

\caption{Representation of the loading vectors for the clusters corresponding to digits $3$, $5$, and $8$. The loading vectors for digit $3$ in (a), for digit $5$  in (b), and for digit $8$ in (c).
{\color{myorange}$\blacksquare$} positive weights; {\color{myblue}$\blacksquare$} negative weights.}
\label{loading_vectors}
\end{figure}

%It is worth mentioning that the dimensions estimated by the HDDC method using the same parameters (\texttt{thresh}$=0.2$ and $q_{\text{max}}=20$) were $q_1=8$, $q_2=9$ and $q_3=8$, which are smaller than those obtained by tHDDC. If the parameter \texttt{thresh} is further refined or the value of $q_{\text{max}}$ is reduced, the HDDC method rarely selects the maximum allowed number of dimensions. In contrast, the tHDDC method may occasionally select the maximum dimension, but it consistently yields a substantially smaller classification error than HDDC.

When further analyzing this data set, it can be seen that the first observations trimmed, i.e., those with the smallest $D^H(x_i;\hat{\theta})$ values, correspond to the artificially generated outliers, and the subsequently trimmed ones are mostly from digits other than ``$3$'', ``$5$'', and ``$8$'', which we had introduced artificially. This is a desirable outcome, as the images and digits not conforming to the structure of the majority digits are so removed first. In Figure~\ref{trimmed_outliers}, we display $16$ examples of these artificially introduced outliers that are trimmed.

\begin{figure}[htb]
    \centering
    \includegraphics[width=0.4\textwidth]{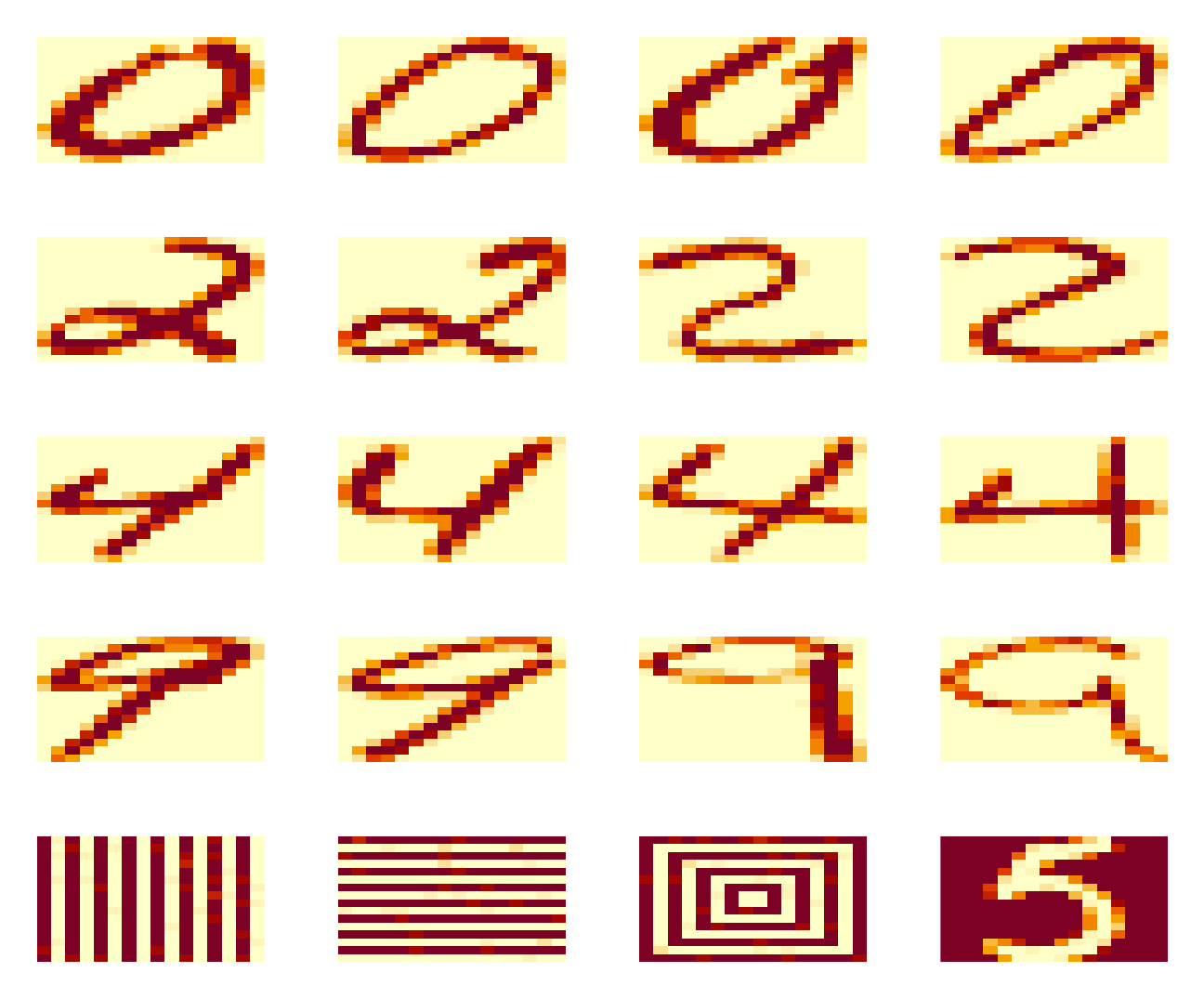}
    \caption{Examples of trimmed outliers artificially introduced into the dataset.}
    \label{trimmed_outliers}
\end{figure}

We can further analyze these trimmed observations with the tools discussed in Section~\ref{sec4.2}, namely the score distances, SD$_i$, and the orthogonal distances, OD$_i$. By jointly plotting these two quantities, together with the cut-offs introduced in Section~\ref{sec4.2}, we can obtain additional insight into these detected anomalies. Figures~\ref{sd_od_3}, \ref{sd_od_5}, and \ref{sd_od_8} display these diagnostic plots based on these SD$_i$ and OD$_i$ distances. Note that trimmed observations are also ``assigned to clusters'' in these plots by assigning each of them to the cluster for which it attains the largest value of $D^H(x;\theta)$.
%The trimming level was set to $20\%$, exceeding the $12\%$ of artificially introduced contamination, in order to accommodate the intrinsic variability of handwritten digits and to account for the presence of atypical shapes even among the supposedly clean digits.
Each observation appears in these plots labeled by the digit it represents, while a cross ``$\times$'' is used for outlying images that do not correspond to digits but rather to synthetic patterns. It could be seen that all these 20 images shown in Figure~\ref{trimmed_outliers} appear outside the considered cutoffs. Note that the second type of outliers (those not representing digits) appears in the diagnostic plot corresponding to digit $8$. % with ``$\times$'' symbols. The color coding is as follows: gray for digit $0$, purple for digit $2$, blue for digit $3$, yellow for digit $4$, green for digit $5$, red for digit $8$, brown for digit $9$, and black for the remaining outliers. The observations shown in Figure~\ref{highlight_observations} (numbered from $\mathbf{1}$ to $\mathbf{12}$) have also been identified.
To further illustrate the structure of these diagnostic plots, Figure~\ref{highlight_observations} displays $12$ particular observations for which their distances have been highlighted, using labels \textbf{1} to \textbf{12}. Note first that some of these 12 highlighted observations are clearly visually ambiguous, making them difficult to unambiguously identify as digits ``$3$'', ``$5$'', or ``$8$'', which further justifies the use of a trimming proportion exceeding the level of artificially added contamination.

\begin{figure}[htbp]
\centering

\begin{subfigure}{0.47\linewidth}
  \centering
  \includegraphics[width=0.8\linewidth]{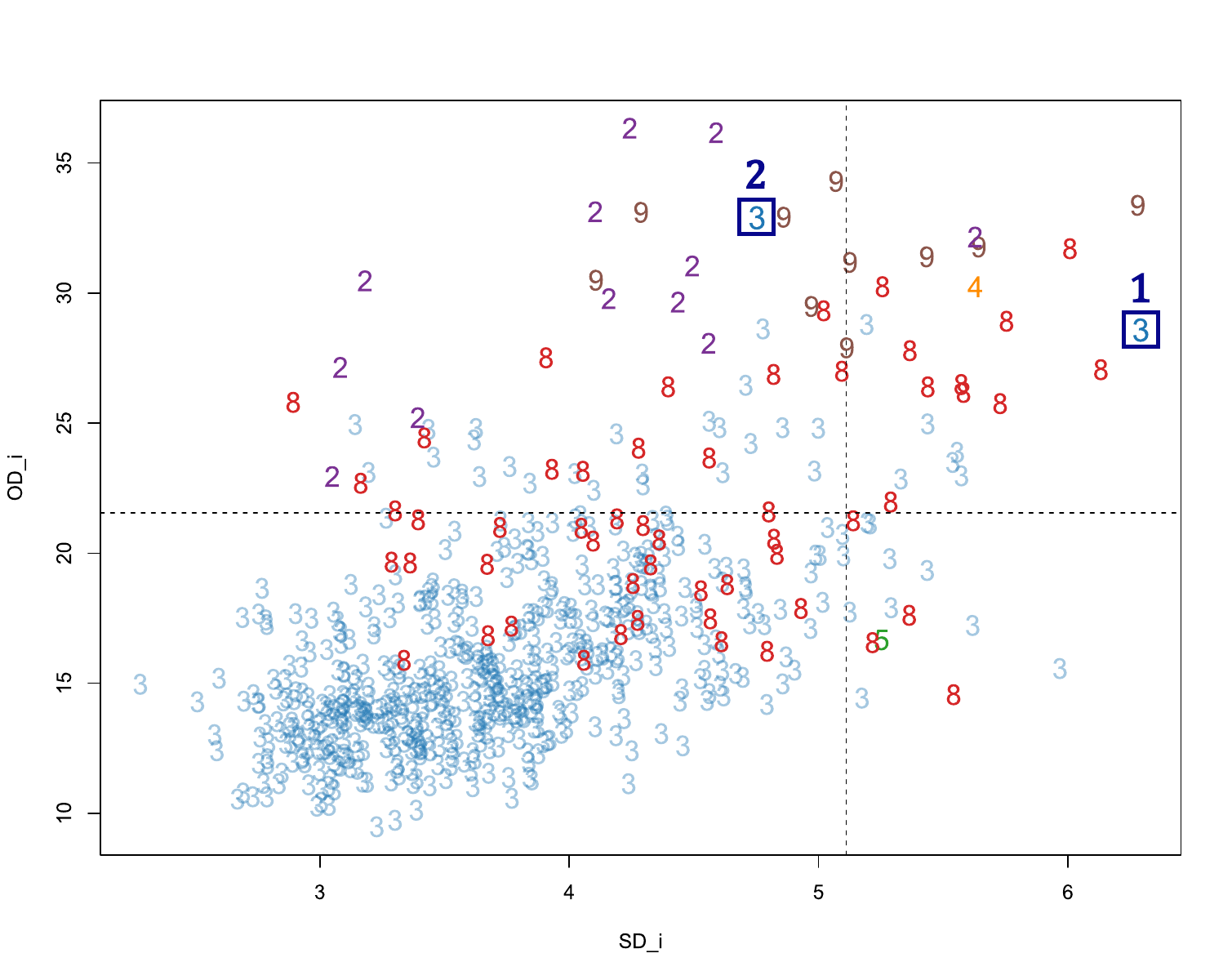}
  \caption{}
  \label{sd_od_3}
\end{subfigure}
\hfill
\begin{subfigure}{0.47\linewidth}
  \centering
  \includegraphics[width=0.8\linewidth]{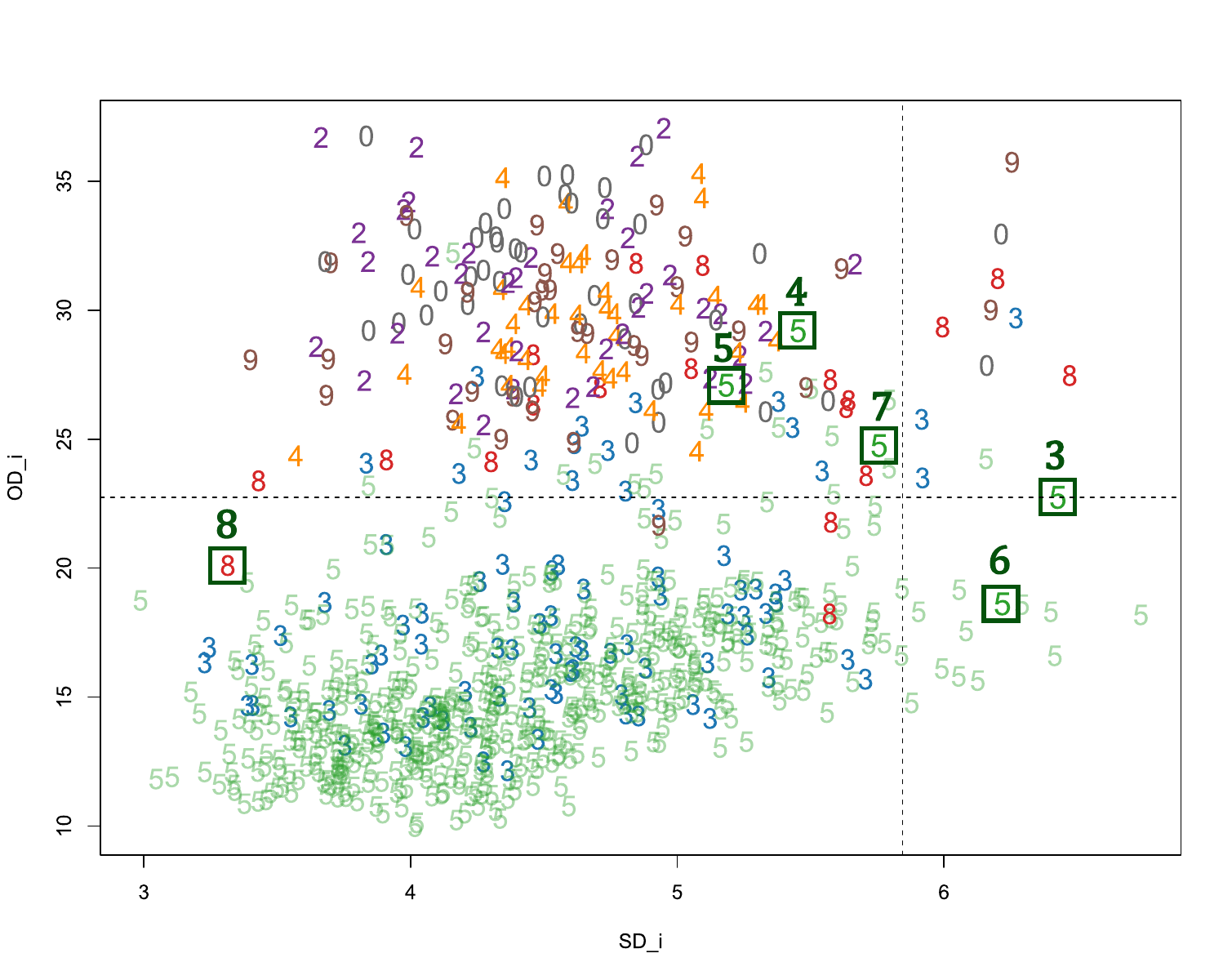}
  \caption{}
  \label{sd_od_5}
\end{subfigure}

\hfill

\begin{subfigure}{0.47\linewidth}
  \centering
  \includegraphics[width=0.8\linewidth]{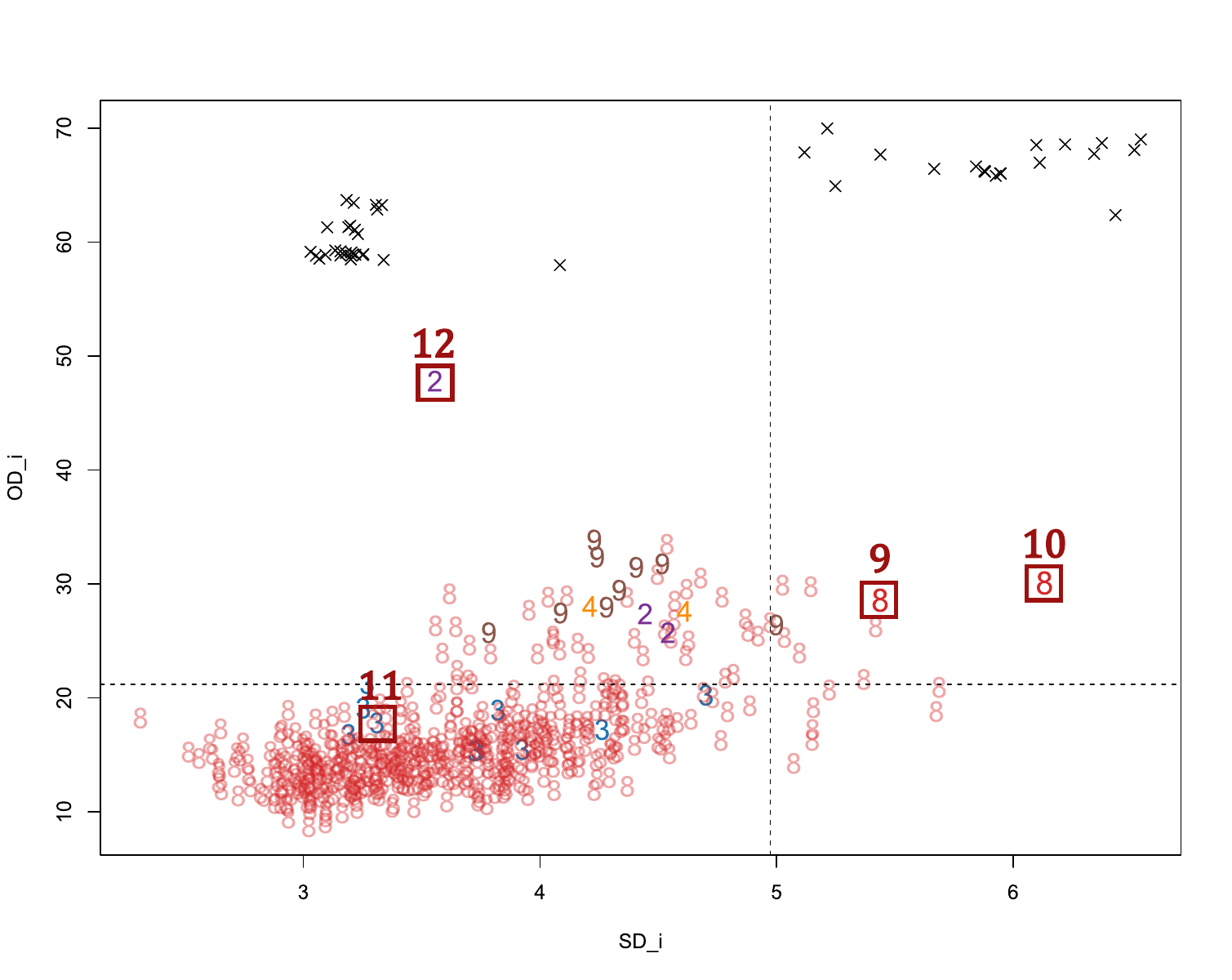}
  \caption{}
  \label{sd_od_8}
\end{subfigure}
\hfill
\begin{subfigure}{0.47\linewidth}
  \centering
  \includegraphics[width=0.8\linewidth]{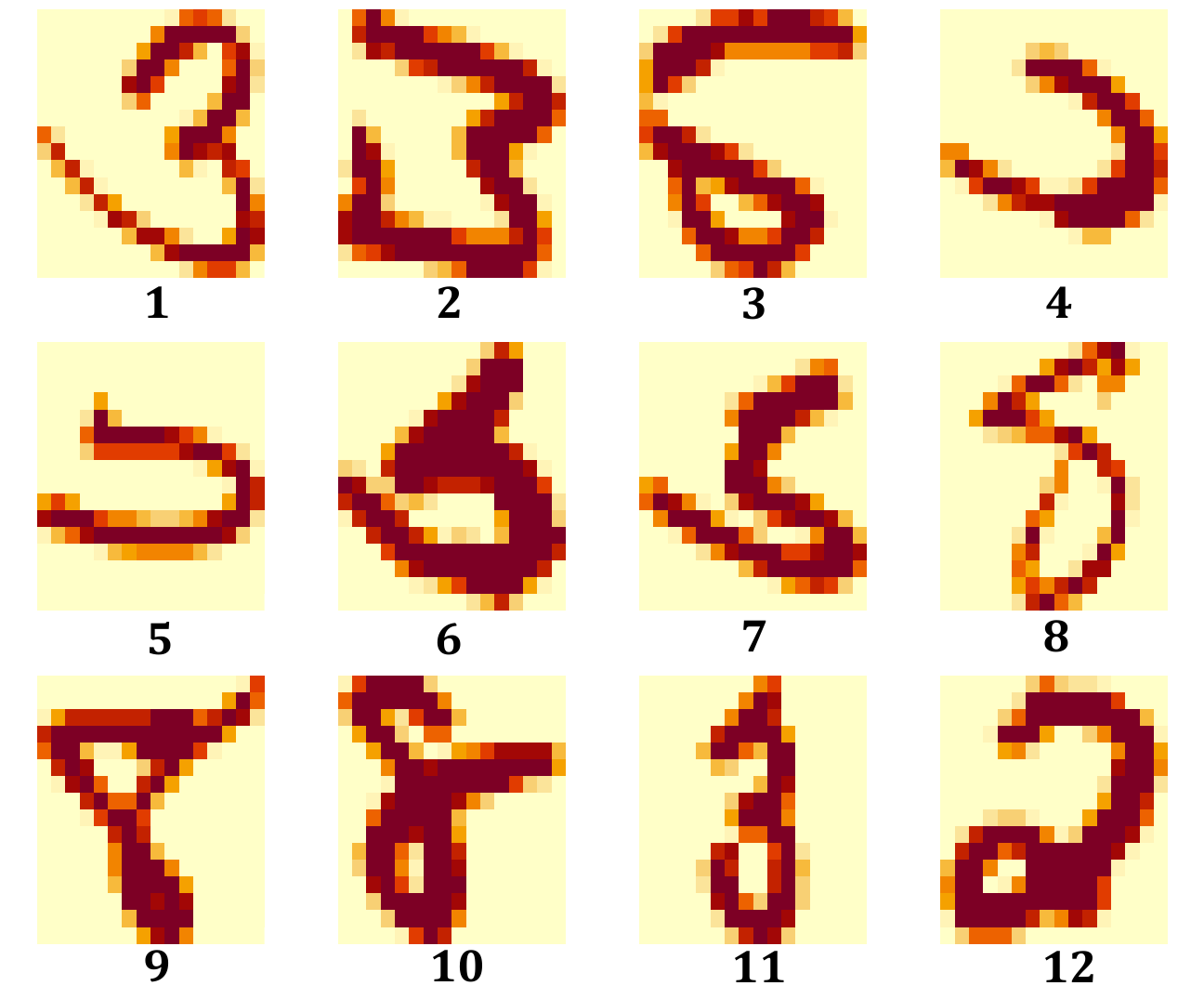}
  \caption{}
  \label{highlight_observations}
\end{subfigure}

\caption{Score distance versus orthogonal distance plots for the images assigned to cluster 3 in (a), cluster 5 in (b), and cluster 8 in (c). Some illustrative observations highlighted in these plots are shown in (d).}
\label{sd_od}
\end{figure}

Observations labeled as $\mathbf{1}$ and $\mathbf{2}$ appear within the SD$_i$ and OD$_i$ diagnostic plot for images assigned to the cluster of $3$'s (Figure~\ref{sd_od_3}),  
but they clearly exhibit visibly nonstandard representations. In the diagnostic plot corresponding to the cluster with the $5$ (Figure~\ref{sd_od_5}), the observations labeled as $\mathbf{3}$, $\mathbf{4}$, $\mathbf{5}$, $\mathbf{6}$, $\mathbf{7}$, and $\mathbf{8}$  appear. The first five correspond to digits originally labeled as ``5'' but exhibiting highly atypical shapes. They appear outside the cutoffs and clearly separated from the main cloud. On the other hand, the image labeled as $\mathbf{8}$, originally labeled as ``8'', lies near the main cloud of points. This observation was misclassified into the cluster including the 5's, as the incomplete handwritten form of the digit ``8'' resembles the digit ``5''. Finally, the last row in Figure~\ref{highlight_observations} contains images assigned to the cluster with the $8$'s (Figure~\ref{sd_od_8}). Images $\mathbf{9}$ and $\mathbf{10}$ lie outside both cut-offs, and exhibit clearly atypical shapes of the digit ``8''. Image $\mathbf{11}$, corresponds to a misclassified digit ``3'', whose closed shape makes it visually similar to the digit ``8''. Finally, image  $\mathbf{12}$ , originally labeled as digit ``2'', appears in the central region of the plot but far from the main cloud of points, reflecting the highly atypical handwritten form of this digit .

Finally, Figure~\ref{silhouette_plot}, displays the discriminant factors $\text{DF}(i)$ introduced in Section~\ref{sec4.2} for each observation, grouped according to the clusters produced by tHDDC, and with the trimmed observations represented in black color. It is worth noting that, within the group of trimmed observations, the last block of observations with discriminant factors exactly equal to $600$ has been obtained by truncating their $\text{DF}(i)$ values at this $600$ threshold. This block exactly corresponds to the artificially generated outliers, whose discriminant factor values were notably larger than 600.

\begin{figure}[htbp]
\centering

\begin{subfigure}{0.65\linewidth}
  \centering
  \includegraphics[width=0.8\linewidth]{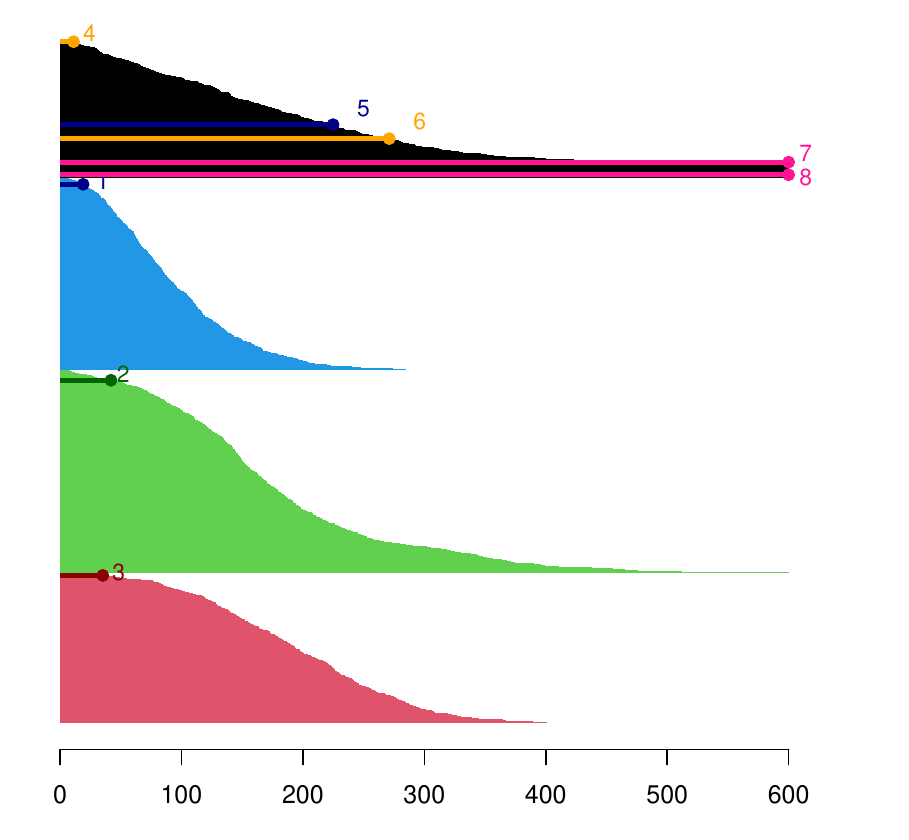}
  \caption{}
  \label{silhouette_plot}
\end{subfigure}
\hfill
\begin{subfigure}{0.25\linewidth} \centering \includegraphics[width=0.95\linewidth]{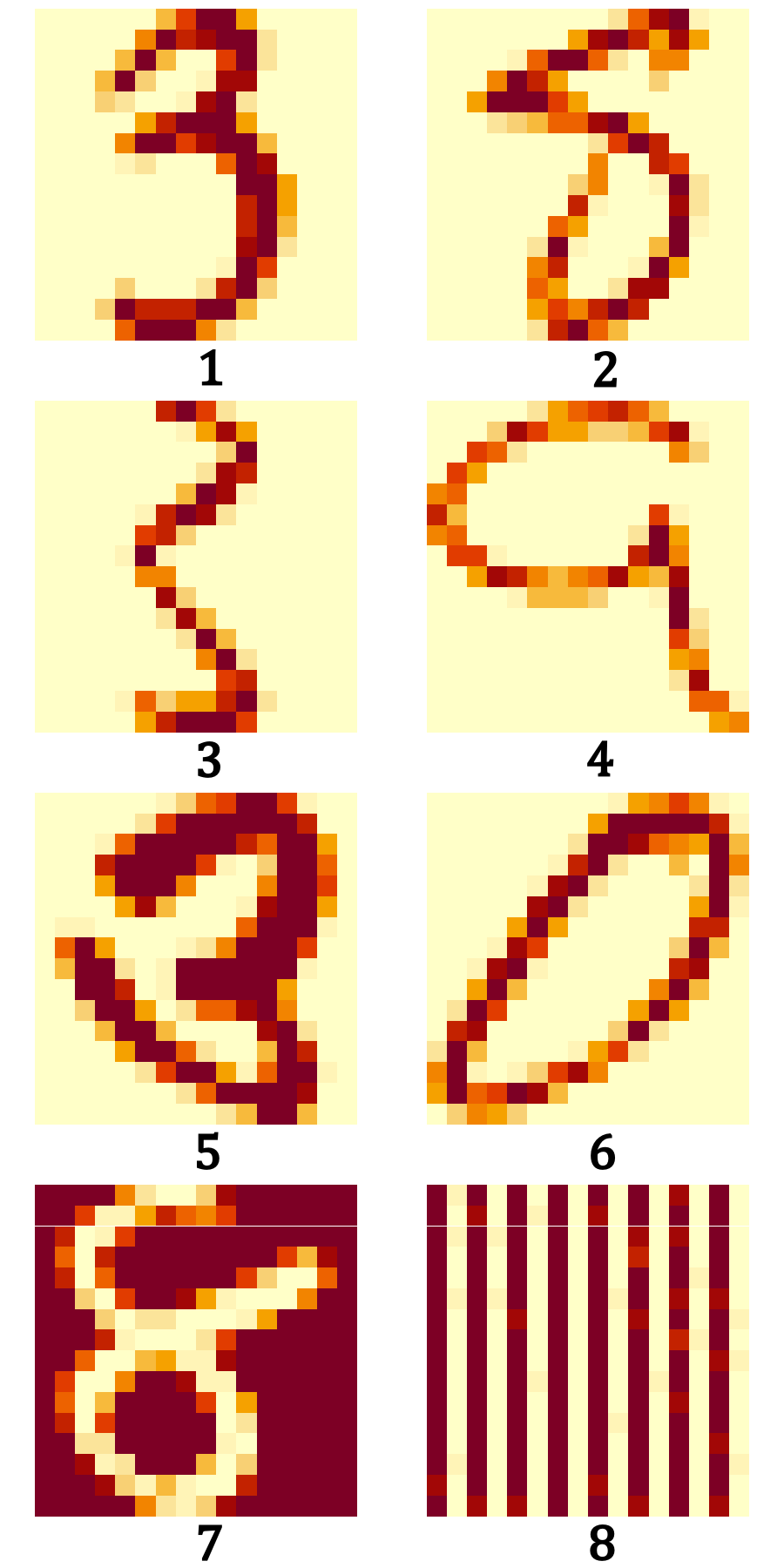} 
\vspace{0.1cm}
\caption{} \label{silhouette_observations} \end{subfigure}

\caption{Display of the $\text{DF}(i)$ values (horizontal axis) for tHDDC results. Colors indicate cluster assignments: black (trimmed), blue (mostly 3's), green (mostly 5's), and red (mostly 8's). Eight illustrative observations highlighted in the silhouette plot are shown in (b). %Note: The horizontal axis is displayed in reverse order (more negative values on the right, values closer to zero on the left) to resemble a standard silhouette plot, making it easier to identify doubtful assignments (values near zero) and strongly assigned or trimmed observations (more negative values).
}
\label{silueta}
\end{figure}

As discussed previously, observations with $\mathrm{DF}(i)$ values close to zero indicate doubtful assignment or trimming decisions. Therefore, clusters containing many such observations would suggest weak separation or a poorly determined clusters' structure. In Figure~\ref{silhouette_plot}, however, only a small number of observations lie close to zero, while most display clearly high and positive $\mathrm{DF}(i)$ values. This pattern provides additional support for the stability of the obtained clustering solution. In Figure~\ref{silhouette_plot}, we highlight $8$ observations that are displayed in Figure~\ref{silhouette_observations}, labeled $\mathbf{1}$ to $\mathbf{8}$. The first four correspond to doubtful assignment cases due to their discriminant factors being close to zero. For example, observation $\mathbf{2}$ visually resembles a digit ``5'', although it was originally labeled as digit ``8'', because that handwritten digit is incomplete. Observations $\mathbf{5}$ and $\mathbf{6}$ correspond to clearly trimmed observations: the first one is a digit ``3'' clearly written in a nonstandard way, and the second corresponds to a digit ``0''. %, which does not conform to the loading vectors defining the subspaces where the clusters of digits $3$, $5$, and $8$ reside.
Finally, images $\mathbf{7}$ and $\mathbf{8}$ correspond to artificially generated outliers exhibiting very large discriminant factor values (above 600).

\section{Conclusions and further directions}\label{sec_7}

A tHDDC methodology has been introduced in this work, providing a robust approach for clustering data in increasing dimensionality by unifying the parsimonious modeling of HDDC with the trimming and eigenvalue-constraint mechanisms underlying TCLUST. The proposed methodology aims to overcome the difficulties that TCLUST faces when dealing with datasets of even moderate dimensionality, while also offering a less simplistic alternative to the approach provided by RLG.

The development of a feasible trimmed EM-type algorithm based on properly modified C-steps, which are able to estimate the intrinsic dimensions, establishes tHDDC as a feasible approach when robustness is particularly required. The theoretical analysis presented (including existence and consistency results) provides additional support for its use. Additionally, tHDDC enables the development of graphical diagnostic tools that are useful for detecting anomalies in moderately high-dimensional datasets with heterogeneity arising from the presence of unknown clusters.

Our experimental results on both synthetic and real datasets demonstrate that tHDDC is capable of achieving better performance than standard TCLUST and RLG in complex scenarios involving moderately high-dimensional settings. Notably, this gain in robustness is achieved alongside a substantial reduction in computational cost compared to TCLUST as the dimension increases. Furthermore, the inclusion of trimming and constraints allows tHDDC to improve the original HDDC in environments with contaminating observations, where HDDC performance deteriorates due to its lack of robustness. 

There are several promising directions for future research. A functional version of the tHDDC algorithm is currently under development. Adapting the trimming and subspace projection techniques to a functional data context represents a natural and highly relevant extension of the proposed methodology. Additionally, another major line of research involves exploring alternative robustification strategies, specifically cellwise trimming. The traditional casewise trimming approach adopted in this paper discards an entire observation if it is identified as an outlier. However, in high-dimensional settings, an observation may be contaminated in only a few variables (cells), while the remaining data entries in the observation still contain valuable structural information \cite{10.1214/07-AOS588}. Incorporating cellwise trimming into the tHDDC framework would allow the algorithm to ignore only the corrupted cell entries rather than the entire observation, potentially leading to more efficient use of the data under the presence of cellwise contamination.

\appendix

\section{Appendix: Proofs of the existence and consistency results}

The existence and consistency proofs follow the general structure of \cite{garcia2008general}, and \cite{garcia2014robust}, adapted to the HDDC scatter parameterization used here. Differences are discussed below, while analogous parts are omitted and the reader is referred to the original proofs.

\renewcommand{\thesection}{\Alph{section}}
\numberwithin{lemma}{section}

\subsection{Proof of the existence result}\label{sec_A1}

Recall that the classification version of the problem can be expressed as the maximization, over 
$\theta \in \Theta_{c_1, c_2}$, of $\text{CL}(\theta; P)$ defined in (\ref{problema_hard}). 
Let us consider a sequence $\{\theta_n\}_{n=1}^\infty = \{ (\pi_1^n, \ldots, \pi_G^n, \mu_1^n, \ldots, \mu_G^n, \Sigma_1^n, \ldots, \Sigma_G^n) \}_{n=1}^\infty \subset \Theta_{c_1,c_2}$ 
 such that
\begin{equation}\label{cond_non_negative}
\lim_{n\to\infty} \text{CL}(\theta_n; P)=\sup_{\theta\in\Theta_{c_1,c_2}} \text{CL}(\theta;P)=M>-\infty.    
\end{equation}

The boundedness from below in (\ref{cond_non_negative}) follows by considering the parameter configuration defined by
$\pi_1 = 1$, $\mu_1 = 0 \in \mathbb{R}^p$, and $\Sigma_1 = \mathbf{I}_p$,
while setting all remaining weights $\pi_g = 0$ and choosing the remaining location and scatter parameters arbitrarily, as long as they satisfy the constraints imposed by $c_1$ and $c_2$ (for instance, with $\Sigma_2 =...=\Sigma_G= \mathbf{I}_p)$. 
This configuration can always be used for any pair of constants $c_1$ and $ c_2$, greater than or equal to 1, and therefore establishes a finite lower bound for $M$.

Since $[0,1]^G$ is a compact set, a subsequence of $\{\theta_n\}_{n=1}^\infty$ (denoted again by the same notation) can be extracted such that
$$
\pi_g^n \rightarrow \pi_g^0 \in [0,1] \text{ for } 1\leq g \leq G.
$$
For this subsequence, let us introduce the following notation:
\begin{align*}
\overline{M_g^n} &= \max_{l=1,\ldots,q_g} \lambda_{gl}^n, \qquad
\underline{M_g^n} = \min_{l=1,\ldots,q_g} \lambda_{gl}^n, \\[4pt]
\overline{m^n} &= \max_{g=1,\ldots,G} \lambda_g^n, \text{ and, }
\underline{m^n} = \min_{g=1,\ldots,G} \lambda_g^n .
\end{align*}
Our goal is to show that it is always possible to extract a further subsequence which converges to an optimal solution. In order to prove this, we first show that $\underline{m^n}$ cannot converge to $0$ and that if $\overline{M_g^n}$ diverge to $\infty$ then necessarily $\mathbb{E}_P[z_g(\cdot,\theta_n)]$ converges to $0$.

We consider the following upper bound:
\begin{multline}\label{cota_sup_cl}
\text{CL}(\theta_n;P) \le
\mathbb{E}_P\Bigg[
\sum_{g=1}^G z_g(\cdot,\theta_n)
\Bigg(
\log(\pi_g^n)
-\frac12\Big(
\norm{\operatorname{Pr}_{\mathcal{B}_g^n}(\cdot)-\mu_g^n}^2_{\mathcal{B}_g^n}
\\
+\frac{1}{\overline{m^n}}
\norm{\operatorname{Pr}_{\mathcal{B}_g^n}(\cdot)-\cdot}^2
+\sum_{j=1}^{q_g}\log(\lambda_{gj}^n)
+(p-q_g)\log(\underline{m^n})
+p\log(2\pi)
\Big)
\Bigg)
\Bigg],
\end{multline}
where $\mathcal{B}_g^n$ denotes the approximating subspace associated with the
$g$th component (for simplicity, the dependence of $\mathcal{B}_g^n$ on $\theta_n$ is omitted from the notation). Starting from \eqref{problema_hard} and using the standard expression for the multivariate normal density, we also obtain the bound
\begin{multline}\label{cota_sup_cl_2}
\text{CL}(\theta_n;P) \le
\mathbb{E}_P\Bigg[
\sum_{g=1}^G z_g(\cdot,\theta_n)
\Bigg(
\log(\pi_g^n)
-\frac12\Bigg(
\frac{1}{\overline{M_g^n}}\norm{\cdot-\mu_g^n}^2
\\
+\sum_{j=1}^{q_g}\log(\lambda_{gj}^n)
+(p-q_g)\log(\underline{m^n})
+p\log(2\pi)
\Bigg)
\Bigg)
\Bigg].
\end{multline}
Moreover, it holds that
\begin{equation}\label{posit}
\mathbb{E}_P\!\left[
\sum_{g=1}^G
z_g(\cdot,\theta_n)
\norm{\operatorname{Pr}_{\mathcal{B}_g^n}(\cdot)-\mu_g^n}^2_{\mathcal{B}_g^n}
\right] \ge 0 ,
\end{equation}
and that there exists a constant $h>0$ such that
\begin{equation}\label{bound_h}
\mathbb{E}_P\!\left[
\sum_{g=1}^G
z_g(\cdot,\theta_n)
\norm{\operatorname{Pr}_{\mathcal{B}_g^n}(\cdot)-\cdot}^2
\right] \ge h .
\end{equation}
The existence of such a constant $h$ is proven in part (i) of  Lemma~\ref{lema_h}.   

Combining inequality~\eqref{cota_sup_cl} with
\eqref{posit} and \eqref{bound_h}, and noting that $\log(\pi_g^n)\le 0$, we obtain
\begin{multline}\label{cota_sup_cl_3}
\text{CL}(\theta_n;P) \le
-\frac12\Bigg(
\frac{h}{\overline{m^n}}
+\sum_{g=1}^G \bigg[
\mathbb{E}_P[z_g(\cdot,\theta_n)]
\sum_{j=1}^{q_g}\log(\lambda_{gj}^n)\bigg]
\\
+(1-\alpha)(Gp-Q)\log(\underline{m^n})
+(1-\alpha)p\log(2\pi)
\Bigg),
\end{multline}
w{here $Q=\sum_{g=1}^G q_g$.

Let us first suppose that $\underline{m^n} \to 0$. Then, the bound~\eqref{cota_sup_cl_3} and the eigenvalue constraint with constant $c_2<\infty$ in (\ref{ER}) would imply that $\overline{m^n} \to 0$, and therefore $\text{CL}(\theta_n;P) \to -\infty$, contradicting (\ref{cond_non_negative}). Moreover, since $\underline{m^n} \le c_2 \overline{m^n}$, the term $1/\overline{m^n}$ asymptotically dominates $\log(\underline{m^n})$.
On the other hand, suppose that for some component $g$ it holds that $\limsup_{n \to \infty} \overline{M_g^n} = \infty$. Then, it would follow that $\mathbb{E}_P[z_g(\cdot,\theta_n)] \to 0$, since otherwise $\text{CL}(\theta_n;P) \to -\infty$, again contradicting~\eqref{cond_non_negative}.

From inequality (\ref{cota_sup_cl_2}), it further follows that, for components $g$ whose eigenvalues remain bounded while the corresponding means diverge ($\limsup_{n\to\infty}\|\mu_g^n\|=\infty$), one necessarily has $\mathbb{E}_P[z_g(\cdot,\theta_n)] \to 0$, since otherwise it would follow that $\text{CL}(\theta_n;P)\to -\infty$. This follows by considering a ball
$B(0,R)$ such that $P[B(0,R)]>\alpha$ and noting that the set
$A_n=\{x:\sum_{g=1}^G z_g(x,\theta_n)=1\}$ must intersect this ball $B(0,R)$ with
positive probability.
%\la{[QUITAR TODO ESTO: Therefore, eigenvalues must be bounded away from zero, and for any component with unbounded eigenvalues or means, the corresponding assignment probabilitiesmust vanish asymptotically. Consequently, at least one of the components must have bounded eigenvalues and means, because otherwise \sum_{g=1}^G \mathbb{E}_P[z_g(\cdot,\theta_n)] \to 0$, contradicting the fact that this sum equals $1 - \alpha$. For components with $\mathbb{E}_P[z_g(\cdot,\theta_n)]\to 0$, the optimal value for the weights can be chosen as $\pi_g=0$, and the sequences of means and eigenvalues associated with such components may be chosen arbitrarily without affecting the limit value of $\text{CL}(\theta_n;P)\to -\infty$ [LO DE $-\infty$ ESTA MAL, NO?], and, thus, a convergent subsequence attaining the optimum can always be extracted.]}

Take also into account that at least one of the sequences $\{\|\mu_g^n\|\}_n$, for $g=1,...,G$,} must be bounded, because otherwise $\sum_{g=1}^G \mathbb{E}_P[z_g(\cdot,\theta_n)] \to 0$, contradicting that $P[x:\sum_{g=1}^G z_g(x,\theta_n)=1] = 1 - \alpha$. The proof concludes by recalling that for the other components with diverging means (but bounded eigenvalues), the limiting weights should verify that $\pi_g^0 = 0$, because otherwise the limiting value of the target function $\text{CL}(\theta_n;P)$ can be improved by redefining the associated components' weights. For such components with $\pi_g^0 = 0$ and $\mathbb{E}_P[z_g(\cdot,\theta_n)] \to 0$, the sequences of means and eigenvalues may be replaced arbitrarily by convergent ones without affecting the limit of $\text{CL}(\theta_n;P)$. Consequently, a convergent subsequence attaining the optimum can always be extracted for all components.

Finally, an analogous argument also applies to the trimmed mixture likelihood.
An initial bound as in (\ref{cond_non_negative}) trivially holds. Then, a similar reasoning to that used for the classification likelihood case can be applied by using the bound
\begin{multline}\label{acot_CL_ML}    \text{ML}(\theta;P)=\mathbb{E}_P\left[z(\cdot,\theta)\log D^M(\cdot,\theta)\right]\leq \mathbb{E}_P\left[z(\cdot,\theta)\log\left(G \max_{g=1,\ldots,G}\left\{\exp D_g(\cdot,\theta)\right\}\right)\right]\\
    \leq \log G(1-\alpha)+\mathbb{E}_P\left[\sum_{g=1}^G z_g(\cdot,\theta)D_g(\cdot;\theta)\right]=(1-\alpha)\log G+\text{CL}(\theta;P),
\end{multline}
where the same definitions for the $z(\cdot,\theta)$ and $z_g(\cdot,\theta)$ functions are used (this result would be the analogous in the HDDC framework of Lemma 1 in \cite{garcia2014robust}).

\subsection{Proof of the consistency result}\label{sec_A2}

The consistency proof combines arguments similar to those seen in Section \ref{sec_A1} combined with empirical process  techniques \citep{van1996weak}. Since the arguments closely follow those in \cite{garcia2008general} and \cite{garcia2014robust}, we only provide a sketch of them and refer to these works for further details.

From the sequence of empirical estimates $\{\theta_n\}_{n=1}^\infty$, as in Section \ref{sec_A1}, we can extract a subsequence (denoted again by the same notation)  such that  $\pi_g^n \rightarrow \pi_g^0 \in [0,1]$, for all $1 \leq g \leq G$, $P$-a.e.. We first see that there exists a constant $M'$ such that   \begin{equation}\label{cota_inf_cl}
        \text{CL}(\theta_n; P_n)\geq M'>-\infty\quad P\textup{-a.e. for}\:\:n\geq n_0.
    \end{equation}
Due to the tightness of the  sequence $\{P_n\}_{n=1}^{\infty}$ of empirical measures, we can find a constant $R>0$ such that $P_n(B(0,R))> 1-\alpha.$ Therefore, we have    \begin{equation}\label{boun_empiric}
        \text{CL}(\theta_n;P_n)\geq\int_{B(0,R)}\left(-\dfrac{p}{2}\log(2\pi)-\dfrac{1}{2}\norm{x}^2\right)dP_n(x)
    \end{equation}
    The bound (\ref{boun_empiric}) is derived by substituting $P$ by $P_n$  in \eqref{cota_sup_cl_2} and taking into account that $\theta_n\in\Theta_{c_1,c_2}$ should improve upon the particular parameter configuration given by   $\pi_1=1$, $\pi_2=...=\pi_G=0$, together with $\mu_g=0$ and $\Sigma_g=\mathbf{I}_p$ for $g=1,...,G$. The bound then follows now by applying the Strong Law of Large Numbers..

    We can also consider  \eqref{cota_sup_cl}, but now substituting $P$ by $P_n$.  %, so we have:
    %\begin{multline*}
    %CL(\theta_n;P_n)
    %\leq \mathbb{E}_{P_n}\left[\sum_{g=1}^G z_g(\cdot)\left(\log(\pi_g^n)-\dfrac{1}{2}\left(\norm{Pr_g(\cdot)-\mu_g^n}^2_{B_g}+\dfrac{1}{\overline{m_n}}\norm{Pr_g(\cdot)-\cdot\:}^2 \right.\right.\right.\\
    %\left.\left.\left.+\sum_{j=1}^{q_g}\log(\underline{M_n})+(p-q_g)\log(\underline{m_n})+p\log(2\pi)\right)\right)\right],
    %\end{multline*}
     The analogue of \eqref{bound_h}, which is also proven in part (ii) of Lemma~\ref{lema_h}, guarantees the existence of a constant $h' > 0$ such that   \begin{equation}\label{boun_h'}\mathbb{E}_{P_n}\left[ \sum_{g=1}^G z_g\left(\cdot;\theta_n\right)\norm{\text{Pr}_{\mathcal{B}_g^n}(\cdot)-\cdot}^2\right]\geq h'.
    \end{equation}    
The proof now follows similar steps to those presented in Section~\ref{sec_A1}, by showing that from~\eqref{cota_sup_cl_3}, but substituting $P$ by $P_n$, and observing that if $\underline{m^n} \to 0$, then $\text{CL}(\theta_n;P_n) \to -\infty$ $P$-a.e., yielding a contradiction with~\eqref{cota_inf_cl}. Additionally, for any component $g$ such that $\limsup_{n \to \infty} \overline{M_g^n} = \infty$, it must hold that $\mathbb{E}_{P_n}[z_g(\cdot,\theta_n)] \to 0$, because otherwise $\text{CL}(\theta_n;P_n) \to -\infty$, so contradicting~\eqref{cota_inf_cl}.

%Working again with inequality~\eqref{cota_sup_cl_2}, but substituting $P$ by $P_n$, we obtain:
%\begin{itemize}
%\item For components $g$ with bounded eigenvalues but diverging means,
%$\limsup_{n\to\infty}\|\mu_g^n\|=\infty$, it necessarily holds that
%$\mathbb{E}_{P_n}[z_g(\cdot,\theta_n)]\to 0$. Otherwise,
%$\text{CL}(\theta_n;P_n)\to -\infty$. This follows by considering a ball
%$B(0,R)$ such that $P_n[B(0,R)]>\alpha$ and noting that the trimming set
%$A_n=\{x:\sum_{g=1}^G z_g(x,\theta_n)=1\}$ must intersect $B(0,R)$ with
%positive probability.
%\end{itemize}

The next step in the proof of this consistency result has to do with the empirical centers $\mu_g^n$. %can be chosen so that their norms are uniformly bounded with probability 1 for those groups with \la{$\pi_g^0 > 0$}. 
As in Section~\ref{sec_A1}, working with inequality~\eqref{cota_sup_cl_2} but again substituting $P$ by $P_n$, we can see that for a components $g$ with bounded eigenvalues but diverging means, it necessarily holds that $\mathbb{E}_{P_n}[z_g(\cdot,\theta_n)] \to 0$, because otherwise, $\text{CL}(\theta_n;P_n) \to -\infty$. This follows by considering again a ball $B(0,R)$ such that $P_n[B(0,R)] > \alpha$ uniformly in $n$, for $n$ large enough, and noting that the set $A_n = \{x : \sum_{g=1}^G z_g(x,\theta_n) = 1 \}$ satisfy that $P_n[B(0,R)\cap A_n]>0$.

Consequently, arguments, analogous to those employed in in Section~\ref{sec_A1} and \cite{garcia2008general}  shows that there exists a compact set $K$ such that $\theta_n \in K$ for all $n \geq n_0$ with probability $1$. In fact, for components whose eigenvalues or location parameters diverge to infinity, the corresponding empirical parameters can be modified without affecting the value of the limit objective function. %For instance, eigenvalues diverging to infinity may be replaced by bounded ones, such as the largest eigenvalue of a component $g$  for which $\mathbb{E}_{P_n}[z_g(\cdot,\theta_n)]$ converges to a strictly positive limit \la{[No claro]}. Similarly, the associated location parameters may be chosen arbitrarily, for instance equal to zero, since these components receive asymptotically vanishing assignment probabilities and therefore do not contribute to the objective function. A similar \la{procedure} for handling parameter divergence was employed in \cite{garcia2008general}.
Note also that at least one component must have bounded eigenvalues and means.%, because otherwise
%$$
%\sum_{g=1}^G \mathbb{E}_{P_n}[z_g(\cdot,\theta_n)] \to 0,
%$$
%which would contradict the fact that this sum equals $1-\alpha$.
%Note also that for components satisfying
%$\mathbb{E}_{P_n}[z_g(\cdot,\theta_n)] \to 0$, the optimal value for the weights can be obtained as  $\pi_g\la{^n} = 0$.

Let $z_g^\ast(x;\theta)=I\{ x: D^H(x;\theta)=D_g(x;\theta)\}$, i.e. all the points in $\mathbb{R}^p$ are assigned to some class through the functions $z_g^\ast$. The proof concludes by noting that the objective function in the empirical case can be rewritten as
\begin{equation*}
    \text{CL}(\theta;P_n)=\int_{\lbrace x\::\:D^H(x;\theta)\geq R(\theta;P_n)\rbrace}\left[\sum_{g=1}^Gz_{g}^\ast(x;\theta)\log(D_g(x,\theta)\la{)}\right]dP_n(x).
\end{equation*}
In the previous expression, we can asymptotically replace $R(\theta;P_n)$ by $R(\theta;P)$, because
\begin{equation*}\label{conv_R}
        \sup_{\theta\in K}\vert R(\theta,P_n)-R(\theta,P)\vert=o_P(1),
\end{equation*}  
by using Lemma \ref{converg_emp} and the fact that the integrand can be bounded from above from some constants uniformly for $\theta$ in the compact set $K$.
Finally, the Glivenko--Cantelli property of the class $\mathcal{H}$ in \eqref{class_H}, together with Theorem~3.2.3 in \cite{van1996weak}, yields the desired consistency, as done in \cite{garcia2008general}.

So far, we have considered the classification trimmed likelihood, but the same scheme of proof can be analogously applied to trimmed mixture likelihoods, as in \cite{garcia2014robust}. In this case, we consider the class of functions
$$
\mathcal{H}' := \left\{ I_{[u,\infty)}\!\left(D^M(\cdot,\theta)\right)\,\log\!\left(D^M(\cdot,\theta)\right) : \theta \in K,\; u \ge 0 \right\},
$$
which is also a Glivenko--Cantelli class. The desired result then follows by the same arguments as in the trimmed classification likelihood case.
\subsection{Auxiliary results}
    \begin{lemma}\label{lema_h}
The following results hold:  
\begin{itemize}
    \item[(i)] If $P$ satisfies condition~\eqref{PR}, then there exists a constant $h>0$ such that inequality~\eqref{bound_h} holds.
    \item[(ii)] If $P$ is an absolutely continuous distribution satisfying condition~\eqref{PR}, then there exists a constant $h'>0$ such that~\eqref{boun_h'} holds with probability $1$ for $n$ large enough.
\end{itemize}
\end{lemma}
\begin{proof}
The proof of both parts relies on a contradiction argument based on the regularity condition \eqref{PR}, which prevents the distribution $P$ from being concentrated on $G$ subspaces of dimensions $q_1, \ldots, q_G$ after trimming a fraction $\alpha$ of the probability mass.

\textsl{(i)} We reason by contradiction. Let us assume that no such $h > 0$ exists. Then, there exists a sequence of parameters $\{\theta_{n}\}_{n=1}^{\infty} \subset \Theta_{c_1, c_2}$ such that
\begin{equation}\label{lim_inf_z}
\lim_{n\rightarrow\infty} \mathbb{E}_{P} \left[ \sum_{g=1}^{G} z_{g}(\cdot; \theta_{n}) \| \text{Pr}_{\mathcal{B}_{g}^{n}}(\cdot) - \cdot \|^{2} \right] = 0,
\end{equation}
where $z_{g}(\cdot, \theta_{n})$ are the optimal assignments for the objective function of the tHDDC algorithm for the $\theta_n$ parameters.

Given an affine subspace $\mathcal{B}$ and $r>0$, let $S(\mathcal{B},r)=\{x\in \mathbb{R}^p: \|x-\text{Pr}_{\mathcal{B}}(x)\| \leq r\}$ denote the ``strip'' of radius $r$ centered at the affine subspace $\mathcal{B}$. Taking into account the arguments in \cite{GarciaEscudero_et_al_2009_RobustLinearClustering} for the RLG, we may consider radii $r_n$ defined as
$$r_n =\inf \{ r \geq 0 : P[\cup_{g=1}^G S(\mathcal{B}_g^n,r)] \geq 1-\alpha\},
$$
such that
\begin{equation*}
 \mathbb{E}_{P} \left[
 \sum_{g=1}^{G} z_{g}(\cdot; \theta_{n})
 \| \text{Pr}_{\mathcal{B}_{g}^n}(\cdot) - \cdot \|^{2} \right]
 \geq \mathbb{E}_{P}
 \left[ I\left\{\cup_{g=1}^G S(\mathcal{B}_g^n,r_n)\right\}(\cdot) \min_{g=1,...,G}\| \text{Pr}_{\mathcal{B}_{g}^n}(\cdot) - \cdot \|^{2} 
 \right],
\end{equation*}
because the $z_g$ are 0-1 functions (as indicator functions) such that $P[x\in \mathbb{R}^p:\sum_{g=1}^G z_g(x; \theta_{n})=1] \geq 1-\alpha$. Consequently, (\ref{lim_inf_z}) would then imply
\begin{equation}\label{contradiction_w}
\lim_{n\rightarrow\infty} \mathbb{E}_{P} \left[  I\left\{\cup_{g=1}^G S(\mathcal{B}_g^n,r_n)\right\}(\cdot) \min_{g=1,...,G}\| \text{Pr}_{\mathcal{B}_{g}^n}(\cdot) - \cdot \|^{2} \right] = 0.
\end{equation}

Let us also consider
$b_{g}^n = \operatorname*{arg\,min}_{x \in \mathcal{B}_{g}^n} \|x\|$, so that the affine subspace $\mathcal{B}_{g}^{n}$ can be written as the translation of the linear subspace generated by $q_g$ orthonormal vectors, passing through $b_{g}^n$.
Like the orthogonal bases belong to a compact group we can extract a convergent subsequence for each index $g$. So working with the corresponding subsequences (denoted again with the same notation $\{\mathcal{B}_{g}^n\}_{n=1}^{\infty}$), the potential lack of convergence of subsequences of the affine subspaces arises only if the sequence of $\|b_{g}^n\|_{n=1}^{\infty}$ diverges to $\infty$.

Taking into account that $P$ is a tight probability, we may consider a sequence $\lbrace \frac{\alpha}{2n}\rbrace_{k=1}^\infty$ and an increasing sequence $\{ M_n\}_{n=1}^\infty$ such that
\begin{equation*}
P(B(0,M_n))>1-\frac{\alpha}{2n}\text{ for every }n.
\end{equation*}
Let us now consider
$J = \{ g \in \{1,\ldots,G\} : \mathcal{B}_{g}^{n} \to \mathcal{B}_{g}^{0} \}$, as
the set of indices $g$ such that the sequence of affine subspaces $\mathcal{B}_{g}^{n}$ converges to limiting affine subspaces $\mathcal{B}_{g}^{0}$. Let us first prove that $J$ is non-empty, i.e., that it is not possible that all the $\{b_{g}^n\}_{n=1}^{\infty}$ sequences diverge. In that case, given that the sets $S_n=\cup_{g=1}^G S(\mathcal{B}_g^n,r_n)$ have probability equal or greater than $1-\alpha$, we would have that $P\big[S_n \cap B(0, M_1)\big] > 1- 3\alpha/2$, and  that for all points $x\in S_n \cap B(0, M_1)$ it holds that
\begin{equation*}
\|x- \text{Pr}_{\mathcal{B}_{g}^n}(x)\|\ge \min_{ 1\leq g\leq G} \|b_g^n\|-M_1 \xrightarrow{n \to \infty} \infty,
\end{equation*}
yielding a contradiction with  \eqref{contradiction_w}. 
Additionally, for the indices $g \notin J$, we can extract further subsequences such that $\|b_g^n\|\geq 2M_n$, for every $n$. We will use this subsequence to show that the sequence of radii $\{ r_{n}\}_{n=1}^\infty$ can not diverge. Suppose, to the contrary, that $r_n \rightarrow \infty$ and let any $g_{0} \in J$, with $b_{g_0}^n \to b_{g_0}^0$ for $b_{g_0}^0\in \mathbb{R}^p$, and take any $R>\|b_{g_0}^0\|+M_1$. For any $x \in B(0, M_1)$, 
we have $\|x- \text{Pr}_{\mathcal{B}_{g_0}^n}(x)\|\leq \|b_{g_0}^n\|+M_1\rightarrow \|b_{g_0}^0\|+M_1 < R$, and thus $B(0,M_1) \subset S(\mathcal{B}_{g_0}^n,R)$, for $n$ large enough, such that $P[S(\mathcal{B}_{g_0}^n,R)]>1-\alpha$ and, consequently, by the definition of $r_n$ we should have $r_n < R$, so contradicting the unboundedness of $\{r_n\}_{n=1}^{\infty}$.

Finally, we will show that divergent subspaces cannot attract mass in the limit, in the sense that $P[S(\mathcal{B}_{g_0}^n,r_n)] \to 0$, as $n \to \infty$, for any $g \notin J$. In that case, due to the divergence of $\|b_g^n\|$ for $g \notin J$, we can take a subsequence satisfying $\|b_g^n\| \geq 2M_n +R$, for some $R<\infty$ satisfying $r_n \leq R$ uniformly. Therefore, if $x$ fulfill that $\|x- \text{Pr}_{\mathcal{B}_{g}^n}(x)\|\leq r_n$ then, the triangular inequality, would tell us that $x\notin B(0,M_n)$ for a further appropriate subsequence, and, consequently, $P[S(\mathcal{B}_{g}^n,r_n)] \leq   P[B(0,M_n)^{\text{c}}] \leq \alpha/2n \to 0$.

To conclude the proof, note that we have just seen that  $\mathbb{E}_{P} \big[\cup_{g\in J}^G S(\mathcal{B}_g^n,r_n)\big] \rightarrow 1 - \alpha$, as $n\to\infty$. Consequently, the fact that there exist subsequences such that $\mathcal{B}_g^n \to \mathcal{B}_g^0$, for some $g \in J$ indices (with $J\neq\emptyset$), and $r_n \to r_0$ for some $0\leq r_0 < \infty$, implies that the convergence in \eqref{contradiction_w} can only hold if the probability mass of $P$ within the set
$\cup_{g\in J}^{G} S(\mathcal{B}_g^0,r_0)$,
whose probability mass is at least $1-\alpha$, is in fact concentrated on the affine subspaces $\mathcal{B}_g^0$, for $g\in J$ (note that $\| \text{Pr}_{\mathcal{B}_{g}^0}(x) - x \|^{2}>0$ when $x\notin \cup_{g\in J}\mathcal{B}_g^0$). This finally implies that the distribution $P$ is concentrated on $G$ subspaces of dimensions
$q_1,\ldots,q_G$.

\textsl{(ii)} To prove the empirical version of (i), we also proceed by contradiction, showing that it is not possible that
$$
\liminf_{n\to\infty}\mathbb{E}_{P_n}\left[\sum_{g=1}^{G}z_g(\cdot;\theta_n)\|\text{Pr}_{\mathcal{B}_g^n}(\cdot)-\cdot\|^2\right]=0,
$$ $P$-almost surely.
Analogously to the previous case, we can replace the optimal assignments in the objective function of tHDDC by the assignment functions associated with the (empirical) RLG, which are based on $G$ ``strips'' around the empirical affine subspaces $\mathcal{B}_g^n$ with the same empirical radius $r_n$. The proof technically follows analogous steps to those shown in the proof of part (i), but relies on the tightness of the sequence of empirical measures $\{P_n\}_{n=1}^{\infty}$. 
Take also into account that the family of functions representing the truncated orthogonal distances to a finite set of affine subspaces forms a Glivenko-Cantelli class, so that uniform convergence transfers this empirical limit back to the population measure.
\end{proof}

%Analogously to the previous case, we can change the optimal assignments for the objective function of the tHDDC by the assignment functions $w_g^k$ described in the same way as in part (i). Since the sequence of empirical measures $\{P_n\}_{n=1}^{\infty}$ converges weakly to $P$ almost surely, it is uniformly tight. This allows us to find a fixed compact ball $B(0,M)$ such that $P_n(B(0,M))>1-\alpha/2$ for all sufficiently large $n$. Consequently, the lower-bounding and compactness arguments from part (i) apply directly to the empirical setting: the empirical trimming radii $r_n$ are bounded almost surely, and we can identify a subset of indices $J\subseteq\{1,\dots,G\}$ such that the non-divergent empirical subspaces $\{B_g^n\}_{g\in J}$ converge to a set of limit subspaces $\{B_g^0\}_{g\in J}$, while divergent subspaces asymptotically attract zero mass. Finally, because the family of functions representing the truncated orthogonal distances to a finite set of affine subspaces forms a Glivenko-Cantelli class, uniform convergence transfers this empirical limit back to the population measure. This implies that $P$ must concentrate at least a mass of $1-\alpha$ on the union $\bigcup_{g\in J}B_g^0$, which directly contradicts the regularity condition and proves the existence of a lower bound $h'>0$ with probability 1 for $n$ sufficiently large.

Other technical results necessary for the proof of consistency are stated below. The complete proofs can be found in \cite{garcia2008general}.

\begin{lemma}
    Given a compact set $K$, the class of functions 
    \begin{equation}\label{class_H}
        \mathcal{H}=\lbrace I_{\left[u,\infty\right)}\left(D^H(\cdot;\theta)\right)\sum_{g=1}^Gz^\ast_g\left(\cdot;\theta\right)\log(D_g(\cdot;\theta))\quad : \quad \theta\in K,\:u\geq0\rbrace
    \end{equation}
    is a Glivenko-Cantelli class, with $z_g^\ast(x;\theta)=I\lbrace x\::\:D^H(x;\theta)=D_g(x;\theta)\rbrace$, i.e. all the points in $\mathbb{R}^p$ are assigned to some class through the functions $z_g^\ast$.
\end{lemma}

\begin{lemma}\label{converg_emp}
    Let $P$ be an absolutely continuous distribution with a strictly positive density function. Then, for every compact subset $K$, we have that 
    \begin{equation}
        \sup_{\theta\in K}\vert R(\theta,P_n)-R(\theta,P)\vert\rightarrow 0,\quad P\textup{-a.e.},
    \end{equation}  
where $R(\theta,P)$ is as defined in Section \ref{sec_32}.
\end{lemma}

\bibliography{references}
\bibliographystyle{apalike}

\end{document}